\documentclass[11pt,letterpaper]{article}

\usepackage{geometry}
\usepackage[american]{babel}
\usepackage[normalem]{ulem}
\usepackage{amsmath, amssymb, cases, amsthm}
\usepackage{thmtools}
\usepackage[shortlabels]{enumitem}
\usepackage{mdframed}
\usepackage{bbm}
\usepackage{bm}
\usepackage{microtype}
\usepackage{xcolor}
\usepackage{makecell}
\usepackage{mathtools}
\usepackage{algorithmic}
\usepackage[procnumbered,ruled,vlined,linesnumbered]{algorithm2e}
\usepackage{float}
\usepackage{varwidth}
\usepackage{modletter}
\usepackage{comment}
\usepackage[T1]{fontenc}
\usepackage{slantsc}
\usepackage{lmodern}

\global\long\def\Ohat{\widehat{O}}
\global\long\def\Otil{\tilde{O}}
\global\long\def\Omegahat{\widehat{\Omega}}

\SetKwInput{KwData}{Input}
\SetKwInput{KwResult}{Output}
\SetKwInput{KwGlobalVar}{Global variables}

\usepackage{xcolor}
\definecolor{ForestGreen}{rgb}{0.1333,0.5451,0.1333}
\definecolor{DarkRed}{rgb}{0.80,0,0}
\definecolor{Red}{rgb}{1,0,0}
\usepackage[linktocpage=true,
pagebackref=true,colorlinks,
linkcolor=DarkRed,citecolor=ForestGreen,
bookmarks,bookmarksopen,bookmarksnumbered]
{hyperref}
\usepackage[capitalize,noabbrev]{cleveref}

\declaretheorem[numberwithin=section,refname={Theorem,Theorems},Refname={Theorem,Theorems}]{theorem}

\declaretheorem[numberlike=theorem]{lemma}

\declaretheorem[numberlike=theorem]{proposition}
\declaretheorem[numberlike=theorem]{corollary}

\declaretheorem[numberlike=theorem,style=definition]{definition}

\declaretheorem[numberlike=theorem,style=definition]{assumption}
\declaretheorem[numberlike=theorem]{claim}
\declaretheorem[numberlike=theorem,style=remark]{remark}

\declaretheorem[numberlike=theorem, refname={Question,Questions},Refname={Question,Questions},name={Question}]{question}

\theoremstyle{definition}

\newcommand{\NoG}{N^{out}_G}

\newcommand{\NiG}{N^{in}_G}

\newcommand{\cGabow}{c_{{\ref{lem:explicit ramanujan}}}}
\newcommand{\cTUZ}{c_{{\ref{thm:disperser}}}}
\newcommand{\eps}{\epsilon}
\newcommand{\poly}{\operatorname{poly}}

\newcommand{\textdeg}{\operatorname{deg}}
\newcommand{\textbig}{\operatorname{big}}
\newcommand{\textsmall}{\operatorname{small}}
\newcommand{\calU}{\mathcal{U}}

\newcommand{\polylog}{\operatorname{polylog}}

\newcommand{\cAd}{\cA_{dist}}
\newcommand{\bF}{\mathbb{F}}

\newcommand{\tka}{\tilde{\kappa}}

\DeclarePairedDelimiter\ceil{\lceil}{\rceil}
\DeclarePairedDelimiter\floor{\lfloor}{\rfloor}

\AtBeginDocument{
  \DeclareFontShape{T1}{lmr}{m}{scit}{<->ssub*lmr/m/scsl}{}
}

\begin{document}
\sloppy

\title{
 Deterministic Vertex Connectivity\\ via Common-Neighborhood Clustering and Pseudorandomness
}

\author{
   Yonggang Jiang\thanks{
       MPI-INF and Saarland University, \texttt{yjiang@mpi-inf.mpg.de}
   } \and Chaitanya Nalam\thanks{
       University of Michigan,
       \texttt{nalamsai@umich.edu}
   }  \and Thatchaphol Saranurak\thanks{
       University of Michigan,
       \texttt{thsa@umich.edu}.
       Supported by NSF Grant CCF-2238138.
        Part of this work was done while at INSAIT, Sofia University ``St. Kliment Ohridski'', Bulgaria. This work was partially funded from the Ministry of Education and Science of Bulgaria (support for INSAIT, part of the Bulgarian National Roadmap for Research Infrastructure).
   }  \and Sorrachai Yingchareonthawornchai\thanks{
       ETH Zürich and the Hebrew University,
       \texttt{sorrachai.yingchareonthawornchai@eth-its.ethz.ch}.  Supported by the European Research Council (ERC) under the European Union’s Horizon 2020 research and innovation programme (grant agreement no. 759557–ALGOCom), and the ERC Starting Grant (CODY 101039914), Dr. Max Rössler, the Walter Haefner Foundation, and the ETH Zürich Foundation. Part of this work was done while at Simon Institute for the Theory of Computing, UC Berkeley.
   }}
\date{}
\maketitle

\pagenumbering{gobble}

\begin{abstract}
We give a deterministic algorithm for computing a global minimum vertex cut in a \emph{vertex-weighted} graph $n$ vertices and $m$ edges in $\Ohat(mn)$ time.\footnote{We use $\tilde{O}(\cdot)$ and $\hO{\cdot}$ to hide $\poly\log(n)$ and $n^{o(1)}$ factors, respectively.} This breaks the long-standing $\Omegahat(n^{4})$-time barrier in dense graphs, achievable by trivially computing all-pairs maximum flows.
Up to subpolynomial factors,
we match the fastest randomized $\Otil(mn)$-time algorithm by \cite{HenzingerRG00}, and affirmatively answer the question by \cite{Gabow06} whether deterministic $O(mn)$-time algorithms exist even for unweighted graphs.
Our algorithm works in directed graphs, too.

In unweighted undirected graphs, we present a faster deterministic $\Ohat(m\kappa)$-time algorithm where $\kappa\le n$ is the size of the global minimum vertex cut. For a moderate value of $\kappa$, this strictly improves upon all previous deterministic algorithms in unweighted graphs with running time $\Ohat(m(n+\kappa^{2}))$ \cite{Even75}, $\Ohat(m(n+\kappa\sqrt{n}))$ \cite{Gabow06}, and $\Ohat(m2^{O(\kappa^{2})})$ \cite{SaranurakY22}. Recently, a linear-time algorithm has been shown by \cite{Korhonen25} for small $\kappa$

Our approach applies the \emph{common-neighborhood clustering}, recently introduced by \cite{BJMY2024}, in novel ways, e.g., on top of weighted graphs and on top of vertex-expander decomposition. We also exploit pseudorandom objects often used in computational complexity communities, including \emph{crossing families} based on dispersers from \cite{WZ99,TUZ01} and \emph{selectors} based on linear lossless condensers \cite{guruswami2009unbalanced,Cheraghchi11}. To our knowledge, this is the first application of selectors in graph algorithms.
\end{abstract}

\clearpage
\tableofcontents

\clearpage

\pagenumbering{arabic}

\section{Introduction}

In the \emph{(global) minimum vertex cut} problem, given an $n$-vertex, $m$-edge graph $G=(V,E)$ with vertex weights $w:V\rightarrow\{1,\dots,W\}$, we must find a vertex set $S\subseteq V$ with minimum total weight $w(S):=\sum_{u\in S}w(u)$ whose removal disconnects $G$. The weight $w(S)$ is called the \emph{weighted vertex connectivity} of $G$. In the unweighted case, we must find $S$ with minimum size $|S|$ whose removal disconnects $G$, and the size $|S|$ is called the \emph{unweighted vertex connectivity} of $G$.

In this paper, we study fast algorithms for computing minimum vertex cuts, a topic that algorithm designers have extensively studied for over half a century \cite{Kleitman1969methods}. Before we survey its development below, observe that we can straightforwardly find a minimum vertex cut by solving $(s,t)$-maximum flow between all pairs of vertices. This algorithm requires $\binom{n}{2}$ maxflow calls, which takes $\Ohat(n^{4})$ time by \cite{brand2023deterministic}. We will refer to this time bound as a \emph{baseline.}\footnote{We use this bound solely for comparison purposes and do not apply it unfairly, Specifically, we do not claim an algorithm is slower than the baseline, while it can actually surpass the baseline using faster maximum flow algorithms.}

\paragraph{Early Development Within the Baseline.}

Since the 60s, there has been an active line of work on fast unweighted vertex connectivity \cite{Kleitman1969methods,Podderyugin1973algorithm,EvenT75,Even75,Galil80,EsfahanianH84,Matula87}. All of these algorithms are deterministic. The fastest among them \cite{Even75,Galil80} make $O(n+\kappa^{2})$ maxflow calls, where $\kappa$ denotes the vertex connectivity. However, when $\kappa=\Omega(n)$, they still require $\Omega(n^{2})$ maxflow calls, which is as slow as the straightforward algorithm.

\paragraph{Beating the Baseline via Randomization.}

In the 80s, Becker~et~al.~\cite{BeckerDDHKKMNRW82} were the first to surpass the baseline for unweighted vertex connectivity. Their randomized algorithm makes $\Otil(n)$ maxflow calls, which takes $\Otil(mn^{1.5})=\tilde{O}(n^{3.5})$ time using the available flow algorithm. Then, Linial, Lovász, and Wigderson~\cite{LinialLW88} discovered a connection between vertex connectivity and rubber band networks, leading to a randomized algorithm with $\tilde{O}(n^{\omega}+n\kappa^{\omega})=O(n^{3.373})$ time.\footnote{$\omega$ is the matrix multiplication exponent.}

Later in the 90s, Henzinger, Rao, and Gabow~\cite{HenzingerRG00} surpassed the baseline in the \emph{weighted} case. Their randomized algorithm in weighted graphs takes $\tilde{O}(mn) = \Otil(n^3)$ time, improving upon \cite{LinialLW88}. Also, they showed a \emph{deterministic} algorithm in weighted graphs with $O(m^{2})$ time. This bound, however, still does not beat the baseline on dense graphs.

\paragraph{Beating the Baseline Deterministically: Unweighted Case.}

The first and only deterministic algorithm that surpassed the baseline was discovered by Gabow \cite{Gabow06}. He introduced a novel use of Ramanujan expanders and developed an algorithm for unweighted graphs that makes $O(n+\kappa\sqrt{n})$ maxflow calls, which takes $\Ohat(m(n+\kappa\sqrt{n}))=\Ohat(mn^{1.5})$ time \cite{brand2023deterministic}\@.\footnote{The time bound stated in \cite{Gabow06} was $O(\min\{n^{0.75},\kappa^{1.5}\}\kappa m+mn)=O(n^{3.75})$ because he utilized the flow algorithms available at that time. In \Cref{sec:gabow alg}, we give a new analysis of his running time in terms of maxflow calls.}

Gabow also posed an open problem: whether there exists a $O(mn)$-time \emph{deterministic} algorithm for unweighted graphs. There has been no progress toward answering this question until now.

\paragraph{Near-Optimal Time via Randomization.}

Around 2020, there was again significant progress in randomized algorithms. For unweighted graphs, a $\tilde{O}(m\kappa^{2})$-time algorithm was developed \cite{NanongkaiSY19,ForsterNYSY20} using novel local cut algorithms. Later, through a new kernelization technique, Li~et~al.~\cite{li_vertex_2021} showed how to compute unweighted vertex connectivity using polylogarithmic maxflow calls, leading to an almost-optimal $\Ohat(m)$ time. Additionally, \cite{chekuri2021isolating} presented a randomized algorithm for $(1+\epsilon)$-approximate weighted vertex connectivity using polylogarithmic maxflow calls via the minimum isolating cut technique \cite{LiP20deterministic,AbboudKT21}.

\paragraph{The Current State.}

So far, progress on deterministic algorithms has been much slower, and to advance any further, we need to answer one of the long-standing open problems listed below:

\begin{enumerate}

\item Can we beat Gabow's $\Ohat(mn^{1.5})$-time algorithm for unweighted vertex connectivity?

\item Can we even beat the $O(n^{4})$-time baseline for weighted vertex connectivity?
\end{enumerate}
In contrast, for randomized algorithms, there is a $\Ohat(m)$-time unweighted algorithm \cite{li_vertex_2021} and a $\tilde{O}(mn)$-time weighted algorithm \cite{HenzingerRG00}.
Can we narrow this randomized-deterministic gap?

\paragraph{Randomized vs.~Deterministic Gaps.}

In the area of graph algorithms, closing the gaps between deterministic and randomized algorithms has been a fruitful research program. In the closely related problem of minimum \emph{edge} cut, significant efforts \cite{KawarabayashiT15,HenzingerRW17,Saranurak21,LiP20deterministic,Li21mincut,HLRW2024} to match Karger’s $\tilde{O}(m)$-time algorithm \cite{Karger00} with deterministic algorithms have generated several influential techniques, including the minimum isolating cut \cite{LiP20,AbboudKT21,abboud2022breaking}, new applications of local flow techniques \cite{KawarabayashiT15,HenzingerRW17,SaranurakW19,ForsterNYSY20,chalermsook_vertex_2020}, strong decomposition \cite{KawarabayashiT15,HLRW2024}, and expander decomposition \cite{Saranurak21,Li21mincut}.
Other related examples include deterministic algorithms for expander decomposition and their numerous applications \cite{ChuzhoyGLNPS20,abs-2106-01567,goranci2021expander} and deterministic min-cost flow algorithms~\cite{BrandCKLPGSS23,chen2024almost,brand2024almost}.

\paragraph{Directed Case.}

The minimum vertex cut problem can be generalized to \emph{directed} graphs in both weighted and unweighted cases. The goal here is to find the smallest vertex set $S$ such that $G\setminus S$ is not \emph{strongly} connected. Since Gabow's algorithm and the straightforward algorithm work on directed graphs with the same guarantee, the two questions above also capture the state of the art for deterministic algorithms in directed graphs.\footnote{The state of the art for randomized algorithms on directed graphs is slightly different. The almost-linear time unweighted algorithm by \cite{li_vertex_2021} also works in directed graphs but has a slower running time of $\Ohat(n^{2})$ instead of $\Ohat(m)$. The $(1+\epsilon)$-approximation weighted $\Ohat(m)$-time algorithm by \cite{chekuri2021isolating} does not work in directed graphs; however, \cite{cen2022minimum} works in $\Ohat(n^{2})$ time.}

\subsection{Our Results}

We give a strong affirmative answer to both questions simultaneously, even for directed graphs.

\begin{theorem}

\label{thm:main weighted} There is a deterministic algorithm that, given a \emph{directed} graph with $n$ vertices, $m$ edges, and \emph{vertex weights} in $\{1,\dots,W\}$, outputs a minimum vertex cut in $\Ohat(mn\log^{4}W)$ time.

\end{theorem}

Thus, up to subpolynomial factors, \Cref{thm:main weighted} resolved the open problem posed by Gabow whether a $O(mn)$-time deterministic algorithm exists for even unweighted graphs. Compared with the state of the art in weighted graphs, it matches the fastest randomized $\tilde{O}(mn)$-time algorithm \cite{HenzingerRG00} and improves upon the deterministic $O(m^{2})$-time algorithm \cite{HenzingerRG00}.

Furthermore, we present an even faster algorithm for undirected unweighted graphs.

\begin{theorem}

\label{thm:main k max flows} There is a deterministic algorithm that, given an \emph{undirected} \emph{unweighted} graph with $n$ vertices, $m$ edges, and vertex connectivity $\kappa\leq n$, outputs a minimum vertex cut in $\Ohat(m\kappa)$ time.

\end{theorem}

This algorithm is strictly faster than \emph{all} known deterministic algorithms for unweighted vertex connectivity for moderate $\kappa$, including the $\Ohat(m2^{O(\kappa^{2})})$-time algorithm by Saranurak and Yingchareonthawornchai \cite{SaranurakY22}, Even's algorithm \cite{Even75}, which requires $O(n+\kappa^{2})$ max-flow calls, and Gabow's algorithm \cite{Gabow06}, which requires $O(n+\kappa\sqrt{n})$ max-flow calls. Recently, Korhonen \cite{Korhonen25} has shown an $O(m\cdot \kappa^{O(\kappa^2)})$-time determinsitic algorithm, which is linear when $\kappa$ is small.

\paragraph{Deterministic Barrier to $o(mn)$ Time.}

When $\kappa = \Omega(n)$,
\Cref{thm:main weighted,thm:main k max flows} take $\Omega(mn)$ time. This suggest a natural question: Can we strictly beat the $O(mn)$ bound deterministically? Below, we highlight that the bottleneck is essentially captured by the hardness of this question.\footnote{The only other bottleneck is a technicality related to deterministic sparse recovery, which seems solvable.}

Let $G=(V,E)$ be an undirected unweighted graph. We say that $(L,S,R)$ is a vertex cut if there is no edge between $L$ and $R$. Our question is how to solve the problem even if all minimum cuts are ``perfectly balanced.''

\begin{question}
\label{que:balance}
Suppose that all minimum vertex cuts $(L,S,R)$ of a graph $G$ are such that $|L|,|S|,|R|=\Omega(n)$. Compute the vertex connectivity of $G$ in $o(mn)$ time.

\end{question}

With randomization, this question is easily solvable in $\Ohat(m)$ time by random sampling.\footnote{Sample $O(\log n)$ pairs of vertices. With high probability, one of the pairs $(a,b)$ is such that $a\in L$ and $b\in R$ for some minimum vertex cut $(L,S,R)$. Assuming this, the $(a,b)$-maximum flow must reveal a minimum vertex cut. } Indeed, the $\Ohat(m)$-time algorithm \cite{li_vertex_2021} highly exploits this. It is striking that none of the deterministic techniques can yield an $o(mn)$-time algorithm; we leave this as a challenging open problem.

\paragraph{Historical Note.}
The preliminary version of this paper~\cite{abs-2308-04695} only solves undirected unweighted deterministic vertex connectivity in $\Ohat(m\kappa^{2})$ time.

\subsection{Our Tools}

Next, we highlight the four technical tools that we developed and may be useful for future applications. The first two are employed for both \Cref{thm:main weighted,thm:main k max flows}, while the last two are used only in \Cref{thm:main k max flows}. To introduce the techniques in an intuitive way, we focus on unweighted undirected graphs in this section.

\subsubsection{Crossing Family.}

At a high level, a \emph{crossing family} $\cP$ is a family of pairs of elements such that, for every two sufficiently large sets $L$ and $R$, there must exist a ``crossing'' pair $(x,y)\in{\cal P}$ where $x\in L$ and $y\in R$.

Below, we provide a formal definition of our crossing family, where $L$ and $R$ can belong to two different universes $A$ and $B$, respectively.

\begin{definition}[Asymmetric crossing family] \label{def:ablr-crossingfamily} Given two sets $A,B$ and two integers $\ell,r$ where $\ell\leq|A|$, $r\leq|B|$, and $\ell\leq r$, an $(A,B,\ell,r)$-crossing family $\mathcal{P}$ is a subset of $A\times B$ such that, for any $L\subseteq A$ and $R\subseteq B$ with sizes at least $\ell$ and $r$, respectively, we have $\mathcal{P}\cap(L\times R)\neq\emptyset$. The degree of an element $x\in A$ in ${\cal P}$ is $|\{y\in B\mid(x,y)\in{\cal P}\}|$.
\end{definition}

Our algorithm in weighted graphs relies crucially on the asymmetric crossing family. We show a construction with the following guarantee.

\begin{restatable}{theorem}{asymcrossingfamilythm} \label{thm:crossing-family-unified} There exists a deterministic algorithm that, given two sets $A,B$ and two integers $\ell,r$ with $\ell\leq|A|$, $r\leq|B|$, and $\ell\leq r$, outputs an $(A,B,\ell,r)$-crossing family in $\Otil(|A|\cdot\max\{1,\frac{|B|-r}{\ell}\})$ time such that the degree of every element in $A$ is at most $\Otil(\max\{1,\frac{|B|-r}{\ell}\})$.
\end{restatable}

This family immediately implies the symmetric version of a crossing family that appears more applicable for finding a minimum vertex cut.

\begin{restatable}[Symmetric crossing family]{corollary}{symcrossingfamily} \label{thm:crossing family} There exists a deterministic algorithm that, given an integer $n$ and a parameter $\alpha\ge1$, outputs $\cP$ of size $\tO{\alpha n}$, a subset of $[n]\times[n]$ in time $\tO{\alpha n}$, such that for any partition $(L,S,R)$ of $[n]$ that satisfies $|R|\geq|L|\geq|S|/\alpha$, we have $\cP\cap(L\times R)\neq\emptyset$. Moreover, for each $x\in[n]$, the out-degree of $x$ in $\cP$, i.e., $|\{y\mid(x,y)\in{\cal P}\}|$, is at most $\tO{\alpha}$.
\end{restatable}

We will also refer to the above $\cP$ as a $(n,\alpha)$-crossing family.

\paragraph{Use Case.}
 To see why \Cref{thm:crossing family} is useful, consider the  graph from \Cref{que:balance}. There must be a pair from the $(n=|V|,\alpha=\Omega(1))$-crossing family ${\cal P}$ from \Cref{thm:crossing family} that crosses a minimum vertex cut $(L,S,R)$. Therefore, if we solve the $(x,y)$-maximum flow for each $(x,y)\in{\cal P}$, we will find the minimum vertex cut using $\tilde{O}(n)$ max flow calls instead of trivially using $\binom{n}{2}$ max flow calls.

\paragraph{Our Construction: Dispersers.}
Both \Cref{thm:crossing-family-unified,thm:crossing family} are proved in \Cref{sec:pseudo random objects} using \emph{dispersers} \cite{WZ99,TUZ01}, which are well-known objects for constructing randomness extractors. To our knowledge, this is the first direct application of dispersers in fast graph algorithms.

\paragraph{Previous Construction: Ramanujan Expanders.}

Gabow \cite{Gabow06} implicitly showed that the edges of Ramanujan expanders can be used to construct a crossing family in the setting of \Cref{thm:crossing family} that consists of $\tilde{O}(n^{3}/(|L||R|))$ pairs (see \Cref{lem:expander mixing}).

Our crossing family from \Cref{thm:crossing family} subsumes this bound and strictly improves upon it, particularly when $|L|\approx|S|\ll n$. In this situation, the family size from \Cref{thm:crossing family} would be $\tilde{O}(n)$, while Gabow's family size would be $\tilde{O}(n^{2}/|L|)$. This improvement is crucial for our purposes.

\subsubsection{Common Neighborhood Clustering}

Crossing families are useful when there exists a minimum vertex cut $(L,S,R)$ where $|L|=\Omega(|S|)$. To address the opposite case where $|L|=o(|S|)$, we exploit the concept of \emph{common neighborhood clustering}, introduced recently in \cite{BJMY2024}. Below, $N_{G}(u)$ denotes the neighbors of $u$ in $G$.

\begin{lemma}[Common-Neighborhood Clustering]\label{lem:comm clustering} For any graph $G=(V,E)$ and parameter $\ell$, there exists a collection of \emph{clusters} $V_{1},\ldots, V_{z}\subseteq V$ with the following properties. For any minimum vertex cut $(L,S,R)$ where $|L|\le\ell$,
\begin{enumerate}
\item \label{enu:cnc capture} $L\subseteq V_{i}$ for some cluster $V_{i}$,
\item \label{enu:cnc common} for every cluster $V_{i}$ and for every $u,v\in V_{i}$, the symmetric difference
\[|N_{G}(u)\triangle N_{G}(v)|=O(\ell\log n), \text{ and}\]
\item \label{enu:cnc sparse} every vertex belongs to at most $O(\log n)$ clusters.
\end{enumerate}
\end{lemma}

\paragraph{Use Case.}
\Cref{lem:comm clustering} is surprisingly useful in many ways. The intuition is as follows.
First, we know that the minimum vertex cut $(L,S,R)$ is ``captured'' by some cluster $V_i$ because, by \Cref{enu:cnc capture}, both $L$ and $S$ are contained in $N_G[V_i] := N_G(V_i)\cup V_i$ for some $V_i$.
Second, the structure of inside each cluster $V_i$ can greatly speed up the flow computation.  Roughly speaking, suppose we want to compute a $(s,t)$-minimum vertex cut where $s,t \in V_i$, by using the existing idea from \cite{li_vertex_2021,BJMY2024}, we are able to delete all edges\footnote{We are omitting details here. The correct operation is to first make each undirected as two bi-direction directed edges, then delete every directed edge from $u$ to $N(s)$, then make every remaining directed edge as undirected.} from a node $u\in V_i$ to a node in $N(s)$, without changing the $(s,t)$-minimum vertex cut. But after deletion, by \Cref{enu:cnc common}, there are only $O(\ell \log n)$ edges left adjacent to every vertex except $s,t$. Thus, this clustering is very effective in the case when $\ell$ is small, e.g., when $|L| \ll |S|$.

Lastly, even if we try to work on each cluster separately, the total work is small because the total number of edges of the graph induced by $N_G[V_i]$ overall clusters $V_i$ is at most $\Otil(m)$ by \Cref{enu:cnc sparse}.

\paragraph{Previous Usage: Randomized, Unweighted, Undirected.}

Common-neighborhood clustering in \cite{BJMY2024}
was used in randomized reductions from global to $(s,t)$-minimum vertex cuts in parallel and distributed models. However, their algorithms  are randomized and work only for unweighted undirected graphs.

\paragraph{Our Usage: Deterministic, Weighted, Directed.}

To prove \Cref{thm:main weighted}, we extend the reach of this technique much further into the deterministic setting, and also to directed and weighted graphs. See \Cref{sec:weight} for details on how we combine both crossing families and common-neighborhood clustering to work in weighted directed graphs.

\subsubsection{Terminal Reduction}

\label{subsec:intro:term reduction}

The next tool is specific to unweighted undirected graphs. It handles the case when $|L|=\Omega(|S|)$, where $(L,S,R)$ is a minimum vertex cut. Below, we say that a terminal set $T\subseteq V$ is \emph{$\alpha$-balanced} if there is a minimum vertex cut $(L,S,R)$ such that $|T\cap L|\geq\alpha|S|$, and $|T\cap R|\geq\alpha|S|$.

\begin{restatable}[Terminal Reduction (Simplified version of \Cref{lem:terminalreduction})]{lemma}{terminalreductionintro}
There is a deterministic algorithm that, given an $n$-vertex $m$-edge undirected graph $G=(V,E)$ along with a $2^{\sqrt{\log n}}$-balanced terminal set $T\subseteq V$, returns a new terminal set $T'\subseteq V$ of size at most $0.9|T|$ and a separator $S^{*}$ such that either
\begin{itemize}
    \item $T'$ is $0.1$-balanced, or
    \item $|S^{*}|=\kappa_{G}$.
\end{itemize}
The algorithm runs in $\Ohat(m\kappa)$ time.\label{thm:terminal reduction intro}
\end{restatable}

\paragraph{Use Case.}

Let us see why \Cref{thm:terminal reduction intro} is useful. Initially, set $T \gets V$. If $T$ is not $0.1$-balanced, then $|L| \le 0.1|S|$, which corresponds to the ``unbalanced'' case that can be handled using common neighborhood clustering. Next, as long as $T$ is $2^{\sqrt{\log n}}$-balanced, a call to \Cref{thm:terminal reduction intro} will either reduce the size of $T$ or yield a minimum vertex cut $S^{*}$. After at most $O(\log n)$ rounds, we will either obtain a minimum vertex cut or a terminal set $T$ that is not $2^{\sqrt{\log n}}$-balanced but is $0.1$-balanced. This is precisely when our last tool, called a \emph{selector}, becomes useful, as we will explain last.

\paragraph{Our Proof: New Use of Common-Neighborhood Clustering.}

The proof of \Cref{thm:terminal reduction intro} turns out to be the most technical part of the paper. Surprisingly, the key ingredient is again common-neighborhood clustering, but we apply it \emph{on top of vertex-expander decompositions}. This gives a completely novel use of the technique. See \Cref{sec:overview terminal sparsification} for details.

\paragraph{Previous Terminal Reductions.}

Theorem 3.3 in \cite{SaranurakY22} gives a $\Ohat(m2^{O(\kappa^{2})})$-time algorithm for performing \textit{vertex reduction}. It serves the same purpose as \Cref{thm:terminal reduction intro}, but instead reduces the total graph size rather than the number of terminal vertices. However, the $2^{O(\kappa^{2})}$-factor seems very challenging to improve to $\mathrm{poly}(\kappa)$, as doing so would imply a breakthrough for vertex sparsifiers.

The similar terminal reduction for minimum \emph{edge} cuts \cite{LiP20deterministic} (Section 4.2) is much simpler to obtain. Ultimately, the underlying reason is that vertex mincuts can separate a graph into many connected components, while edge mincuts always separate a graph into two components.

\subsubsection{Selectors}

The final tool for our $\Ohat(m\kappa)$-time algorithm for unweighted undirected graphs is a pseudorandom object called \emph{selectors.} Intuitively, a selector is a family $\cS$ of subsets that, given any two disjoint sets $L$ and $S$ where $|L|\ge\eps|S|$, there must exist a set $U\in\cS$ that ``hits $L$ once'' and ``misses $S$.''

\begin{restatable}[$(n,k,\epsilon)$-selector]{definition}{selectordefn} For every positive integers $n,k$ with $k<n$ and parameters $\eps\in[0,1]$, an \emph{$(n,k,\epsilon)$-selector} ${\cal {S}}$ is a family of subsets of $[n]$ such that, for every two disjoint subsets $L$ and $S$ of $[n]$ with sizes satisfying $\epsilon k<|L|\le k$ and $|S|\le k$, there exists a set $U\in{\cal S}$ such that $|U\cap L|=1$ and $|U\cap S|=0$.
\label{def:selector}
\end{restatable}

We prove the following theorem about selectors in~\cref{sec:selector construction}.

\begin{restatable}{theorem}{selectoreasy} \label{thm:selectoreasy} For any integer $n$ and parameter $\epsilon\ge1/n^{o(1)}$, there exists $\tilde{k}\ge n^{1-o(1)}$ depending on $\epsilon$ such that, given any $k\le\tilde{k}$, we can compute a $(n,k,\eps)$-selector $\cS$ of size at most $|\cS|=\Ohat(k)$ in $\hO n$ time, and every set $U\in\cS$ has size at least $2$.
\end{restatable}

\paragraph{Use Case.}

Following the discussion below \Cref{thm:terminal reduction intro}, the remaining case to handle is when the terminal set $T_{0}$ is not $2^{\sqrt{\log n}}$-balanced but is $0.1$-balanced. That is, there is a minimum vertex cut $(L,S,R)$ where $0.1|S|\le|T_{0}\cap L|\le2^{\sqrt{\log n}}|S|=\Ohat(\kappa)$.

Construct a $(n=|T_{0}|,k=|T_{0}\cap L|,\eps=0.1/2^{\sqrt{\log n}})$-selector ${\cal F}$ on $T_{0}$. By \Cref{thm:selectoreasy}, $|\mathcal{F}|\le \hO{|T_{0}\cap L|}$, and for some $U\in\mathcal{F}$, a minimum vertex cut $(L,S,R)$ must \textit{isolate} $U$, meaning that $|L\cap U|=1$, $|S\cap U|=0$, and $|R\cap U|>0$ because $|U|\ge 2$. Thus, for each $U\in\mathcal{F}$, we compute minimum isolating vertex cuts of $U$ in $O(\log|U|)$ maxflow calls by \cite{li_vertex_2021}.\footnote{The minimum isolating cuts algorithm of \cite{li_vertex_2021} only works on undirected graphs.} This must reveal the minimum vertex cut. The total running time is $\Ohat(m) \cdot |\mathcal{F}| \leq \hO{m\kappa}$.

\paragraph{Our Construction: Linear Lossless Condensers.}

It is crucial that our selector has size $|\cS|=\Ohat(k)$ so that our final running time is $\hO{m\kappa}$. Selectors of this size exist in the literature \cite{CGR02,Indyk02,CK05}, but they do not explicitly guarantee that each set $U\in\cS$ has size $|U|\ge 2$.\footnote{\Cref{def:selector} generalizes the selector defined in \cite{CGR02,Indyk02}, which takes $\eps=1/2$. A different notation for the $(n,k,r)$-selector defined in \cite{CK05} implies our $(n,k,\eps)$-selector above when $r=(1-\eps)k$.}

To address this issue, we construct the selectors ourselves. To ensure that $|U| \ge 2$, we use a lossless condenser from \cite{guruswami2009unbalanced} that is also \emph{linear}, as constructed by \cite{Cheraghchi11}. Given the right object, the construction is straightforward.

\paragraph{Related Pseudorandom Objects.}

Hit-and-miss families from \cite{karthik2021deterministic} or splitters (e.g., \cite{LiP20deterministic} and many results in the FPT literature) have been used in other graph algorithms and can be used to construct selectors.\footnote{At a high level, a hit-and-miss family is the same as a selector but does not require that $|L| \ge \eps|S|$. A splitter is even stronger; it further guarantees that, for every $x\in L$, there exists a set $U$ such that $U\cap L=\{x\}$ and $U\cap S=\emptyset$.}
Unfortunately, both objects are too strong and would imply a selector of size $\Omega(k^{2})$ \cite{karthik2021deterministic}, resulting in a final running time of $\Omega(m\kappa^{2})$, which is too slow.

\section{Preliminaries}\label{sec:prelims}

We use $\polylog(f)$ to denote $\log^c(f)$ for some constant $c$. We use $\tO{f}$ to hide poly-logarithmic factors and $\hO{f}$ to hide sub-polynomial factors, i.e., $\tO{f}=f\cdot \polylog(f)$ for some constant $c$ and $\hO{f}=f\cdot f^{o(1)}$.

For a set $A$, we use $2^A$ to denote the power set of $A$, i.e., each element in $2^A$ is a subset of $A$. For two sets $A,B$, we use $A\triangle B=(A-B)\cup(B-A)$ to denote the symmetric difference between $A$ and $B$.

\paragraph{Directed weighted graph.} A directed weighted graphs is denoted by $G=(V,E,w)$ where $V$ is the vertex set, $E\subseteq V\times V$ is the edges set and $w:V\to\mathbb{N}^+$ assigns a positive integer weight to each vertex. We always use $n=|V|,m=|E|$ to denote the number of nodes and edges. We use $W$ to denote the maximum weighted, i.e., $W=\max_{v\in V}w(v)$. We use $G^R$ to denote the reversed graph $(V,E^R=\{(v,u)\mid (u,v)\in E\},w)$.

For $u\in V$, we use $N^{out}_G(u)=\{v\in V\mid (u,v)\in E\}$ to denote the out-neighborhood set of $u$. For simplicity, we also use $N_G(u)$ to denote $N^{out}_G(u)$. We write $N_G[u]=N_G(u)\cup\{u\}$. For a vertex set $A\subseteq V$, we use $G[A]=(A,\{(u,v)\in E\mid u,v\in A\})$ to denote the induced subgraph on $A$. We write $w(A)=\sum_{v\in A}w(v)$. We use $N_G(A)=\left(\cup_{v\in A}N_G(u)\right)-A$ to denote the neighborhood set of $A$. We write $N_G[A]=N_G(A)\cup A$. For two vertex sets $A,B\subseteq V$, we use $\overrightarrow{E_G}(A,B)$ to denote the set of edges from $A$ to $B$, and $\overleftarrow{E_G}(A,B)$ to denote the edges from $B$ to $A$. When the context is clear, $G$ is omitted in the subscript.

\paragraph*{Undirected unweighted graphs.} Let $G=(V,E)$ be an undirected unweighted graph. Let $A, B \subseteq V$. We use $E(A,B)$ to denote the set of edges $(u,v)$ where $u\in A, v\in B$. We use $N_G(u)$ to denote the neighborhood set of $u$ in $G$, and $N_G[u]=N_G(u)\cup \{u\}$. We use $N_G(A)=(\cup_{u\in A}N(u))-A$ to denote the neighborhoods of $A$ and $N_G[A]=N_G(A)-A$. The minimum degree $\delta_G = \min_{u \in V}|N_G(u)|$.

\paragraph{Vertex cut.} A \emph{vertex cut} is a tri-partition of the graph $(L,S,R)$\footnote{Which means $L\cup S\cup R=V$ and $L,S,R$ are mutually disjoint.} with the guarantee that there are no edges from $L$ to $R$. We call $S$ a \emph{separator}. The weight of a vertex cut $(L,S,R)$ is $w(S)$, and the size of the vertex cut is $|S|$. A vertex cut or separator is \emph{minimum} if its weight is minimum (or the size is minimum if the graph is unweighted).

\paragraph{$(s,t)$-vertex cut.} A vertex cut $(L,S,R)$ is a $(s,t)$-vertex cut for two vertices $s,t\in V$ if $s\in L,t\in R$. We call $S$ as a $(s,t)$-separator. We use $\kappa_G(s,t)$ to denote the weight of the minimum weight $(s,t)$-vertex cut (or the size if the graph is unweighted). When the context is clear, we omit $G$ in the subscript.
By folklore reduction, we can find a minimum $(s,t)$-separator by one call to max flow.

\begin{lemma}[\cite{brand2023deterministic}]\label{lem:stseparator}
    There is a deterministic algorithm finding a minimum $(s,t)$-separator in $\hO{m}\cdot \log W$ time.
\end{lemma}

\section{An \texorpdfstring{$\hO{mn}$}{hatO(mn)}-Time Algorithm for Weighted Directed Graphs}
\label{sec:weight}

Suppose $(L,S,R)$ is one of the minimum vertex cuts. We will assume $w(L)\le w(R)$ by working on both $G$ and the reversed graph $G^{R}$.
Let $\lambda$ be a sufficiently large constant throughout this section. In~\cref{subsec:weightedunbalanced,subsec:weightedbalanced}, we will show how to handle the case when $w(R)>\lambda w(L)\log n$ and $w(R)\le \lambda w(L)\log n$, separately.

\begin{lemma}[Lopsided Case]\label{lem:weightedunbalanced}
    There is a deterministic algorithm that given a directed weighted graph $G=(V,E,w)$, outputs a vertex cut in $\hO{mn}\cdot \log^5 W$ time. Furthermore, if one of the minimum vertex cut $(L,S,R)$ in $G$ satisfies $w(R)>\lambda w(L)\log n$, then the output is a minimum vertex cut of $G$.
\end{lemma}

\begin{lemma}[Symmetric Case]\label{lem:weightedbalanced}
    There is a deterministic algorithm that given a directed weighted graph $G=(V,E,w)$, outputs a vertex cut $(L,S,R)$ in $\hO{mn}\cdot \log^4 W$ time. Furthermore, if one of the minimum vertex cut $(L,S,R)$ in $G$ satisfies $w(L)\le w(R)\le \lambda w(L)\log n$, then the output is a minimum vertex cut of $G$.
\end{lemma}

The above two lemmas immediately imply \Cref{thm:main weighted}. We give the overview for proving both lemmas next.

\subsection{Overview}\label{subsec:weightedoverview}

Fix a vertex mincut $(L,S,R)$ in the input weighted directed graph $G$. We can assume we know the weights $w(L),w(S)$, and $w(R)$ up to a factor of $2$ by guessing. Assume also that $w(L)\le w(R)$. Otherwise, the argument is symmetric.

We improve upon the straightforward algorithm that computes $(s,t)$-maximum flow overall $\binom{n}{2}$ all pairs $(s,t)$ using the following strategy:
\begin{enumerate}
\item Compute an $(s,t)$-maximum flow only for pairs $(s,t)$ from the certain crossing family ${\cal P}$, and
\item Before computing $(s,t)$-maximum flow, ``compress'' the graph to reduce the  \emph{amortized} size of the maxflow instances. More precisely, each edge appears in at most $\tilde{O}(n)$ maxflow instances.
\end{enumerate}
We emphasize that our goal is to bound the \emph{amortized} size of maxflow instances instead of the \emph{worst-case} size as done in prior work \cite{li_vertex_2021}. This is a key difference between our approach and the one of \cite{li_vertex_2021}.

\paragraph{Graph Compression.}

First, we show how to compress maxflow instances and obtain a $\Ohat(mn)$-time algorithm given the following crossing family $\cP$ such that
\begin{enumerate}
\item There exists $(s,t)\in\cP$ with $s\in L$ and $t\in R$, and
\item The \emph{degree of $u$ in ${\cal P}$} defined as $\deg_{{\cal P}}(u):=|\{v\mid(u,v)\in\cP\}|$) is at most $\tO{n\cdot w(u)/w(R)}$.
\end{enumerate}
For each $(s,t)\in\cP$, we will compress the graph without changing the $(s,t)$-vertex mincut by deleting the edges between vertices in $\NoG(s)$ as also done in \cite{li_vertex_2021}. This is safe because any $(s,t)$-vertex maxflow never need to use edges between vertices in $\NoG(s)$.

Now, we argue that, for each edge $(u,v)$, the total number of maxflow instances that $(u,v)$ is not deleted is $\tO n$. Consider the maxflow instance for $(s,t)\in{\cal P}$, if $(u,v)$ is not deleted, this means that either $u\notin\NoG(s)$ or $v\notin\NoG(s)$. That is, either $s\not\in\NiG(u)$ or $s\not\in\NiG(v)$. It remains to show that the number of pairs $(s,t)\in\cP$ where $s\not\in\NiG(u)$ or $s\not\in\NiG(v)$ is at most $\tO n.$

Indeed, according to the degree bound of $\cP$, the total number of such $s,t$ pair is at most

\[
\sum_{a\in(V-\NiG(u))\cup(V-\NiG(v))}\deg_{P}(v)=\tO{n\cdot\frac{\sum_{a\in(V-\NiG(u))\cup(V-\NiG(v))}w(a)}{w(R)}}.
\]
Note that $\NiG(u),\NiG(v)$ are vertex cuts and so $w(\NiG(u)),w(\NiG(v))\ge w(S)$. Hence,
\[
\sum_{a\in(V-\NiG(u))\cup(V-\NiG(v))}w(a)\le2\cdot(w(V)-w(S))=2\cdot(w(L)+w(R))\le4w(R).
\]
where the last inequality is by the assumption $w(L)\le w(R)$. Combining the two inequalities above, we conclude that the total number of $(s,t)\in\cP$ where $s\not\in\NiG(u)$ or $s\not\in\NiG(v)$ is
\[
\sum_{a\in(V-\NiG(u))\cup(V-\NiG(v))}\deg_{P}(v)=\tO{n\cdot\frac{4w(R)}{w(R)}}=\tO{n}.
\]

\paragraph{Crossing Family Construction.}

Now, we construct the desired crossing family $\cP$.

Suppose for a moment that we are in the unweighted case. If we apply \cref{thm:crossing family}, we would get the family ${\cal P}$ such that there exists $(s,t)\in\cP$ with $s\in L$, $t\in R$ and (ii) the degree of $u$ in $\cP$ is at most $\tO{\alpha}=\tO{|S|/|L|}$. By weight bucketing and using the fact that $|S|\le n$, a natural generalization of \cref{thm:crossing family} to the weighted case will give us the similar family ${\cal P}$ where the degree of $u$ in $\cP$ is at most $\tO{n\cdot w(u)/w(L)}$.

In the symmetric setting of \cref{lem:weightedbalanced}, i.e., $w(R)=\tT{w(L)}$, then the degree bound of $u$ in ${\cal P}$ becomes $\tO{n\cdot w(u)/w(R)}$, exactly the property that we want.

Thus, the hard case is in the lopsided setting of \cref{lem:weightedunbalanced} where possibly $w(L)\ll w(R)$. We sketch our strategy in this case below. The key step is that, without knowing $R$, we can find a vertex set $R'$ that ``approximates $R$ well'' in the sense that $w(R'\triangle R)=\tO{w(L)}$. This key step is explained in the next paragraph.

Given $R'$, we will use the asymmetric crossing family from \cref{thm:crossing-family-unified} to construct ${\cal P}$. Suppose for a moment again that we are in the unweighted case. So $|R'\triangle R|=\tO{|L|}$. Consider an $(A,B,\ell,r)$-crossing family where $A\gets V$, $B\gets R'$, $\ell\gets|L|$, and $r\gets|R\cap R'|$. Using \cref{thm:crossing-family-unified}, we would obtain ${\cal P}$ where there is $(s,t)\in\cP$ where $s\in L$ and $t\in R\cap R'$ and, for each $u\in V$, the degree of $u$ in ${\cal P}$ is at most $\Otil(\frac{|B|-r}{\ell})=\Otil(\frac{|R'|-|R\cap R'|}{\ell})=\Otil(\frac{|R' - R|}{\ell})=\Otil(1)$.
Note that it is valid to apply \cref{thm:crossing-family-unified}, which requires $\ell\le r$, because we are in the lopsided case where $R$ is much larger than $L$. Formally, the condition $|L|\le |R\cap R'|$ is valid  because $R$ is much larger than $L$ and $|R'\triangle R|=\tO{|L|}$, i.e., $R$ and $R'$ are basically the same (up to $\tO{|L|}$ elements).

By weight bucketing, a natural generalization of \cref{thm:crossing-family-unified} to the weighted case (formally shown in \cref{lem:unbalanced apply crossingfamily}) will give a degree bound of $\tO{n\cdot w(u)/w(R)}$, exactly what we want.
To see this, since $R' - R$ takes a fraction of $\tilde{O}(w(L)/w(R))$ weight in $R$,
there must exist a bucket such that $R-R'$ takes a fraction of $\tilde{O}(w(L)/w(R))$ vertices in $R'$. However, $R'$ can be as large as $n$, so the degree bound becomes $\tO{\frac{n\cdot(w(L)/w(R))}{w(L)}}=\tO{n/w(R)}$. The additional $w(u)$ factor comes from the fact that we also bucket the weights in $A\leftarrow V$, which is similar to duplicating a vertex $u$ for $w(u)$ times.

\paragraph{Approximate $R$ without Knowing It.}

To explain the idea of how to find $R'$ such that $w(R\triangle R')=\tO{w(L)}$ without knowing $R$. We will further assume here that the graph satisfies \emph{common-neighborhood} property, i.e., for every $u,v\in V$, we assume $w(\NoG(u)\triangle\NoG(v))\le\tO{w(L)}$. This assumption can be removed using the common-neighborhood clustering from \Cref{lem:comm clustering}, generalized to the weighted case in \cref{lem:weightedcommonneighborhood}.

First, observe that, for every node $u\in L$, since $w(\NoG(u))\ge w(S)$ and $\NoG(u)\subseteq L\cup S$, we have $w(\NoG(u)\triangle S)=O(w(L))$. With the common-neighborhood property, we further deduce that, for every $v\in V$ we must have $w(\NoG(v)\triangle S)=O(w(L))$. Intuitively, this means that $S$ receives a lot of edges from every node in $V$, which gives us enough information for $S$ so that we can approximately find $S'$ with $w(S'\triangle S)=\tO{w(L)}$.
Simply by setting $R'=V-S'$, we can get $w(R'\triangle R)=\tO{w(L)}$. This is formally explained in \cref{lem:largelyintersectingR}.

\subsection{Lopsided Case: Proof of \texorpdfstring{\Cref{lem:weightedunbalanced}}{weighted unbalanced}}\label{subsec:weightedunbalanced}

\paragraph{Step 1: common-neighborhood clustering.} We first compute a common-neighborhood clustering described in~\cref{lem:weightedcommonneighborhood} below. In~\cref{lem:weightedcommonneighborhood}, we should think of $\ell$ as an approximation of $w(L)$ where $(L,S,R)$ is a minimum lopsided vertex cut. $\ell$ can be guessed by the powers of $2$.

\begin{lemma}[Common-neighborhood Clustering]\label{lem:weightedcommonneighborhood}
    There is a deterministic algorithm that given a directed weighted graph $G=(V,E,w)$ and an integer $\ell$, outputs a set of clusters $\cC\subseteq 2^V$ in $\tO{mn}$ time such that
    \begin{enumerate}
        \item (sparse) for every $v\in V$, there are at most $O(\log n)$ sets in $\cC$ that contains $v$,
        \item (common-neighborhood) for every $C\in \cC$ and $u,v\in C$, $w(\NoG(u)\triangle N_G^{out}(v))= O(\ell\log n)$.
        \item (cover) for every minimum vertex cut $(L,S,R)$ with $w(L)\le \ell$, there is $C\in \cC$ where $L\subseteq C$.
    \end{enumerate}
\end{lemma}
\begin{proof}
    Define an undirected graph $G'=(V,E')$ where $E'=\{(u,v)\mid w(N_G^{out}(u)\triangle N_G^{out}(v))\le 2\ell\}$. $G'$ can be constructed in $O(mn)$ time by explicitly computing $N_G^{out}(u)\triangle N_G^{out}(v)$ for any $u,v$.

    We run \emph{sparse neighborhood cover}~\cite{AwerbuchBCP98} on $G'$ to get a set of clusters $\cC\subseteq 2^V$, which has running time nearly linear on the size of $G'$, which is $\tO{n^2}=\tO{mn}$. According to~\cite{AwerbuchBCP98}, $\cC$ has the following properties.
    \begin{enumerate}
        \item \textbf{For every $v\in V$, there are at most $O(\log n)$ sets in $\cC$ that contains $v$.} This gives the \emph{(sparse)} property as in~\cref{lem:weightedcommonneighborhood}.
        \item \textbf{The diameter of each $C\subseteq \cC$ is at most $O(\log n)$.} This means there is a path with length $O(\log n)$ from $u$ to $v$ in $C$ for any $u,v\in C\in\cC$. According to the definition of $G$, if there is an edge $(a,b)\in E'$, then we can add and delete vertices with weighted at most $2\ell$ from $N_G^{out}(a)$ to get $N_G^{out}(b)$. Thus, we can add and delete vertices with weight at most $2\ell\cdot O(\log n)$ from $N_G^{out}(u)$ to get $N_G^{out}(v)$ since there is a path with length $O(\log n)$ connecting them, which means $w(N_G^{out}(u)\triangle N_G^{out}(v))=O(\ell\log n)$, giving \emph{(common-neighborhood)}.
        \item \textbf{For every $v\in V$, there exists $C\in\cC$ with $N_{G'}[v]\subseteq C$.} We will prove that, for every two vertices $u,v \in L$, we have $w(N_G^{out}(u)\triangle N_G^{out}(v))\le 2\ell$. This implies that  $L\subseteq N_{G'}[u]$. Since we know that $N_{G'}[u] \subseteq C$ for some $C\in\cC$, this implies \emph{(cover)}. Notice that this proof appears in many previous works \cite{li_vertex_2021,BJMY2024}, and we prove it here for completeness.

        To prove $w(N_G^{out}(u)\triangle N_G^{out}(v))\le 2\ell$, notice that $w(N_G^{out}(u)\triangle N_G^{out}(v))=2w(N_G^{out}(u)\cup N_G^{out}(v))-w(N_G^{out}(u))-w(N_G^{out}(v))$. We have $w(N_G^{out}(u)\cup N_G^{out}(v))\le w(S)+w(L)$ since $N_G^{out}(u),N_G^{out}(v)\subseteq L\cup S$. We also have $w(N_G^{out}(u))\ge w(S)$ since $(L,S,R)$ is a minimum vertex cut of $G$ and $N_G^{out}(u)$ is a separator of $G$. Thus, we get $w(N_G^{out}(u)\triangle N_G^{out}(v))\le 2(w(L)+w(S))-w(S)-w(S)=2w(L)\le 2\ell$.
    \end{enumerate}

    The output of the sparse neighborhood cover algorithm on $G'$ gives the common-neighborhood clustering as desired.
\end{proof}

To better convey the idea, the rest of the algorithm will be described based on the following assumption.
\begin{assumption}\label{ass:LinC}
    $(L,S,R)$ is a minimum vertex cut of $G$ which satisfies $w(R)>\lambda w(L)\log n$, and we are given a cluster $C\in\cC$ with $L\subseteq C$ (where $\cC$ is constructed by~\cref{lem:weightedcommonneighborhood}).
\end{assumption}

\cref{ass:LinC} can be removed by running the rest of the algorithm for every $C\in\cC$.
Moreover, we assume that we know the approximated value of $w(L)$ and $w(R)$ denoted by $\ell$ and $r$. To remove this assumption, we can guess $\ell,r$ by the powers of $2$.

\paragraph{Step 2: identifying a set largely intersecting $R$.} In this step, we will construct a set that can be modified from $R$ by adding or deleting vertices with total weight at most $O(w(L)\log n)$, i.e., it intersects $R$ a lot.

Define $V_{high}=\{v\in V\mid w(N^{in}_G(v)\cap C)>0.9w(C)\}$, i.e., $V_{high}$ contains all the nodes which has in-neighbors almost covering $C$. Define $V_{low}=V-V_{high}$. We claim that $V_{low}$ is the set that largely intersects $R$.

\begin{lemma}\label{lem:largelyintersectingR}
    Assume~\cref{ass:LinC}, we have $w(V_{low}\triangle R)=O(w(L)\log n)$.
\end{lemma}
\begin{proof}
    We first show that it suffices to prove $w(V_{high}\triangle S)=O(w(L)\log n)$. Suppose we have $w(V_{high}\triangle S)=O(w(L)\log n)$. We show that the total weight changes due to swapping in and swapping out vertices from the set $V_{low}$ to get the set $R$ is $O(w(L)\log n)$. Since $V_{low}=V-V_{high}$, we can add or delete vertices with weight at most $O(w(L)\log n)$ from $V_{low}$ to get $V-S$. Notice that $V-S=L\cup R$, thus, by further deleting $L$ from $V-S$ (which has weight $w(L)$), we get $R$. Thus, $w(V_{low}\triangle R)=O(w(L)\log n)$ as desired.

    We prove that $w(S-V_{high})=O(w(L)\log n)$. Define the potential function $Q=\sum_{(u,v)\in E(C,S)}w(u)\cdot w(v)$.  One way to express $Q$ is by $Q=\sum_{v\in C}w(v)w(N^{out}(v)\cap S)$. Let $x \in L\subseteq C$ (remember that $L,S,R$ are defined in~\cref{ass:LinC}) be an arbitrary vertex in $L$, we have $w(N^{out}(x)\cap S)\ge w(S)-w(L)$ since $w(N^{out}(x))\ge w(S)$ (as $S$ is a minimum separator) and $N^{out}(x)\subseteq L\cup S$. According to (common-neighborhood) property of~\cref{lem:weightedcommonneighborhood}, for any $v\in C$, we have $w(N^{out}(v) \triangle N^{out}(x)) \leq O(w(L) \log n)$, which means \begin{equation}
w(N^{out}(v)\cap S)\ge w(S)-w(L)-O(w(L)\log n)\geq w(S)-O(w(L)\log n).\label{eq:out to S}
\end{equation}
Thus, we get
    \begin{align*}
        Q=&\sum_{v\in C}w(v)w(N^{out}(v)\cap S)\\
        \ge&\sum_{v\in C}w(v)\cdot\left(w(S)-O(w(L)\log n)\right)\\
        =&w(C)\cdot \left(w(S)-O(w(L)\log n)\right)
    \end{align*}
    On the other hand, $Q$ can be expressed by $Q=\sum_{v\in S}w(v)\cdot w(N^{in}(v)\cap C)$, which according to the definition of $V_{high}$ is
    \begin{align*}
        Q=&\sum_{v\in S}w(v)w(N^{in}(v)\cap C)\\
        \le&\left(\sum_{v\in S\cap V_{high}}w(v)\cdot(w(C))\right)+\left(\sum_{v\in S\cap V_{low}}w(v)\cdot(0.9w(C))\right)\\
        =&\left(w(S\cap V_{high})+0.9w(S\cap V_{low})\right)\cdot w(C)
    \end{align*}
    By combining the two inequalities and canceling the $w(C)$ factor, we get
    \begin{align*}
    w(S\cap V_{high})+0.9w(S\cap V_{low})&\ge w(S)-O(w(L)\log n)\\
    O(w(L)\log n)&\ge 0.1w(S\cap V_{low})
    \end{align*}
    i.e., $w(S-V_{high})=O(w(L)\log n)$.

    Next, we prove that $w(V_{high}-S)=O(w(L)\log n)$. Define $Q'=\sum_{(u,v)\in E(C,V_{high}-S)}w(u)w(v)$. One way to express $Q'$ is $Q'=\sum_{v\in C}w(v)w(N^{out}(v)\cap (V_{high}-S))$.
    For all $v \in C$, we claim that
    \[
        w(N^{out}(v)\cap (V_{high}-S)) \leq w(N^{out}(v)) - w(N^{out}(v) \cap S)  \le O(w(L)\log n).
    \]
    The last inequality follows because $w(N^{out}(v)\cap S)\ge w(S)-O(w(L)\log n)$ by \eqref{eq:out to S}. Also, we have $w(N^{out}(v)) \leq w(L \cup S) + O(w(L)\log n)$. This is because for any vertex $x \in L \subseteq C$, $w(N^{out}(x)) \leq w(L \cup S)$ and $w(N^{out}(v) \triangle N^{out}(x)) \leq O(w(L) \log n)$ by the (common-neighborhood) property of~\cref{lem:weightedcommonneighborhood}.
    Thus,
    \begin{align*}
        Q'=\sum_{v\in C}w(v)\cdot w(N^{out}(v)\cap (V_{high}-S))
        \leq w(C)\cdot O(w(L)\log n))
    \end{align*}
    On the other hand, $Q'$ can be expressed as
    \begin{align*}
        Q'=\sum_{v\in V_{high}-S}w(v)\cdot w(N^{in}(v)\cap C)
        \ge w(V_{high}-S)\cdot 0.9w(C)
    \end{align*}
    The last inequality follows from the definition of $V_{high}$.

    By combining the two inequalities, we get $w(V_{high}-S)=O(w(L)\log n)$.
\end{proof}

Note that $V_{low}$ can be constructed in $O(m)$ time by iterating all edges.

\paragraph{Step 3: constructing a crossing family.} In this step, we will construct a crossing family, which is a set of vertex pairs such that one of them will cross $L,R$. Formally, we will prove the following lemma.

\begin{lemma}\label{lem:unbalanced apply crossingfamily}
    Given $V_{low}$ which is constructed in step 2 satisfying~\cref{lem:largelyintersectingR}, there is a deterministic algorithm that constructs a family of pairs $P\subseteq C\times V_{low}$
    such that
    \begin{enumerate}
        \item If~\cref{ass:LinC} is true, then there exists $(u,v)\in P$ such that $u\in L,v\in R$,
        \item for every $u\in C$, define $\deg_P(u)=|\{v\in V_{low}\mid (u,v)\in P\}|$, we have $\deg_P(u)=\tO{1+n\cdot \frac{w(u)}{w(R)}}\cdot \log^3W$.
    \end{enumerate}
    The running time is $O(n|C|)$.
\end{lemma}
\begin{proof}
    \textit{(Algorithm):} Split the weights in $V$ into $q=\lceil\log W\rceil$ buckets, denoted the vertices in the $i$-th bucket as $V_i$, i.e., nodes in $V_i$ have weights between $2^{i-1}$ to $2^i$. Define $C_i=V_i\cap C,D_i=V_{low}\cap V_i$. For every $i,j\in[q]$, apply~\cref{thm:crossing-family-unified} on vertices sets $C_i,D_j$ with $l=w(L)/(2^i\log W),r=|D_j|-nw(L)\log W\log^2 n/w(R)$. Let the returned edge set be $P_{i,j}$. Add all $(u,v)\in P_{i,j}$ to $P$.

    \textit{(Correctness):} Let $L_i=V_i\cap L$. There must exists $i$ such that $w(L_i)\ge w(L)/\log W$. Remember that $L\subseteq C$, which means $|C_i\cap L|=|L_i|\ge w(L)/2^i\log W$. Remember that $w(V_{low}\triangle R)=O(w(L)\log n)$ by \cref{lem:largelyintersectingR} and $w(R)>\lambda w(L)\log n$ for sufficiently large constant $\lambda$ by \cref{ass:LinC}. Thus, we have $w(V_{low}\cap R)\ge w(R)/2$. There must exists $j$ such that $w(D_j\cap R)\ge w(R)/(2\log W)$. Notice that $|D_j\cap R|\ge w(D_j\cap R)/2^j\ge w(R)/(2^{j+1}\log W)$, implying $2^j=\Omega(w(R)/|D_j\cap R|\log W)=\Omega(w(R)/n\log W)$. Thus, we have $|D_j-R|\le w(D_j-R)/2^{j-1}=O(nw(L)\log W\log n/w(R))$. From~\cref{thm:crossing-family-unified}, $P_{i,j}$ must crosses $L,R$.

    Now, we show the bound for $\deg_P(u)$. According to~\cref{thm:crossing-family-unified}, the left degree for the set $P_{i,j}$ is bounded by
    \[\tO{1+\frac{nw(L)\log W\log^2 n/w(R)}{w(L)/2^i\log W}}=\tO{1+n\cdot \frac{2^i}{w(R)}\cdot \log^2W}\]

    Notice that $2^i$ is roughly the weight of every node in $C_i$. Another $\log W$ factor comes from the fact that each node is included in $P_{i,j}$ for $\log W$ times.

    \textit{(Running time):} According to~\cref{thm:crossing-family-unified}, the running time for constructing $P_{i,j}$ is bounded by at most $|C_i|\cdot |D_j|$, thus, the total running time is at most $\tO{n|C|}$.
\end{proof}

\paragraph{Step 4: running max-flows on sparsified graphs.} The idea is to find $(s,t)$-separator using~\cref{lem:stseparator} for every pair $(s,t)\in P$ where $P$ is constructed by~\cref{lem:unbalanced apply crossingfamily}, and the minimum among them is the minimum cut we want. However, remember that in~\cref{lem:unbalanced apply crossingfamily}, we only guarantee $\deg_P(u)=\tO{1+n\cdot w(u)/w(R)}$, which means when $w(R)$ is small, it can be as large as $n$. Thus, $P$ can be as large as $n^2$. Calling $n^2$ max flows may take $\Omega(n^2m)$ time, which is too slow.

To handle the issue specified above, we need to sparsify each max flow instance by deleting ``useless'' edges (i) the edges between nodes in $N(s)$, (ii) the edges between nodes outside $C$. We also need to add an edge from every vertex in $\NoG(C)$ to $t$ in order to make sure the cut we found in $G_{s,t,C}$ has the left-hand side in $C$. Formally, we define the sparsified graph $G_{s,t,C}$ as follows:
\begin{align*}
    E_t& = \bigcup_{ u\in N_G^{out}(C)}\{(u,t)\},\\
    G_{s,t,C}&=(C\cup N_G^{out}(C)\cup\{t\},(E_G(C)-E_G(\NoG(s),\NoG(s))\cup E_t).
\end{align*}

The following two lemmas show that $G_{s,t,C}$ preserves the minimum separator in $G$.

\begin{lemma}\label{lem:sparsification}\label{lemma:sparsify valid cut}
    Any $(s,t)$-separator in $G_{s,t,C}$ is a $(s,t)$-separator in $G$.
\end{lemma}
\begin{proof}
Suppose $(L',S',R')$ is a $(s,t)$-vertex cut in $G_{s,t,C}$. Since any vertex in $\NoG(C)$ has an edge to $t$, we have $L'\subseteq C$. Now we prove that $\NoG(L')\subseteq N^{out}_{G_{s,t,C}}(L')$, certifying that $S'$ is a $(s,t)$-separator in $G$. If $v\in \NoG(L')\cap \NoG(s)$, then $v\in N^{out}_{G_{s,t,C}}(L')$ since every edge from $s$ to neighbors of $s$ are preserved in $G_{s,t,C}$; if $v\in\NoG(L')-\NoG(s)$, there exists $u\in L'$ such that $(u,v)\in E$, since $v\not\in \NoG(s)$, $(u,v)$ is preserved in $G_{s,t,C}$. So, $v \in N^{out}_{G_{s,t,C}}(L')$.
\end{proof}

    The above lemma implies that the minimum $(s,t)$-separator in $G_{s,t,C}$ has weight at least the weight of the minimum $(s,t)$-separator in $G$.
\begin{lemma}\label{lemma:sparsify preserve mincut}
Let $(L,S,R)$ be the minimum vertex cut satisfying \cref{ass:LinC}. If  $s\in L,t\in R$, then the minimum $(s,t)$-separator in $G_{s,t,C}$ is a minimum $(s,t)$-separator in $G$.
\end{lemma}
\begin{proof}
    We have $N^{out}_{G_{s,t,C}}(L)\subseteq \NoG(L)$ because the only edges that exist in $G_{s,t,C}$ but do not exist in $G$ are the edges from $\NoG(C)$ to $t$, but we know $L \subseteq C$ by \Cref{ass:LinC}.
    Since $\NoG(L) \subseteq S$, we have $N^{out}_{G_{s,t,C}}(L) \subseteq S$. This implies that $S$ is a $(s,t)$-separator in $G_{s,t,C}$.

    Thus, the minimum $(s,t)$-separator in $G_{s,t,C}$ has weight at most the weight of the weight of the minimum $(s,t)$-separator in $G$. By \Cref{lemma:sparsify valid cut}, we know the the two weights must be equal.
\end{proof}

Next, we will prove that the cumulative number of edges in all max flow calls is roughly $mn$. To do this, instead of giving an upper bound on the number of edges in $G_{s,t,C}$ for every $s,t,C$ (which actually can be very large) as done in \cite{li_vertex_2021}, we count how many times each edge is in a max flow call. This is the crucial difference between our argument and the one in \cite{li_vertex_2021}.

\begin{lemma}\label{lem:sizeofmaxflow}
    The total number of edges in all $G_{s,t,C}$ among $(s,t)\in P,C\in\cC$ is at most $\tO{mn}\cdot\log^2W$, where $P$ is constructed in~\cref{lem:unbalanced apply crossingfamily} and $\cC$ is constructed in~\cref{lem:weightedcommonneighborhood}. That is,
    $$\sum_{(s,t) \in P, C \in \cC} |E(G_{s,t,C})| =  \tO{mn}\cdot\log^3W.$$
    Moreover, with $O(mn)$ time of preprocessing, each $G_{s,t,C}$ can be constructed in time proportional to the number of edges in $G_{s,t,C}$.
\end{lemma}
\begin{proof}
    Given an arbitrary edge $(u,v)\in E$, we prove that there are at most $\tO{n}$ different $G_{s,t,C}$'s that contain it. Notice that besides edges in $E$, there are edges in $E_t$ that are also in $G_{s,t,C}$. However, the number of these edges is $|\NoG(C)|$, which is charged to edges in $E$ intersecting $G_{s,t,C}$. Thus, proving there are at most $\tO{n}$ different $G_{s,t,C}$ that contains an arbitrary edge $(u,v)\in E$ suffices.

    For $(u,v)$ to be included in $G_{s,t,C}$, one of $u,v$ must be in $C$. Since $u,v$ can be included in at most $O(\log n)$ different $C$ according to~\cref{lem:weightedcommonneighborhood}, we just need to prove that for a fixed $C$ with $u\in C$ (or $v\in C$), $(u,v)$ is included in $\tO{n}$ different $G_{s,t,C}$. According to the construction of $G_{s,t,C}$, either $u\not\in \NoG(s)$ or $v\not\in\NoG(s)$. According to~\cref{lem:unbalanced apply crossingfamily}, the total number of pairs $(s,t)\in P$ satisfying $u\not\in\NoG(s)$ (which is equivalent to $s\not\in\NiG(u)$) can be calculated as
    \begin{align*}
        & \sum_{s\not\in \NiG(u)}\tO{1+n\cdot \frac{w(s)}{w(R)}}\cdot\log^3W\\
        = & \tO{n+n\cdot \frac{\sum_{s\not\in\NiG(u)}w(s)}{w(R)}}\cdot\log^3W\\
        = & \tO{n+n\cdot \frac{2w(R)}{w(R)}}\cdot\log^3W\\
        = & \tO{n}\cdot\log^3W
    \end{align*}

    The last equality is because $w(\NiG(u))\ge w(S)$ (because $\NiG(u)$ is a vertex separator), which implies $w(V-\NiG(u))\le w(L)+w(R)\le 2w(R)$. The same argument goes for the case when $v\not\in\NoG(s)$.

    Now we show fast construction of $G_{s,t,C}$. For the $O(mn)$ time pre-processing, we compute exactly $\NoG(u)-\NoG(v)$ and $\NiG(u)-\NoG(v)$ for any $u,v\in V$. Now for any $s,t,C$, in order to construct $G_{s,t,C}$, we need to find all edges for $G_{s,t,C}$. Notice that all edges for $G_{s,t,C}$ are either (1) edges adjacent to $s$, which can be constructed directly from $\NoG(s)$, (2) edges adjacent to vertices in $C$ minus the edges between $\NoG(s)$, this can be found by including all edges adjacent to vertices in $C-\NoG(s)$, and then for any $u\in C\cap \NoG(s)$, including all edges in $\NiG(u)-\NoG(s)$ and $\NoG(u)-\NoG(s)$, the latter is directly found from the pre-computed sets (it is too slow to compute them again as $\NiG(u),\NoG(u)$ could be much larger than $\NiG(u)-\NoG(s)$ and $\NoG(u)-\NoG(s)$), (3) edges adjacent to $t$, this is essentially determined by the neighbors of $t$ and $\NoG(C)$, where the latter is found in the previous step.
\end{proof}

In summary, the algorithm is as follows. Firstly, guess $w(L)$ and $w(R)$ by powers of $2$ (leading to $\tO{\log^2 W}$ many possibilities). For each guess, run the following algorithm to find a cut and the minimum cut among all cuts is the final output.
\begin{itemize}
    \item \textit{(Step 1)} construct common-neighborhood clustering using~\cref{lem:weightedcommonneighborhood}, which gives $\cC$ in $\tO{mn}$ time,
    \item for every $C\in\cC$ (notice that $|\cC|=\tO{n}$), do the following steps to find a vertex cut, the finial output will be the minimum among them,
    \begin{itemize}
        \item \textit{(Step 2)} find $V_{low}$ with respect to $C$ in $O(m)$ time,
        \item \textit{(Step 3)} construct a crossing family $P$ in $\tO{n|C|}$ time according to \Cref{lem:unbalanced apply crossingfamily}. Note that the total time is $\sum_{C \in \cC}\tO{n|C|} = \tO{n^2}$.
        \item \textit{(Step 4)} for every $(s,t)\in P$, construct $G_{s,t,C}$ and find the minimum $(s,t)$-separator in $G_{s,t,C}$, which must be a separator in $G$ according to~\cref{lem:sparsification}, and the total running time is $\hO{mn}\log^3W$ according to~\cref{lem:sizeofmaxflow}.
    \end{itemize}
\end{itemize}

Now we argue the correctness. Let $(L,S,R)$ be the minimum vertex cut in $G$ satisfying $w(R)>\lambda w(L)\log n$. According to~\cref{lem:weightedcommonneighborhood}, there exists $C\in \cC$ such that $L\subseteq C$. Thus, \cref{ass:LinC} is true for at least one $C\in \cC$, for which~\cref{lem:unbalanced apply crossingfamily} holds. So, there exists $(s,t)\in P$ with $s\in L,t\in R$. According to~\cref{lemma:sparsify preserve mincut}, it will return a valid vertex cut of $G$ with minimum weight $w(S)$. Moreover, any cut found by the algorithm must be a valid vertex cut according to~\cref{lem:sparsification}. Therefore, the output must be a valid vertex cut with minimum weight $w(S)$.

\subsection{Symmetric Case: Proof of \texorpdfstring{\Cref{lem:weightedbalanced}}{weighted balanced}}
\label{subsec:weightedbalanced}

In the case when $w(L)\le w(R)\le \lambda w(L)\log n$, we will first use~\cref{lem:balancedcrossingfamily} to construct a family of pairs such that one of them crosses $L$ and $R$. As in the previous section, we assume $\ell$ is an approximation of $w(L)$.

\begin{lemma}\label{lem:balancedcrossingfamily}
    When a minimum vertex cut $(L,S,R)$ satisfies $w(L)\le w(R)\le \lambda w(L)\log n$, there is a deterministic algorithm with running time $\tO{n^2\log^2W}$ that constructs a family of pairs $P\subseteq V\times V$ such that
    \begin{enumerate}
        \item there exists $(u,v)\in P$ such that $u\in L,v\in R$ or $v\in L,u\in R$,
        \item for every $u\in V$, define $\deg_P(u)=|\{v\in V\mid (u,v)\in P\}|$, we have $\deg_P(u)=\tO{1+n\cdot \frac{w(u)}{w(R)}}\cdot\log^2 W$.
    \end{enumerate}
\end{lemma}

\begin{proof}
    \textit{(Algorithm):} Split the weights into $q=\lceil\log W\rceil$ buckets, denoted the vertices in the $i$-th bucket as $V_i$, i.e., nodes in $V_i$ have weights between $2^{i-1}$ to $2^i$. For every $i,j\in[q]$, apply~\cref{thm:crossing family} on the vertex set $V_i\cup V_j$ (treat $|V_i\cup V_j|$ as $n$ in~\cref{thm:crossing family} and relates every integer in $[n]$ to a vertex in $V_i\cup V_j$) with $\alpha=\min(n,n\cdot(2^{\max(i,j)}/\ell)\cdot \log W\cdot \lambda^2\log n)$ (recall that $\ell$ is a 2-approximation of $w(L)$), let the returned $(|V_i\cup V_j|,\alpha)$ crossing family be $P_{i,j}$. If $i\ge j$, add all $(u,v)\in P_{i,j}$ with $u\in V_i$ to $P$; otherwise, add all $(u,v)\in P_{i,j}$ with $u\in V_j$ to $P$. The total running time is $\tO{n^2 \log^2 W}$ since constructing each crossing family cost $\tO{n^2}$ as $\alpha\le n$.

    \textit{(Correctness):} Let $L_i=V_i\cap L$ and $R_i=V_i\cap R$. There must exists $i$ such that $w(L_i)\ge w(L)/\log W$ and $j$ such that $w(R_j)\ge w(R)/\log W$. We will prove that $P_{i,j}$ crosses $L,R$. Notice that $|(V_i\cup V_j)\cap L|\ge |L_i|\ge w(L_i)/2^{i}$, and $|(V_i\cup V_j)\cap R|\ge |R_j|\ge w(R_j)/2^j$. Thus, we have
    $$|(V_i\cup V_j)\cap S|/|(V_i\cup V_j)\cap L|\le n\cdot (2^{i}/w(L))\cdot\log W\le \alpha$$
    and
    $$|(V_i\cup V_j)\cap S|/|(V_i\cup V_j)\cap R|\le n\cdot (2^{j}/w(R))\cdot\log W\le \alpha$$
    (the inequality holds when $\alpha<n$, but when $\alpha\ge n$, the crossing family is complete which must cross $L$ and $R$). This implies that there is $(u,v)\in P_{i,j}$ with $u\in L,v\in R$ and another $(u',v')\in P_{i,j}$ with $u'\in R,v'\in L$. In any case, the correctness for (1) is proved. Moreover, according to \cref{thm:crossing family}, we have the out-degree of $u$ in $P_{i,j}$ is at most $\tO{1+n\cdot\frac{2^{\max(i,j)}\log W}{\ell}}$. Remember that we only add $(u,v)\in P_{i,j}$ to $P$ if $u\in V_i$ with $i\ge j$, which implies $w(u)=\Omega(2^{\max(i,j)})$. Also remember that $\ell=\tT{w(R)}$. Notice that each $u$ can be included in at most $\log W$ many $P_{i,j}$, thus, we get $\deg_P(u)\le \tO{1+n\cdot\frac{w(u)}{w(R)}}\cdot \log ^2W$.
\end{proof}

Similar to~\cref{lem:sparsification}, the total number of edges in the sparsified max flow instances is $\tO{mn\log^2W}$.

\begin{lemma}\label{lem:correctnessbalanced}
    Let $P$ be the crossing family from~\cref{lem:balancedcrossingfamily}. For $s,t\in V$, define $G_{s,t}$ as the graph after deleting edges between $\NoG(s)$, and the edges between $\NiG(t)$. We have
    \begin{enumerate}
        \item the minimum $(s,t)$-vertex cut does not change in $G_{s,t}$ compared to $G$,
        \item the total number of edges of $G_{s,t}$ among all $(s,t)\in P$ and $(t,s)\in P$ is $\tO{mn}\cdot\log^2W$.
    \end{enumerate}
    Moreover, $G_{s,t}$ can be constructed in time proportional to the number of edges in $G_{s,t}$ with $O(mn)$ time preprocessing.
\end{lemma}
\begin{proof}
    The first point is easy to see: vertex cut in $G$ is clearly vertex cut in $G_{s,t}$, and suppose $(L,S,R)$ is a $(s,t)$-vertex cut in $G_{s,t}$, we have $\NoG(s)\subseteq L\cup S$ and $\NiG(t)\subseteq S\cup R$, which means edges between $\NoG(s)$ or between $\NiG(t)$ cannot cross $L,R$, implying that $(L,S,R)$ is a $(s,t)$-vertex cut in $G$.

    To see the second point, we will prove that for any edge $(u,v)\in E$, there are at most $\tO{n\log^2W}$ possible $(s,t)\in P$ or $(t,s)\in P$ such that $G_{s,t}$ includes $(u,v)$. There are two possibilities for $(u,v)$ to be included in $G_{s,t}$, the first one is that $(u,v)$ is an edge from $s$ to $\NoG(s)$ or from $\NiG(t)$ to $t$. This case can happen at most $2n$ times, since either $u$ is fixed to $s$ or $v$ is fixed to $t$, and $P$ is not a multiset. The second case is that one of $u,v$ is in $V-\NoG[s]$. Let us assume $u\in V-\NoG[s]$, which implies $s\in V-\NiG(u)$. The total number of such $s,t$ pair in $P$ is at most
    \[2\sum_{v\in V-\NiG(u)}\deg_P(v)=\tO{n+n\cdot\frac{\sum_{v\in V-\NiG(u)}w(v)}{w(R)}}\cdot\log^2W=\tO{n}\cdot\log^2W\]
    The last equality is because $\NiG(u)$ is a vertex cut, which implies $w(\NiG(u))\ge w(S)$ and $w(V-\NiG(u))\le w(L)+w(R)\le 2w(R)$.

    To see fast construction of $G_{s,t}$, notice that we can use $O(mn)$ time preprocessing to get the symmetric difference between neighborhood set of every two vertices, then to construct $G_{s,t}$ we simply need to find neighbors of vertices $V_{s,t}$ minus the neighbors of $s,t$.
\end{proof}

In summary, the algorithm first generate $\tO{\log W}$ guesses of $w(L)$, for each of them, use~\cref{lem:balancedcrossingfamily} to get a crossing family $P$ with a pair $(u,v)\in P$ satisfying either $u\in L,v\in R$ or $v\in L,u\in R$, so we need to construct $G_{u,v}$ for every $(u,v)\in P$ or $(v,u)$ and find the minimum $(s,t)$-separator, the total running time and correctness is guaranteed by~\cref{lem:correctnessbalanced}.

\section{\texorpdfstring{An $\widehat{O}(m \kappa)$}{Ohat(mk)}-Time Algorithm for Unweighted Undirected Graphs}
\label{sec:Algorithmunweighted}

The goal of this section is prove \Cref{thm:main k max flows}. The algorithm consists of three ingredients: (1) the unbalanced algorithm, (2) the terminal reduction algorithm, and (3) the terminal-balanced algorithm.

The guarantee of the algorithm for the unbalanced case is summarized as follows. One can view the minimum degree $\delta$ as an approximate size of a vertex mincut.

\begin{lemma}\label{lem:unbalanced}
    There is a deterministic algorithm that takes inputs an $n$-nodes $m$-edges undirected graph $G=(V,E)$, and outputs a vertex cut of $G$. If there exists a minimum vertex cut $(L,S,R)$ of $G$ satisfying $|L|\le \lambda\cdot \delta$ for a sufficiently large constant $\lambda$ and $|L| \leq |R|$, then the output is a minimum vertex cut of $G$. The algorithm runs in  $\widehat{O}(m\delta)$ time.
\end{lemma}

\newcommand{\balanced}[1]{#1-balanced}
\newcommand{\strongbalanced}[1]{#1-strongly-balanced}

For the balanced case, we first define balanced terminal sets.
\begin{definition}\label{def:balancedterminalset}
    For an undirected graph $G=(V,E)$, a terminal set $T\subseteq V$ is called \emph{\strongbalanced{$\alpha$}} if for every minimum vertex cut of $G$ denoted by $S$ and every connected component of $G[V-S]$ denoted by $L$, we have $|T\cap L|\ge \alpha|S|$.
\end{definition}

The terminal reduction algorithm is as follows.

\begin{lemma}\label{lem:terminalreduction}
    There is a deterministic algorithm \textsc{TerminalReduction}$(G,T,k)$ (\cref{alg:terminalreduction}) that takes as inputs an $n$-vertex $m$-edge undirected graph $G=(V,E)$, along with a \strongbalanced{$2^{\sqrt{\log n}/2}$} terminal set $T\subseteq V$ and a cut parameter $k\ge\kappa_G$ and $k=O(\kappa_G)$, outputs $(S',T')$ where
    \begin{itemize}
        \item $|T'|\le 0.9|T|$,
        \item either $S'$ is a minimum vertex cut of $G$ or $T'\subseteq V$ is a \strongbalanced{$0.01$} terminal set.
    \end{itemize}
    The algorithm runs in time $\widehat{O}(mk)$.
\end{lemma}

The final ingredients handle the case when the terminal set is balanced.

\begin{lemma}\label{lem:balancedterminalalgorithm}
    There is a deterministic algorithm \textsc{BalancedTerminalVC}$(G,T,k)$ that takes as inputs an undirected graph $G=(V,E)$, a terminal set $T\subseteq V$ and a cut parameter $k$, outputs a vertex cut of $G$. If there exists a vertex cut $(L,S,R)$ of $G$ such that $|T\cap L|,|T\cap R|=\widehat{\Omega}(|T\cap S|)$ and $|T\cap L|=\widehat{O}(k)$, then the output cut is a minimum vertex cut of $G$. The algorithm runs in $\widehat{O}(mk)$ time.
\end{lemma}

\begin{remark}
    The proof is based on the isolating vertex cut lemma~\cite{li_vertex_2021}, which requires the graph to be undirected.
\end{remark}

Via the three algorithmic components from \cref{lem:unbalanced,lem:terminalreduction,lem:balancedterminalalgorithm}, we are ready to prove \cref{thm:main k max flows}.

\paragraph{Proof of~\cref{thm:main k max flows}.}
\begin{proof}
    By the sparsification algorithm of \cite{NagamochiI92}, we can assume the input undirected graph $G$ has number of edges $m \le \kappa_G n$ by $O(m)$ time preprocessing. Thus, the minimum degree $\delta \le 2\kappa_G$.
    Notice that $\delta\ge\kappa_G$. Thus, by setting $k=\delta$ we have $k\ge\kappa_G$ and $k=O(\kappa_G)$.

    \paragraph{Algorithm.}  The algorithm first use \cref{lem:unbalanced} on $G$ to get a vertex cut $S^*$. Then the algorithm set $T_0=V$, and for $i=0,1,2,...$ run
    \[S_i\leftarrow\textsc{BalancedTerminalVC}(G,T_i,k)\] and
    \[(S'_i,T_{i+1})\leftarrow\textsc{TerminalReduction}(G,T_i,k)\] until $T_{i+1}$ becomes empty. The returned cut will be the minimum one among $S^*$ and $S_i,S'_i$ for all $i$.

    \paragraph{Correctness.} If there exists a minimum vertex cut $S$ and a minimum size connected component $L$ of $G[V-S]$ (that makes $|V-L-S|\ge |L|$) such that $|L| \le \lambda|S|=O(\delta)$ for a large constant $\lambda$, then according to \cref{lem:unbalanced}, $S^*$ is a minimum vertex cut. Otherwise, a terminal set $T_0 = V$ must be \strongbalanced{$0.01$}.

    Let $i$ be the minimum index such that $T_i$ is not a \strongbalanced{$2^{\sqrt{\log n}/2}$} terminal set. Such $i$ must exist because the size $|T_i|$ is strictly decreasing.
    Suppose that none of  $S^*,S_0,S'_0,\dots,S_{i-1},S'_{i-1}$ is a minimum vertex cut; otherwise, we are done. We will show that $S_i$ must be a minimum vertex cut.
    If $i=0$, $T_0=V$ is \strongbalanced{$0.01$} as observed above.
    If $i>0$, since $T_{i-1}$ is a \strongbalanced{$2^{\sqrt{\log n}/2}$} terminal set, according to~\cref{lem:terminalreduction}, $T_i$ must be a \strongbalanced{0.01} terminal set because we assume $S'_{i-1}$ is not a minimum vertex cut.
    In any case, $T_i$ is a \strongbalanced{$0.01$} terminal set and not a \strongbalanced{$2^{\sqrt{\log n}/2}$} terminal set.

    By the definition, there exists a vertex cut $S$ and a connected component of $G[V-S]$ denoted by $L$ such that $|T_i\cap L|<2^{\sqrt{\log n}/2}|S|$. We also have $|T_i\cap L|>\Omega(|S|)$ and $|V-S-L|>\Omega(|S|)$. By \cref{lem:balancedterminalalgorithm}, $S_i$ is a minimum vertex cut of $G$.

    \paragraph{Complexity.} The sparsification by \cite{NagamochiI92} takes $O(m)$ time.
    \Cref{lem:unbalanced} takes $\widehat{O}(m\delta) = \widehat{O}(m\kappa_G)$ time.
    According to~\cref{lem:terminalreduction}, $|T_{i+1}|\le 0.9|T_i|$, so there can be at most $O(\log n)$ loop of algorithm. Each loop takes time $\widehat{O}(m\kappa_G)$ according to~\cref{lem:terminalreduction,lem:balancedterminalalgorithm}.
\end{proof}

\subsection{Organization and Overview}

\paragraph{Organization.}
The rest of this section is devoted for proving \cref{lem:unbalanced,lem:terminalreduction,lem:balancedterminalalgorithm}. \cref{lem:balancedterminalalgorithm} serves as subroutines for \cref{lem:unbalanced,lem:terminalreduction} so we prove it first in~\cref{sec:base cases alg}. Another subroutine we are going to use is common-neighborhood clustering, which we define and show an algorithm in~\cref{sec:sparseneighborhoodcover}. After that, we prove \cref{lem:unbalanced} in~\cref{subsec:lemmadatastructure,sec:proof of unbal assume ds} and \cref{lem:terminalreduction} in~\cref{sec:terminal sparse alg}.

 We give the technical overview of \cref{lem:unbalanced} and \cref{lem:terminalreduction} in this section as the proofs  are quite involved.

\subsubsection{Overview of Unbalanced Case: \texorpdfstring{\Cref{lem:unbalanced}}{unbalanced case lemma}}

Let $(L,S,R)$ be a vertex mincut where $|L| \leq |R|$. In this section, we assume that $G[L]$ is connected, and $|L|, |S|$ are known (otherwise, we can approximate the sizes by a factor of 2 using binary search) and
\begin{align}  \label{eq:LltS}
|L| \leq O(|S|)
\end{align}
\paragraph{Fast Common-Neighborhood Clustering.}

The key technical tool is the fast common-neighborhood clustering (\Cref{lem:weightedcommonneighborhood}). We compute collection of clusters $V_1,\ldots V_z$ such that
\begin{enumerate} [nolistsep]
    \item $L \subseteq V_i$ for some cluster $V_i$,
    \item for all cluster $V_i$, and for every pair $u,v \in V_i$, $|N(u) \triangle N(v)| \leq O(|L| \log n)$, and
    \item every vertex belongs to $O(\log n)$ clusters.
\end{enumerate}

The construction in \Cref{lem:weightedcommonneighborhood} takes $\tilde O(mn)$ time because we explicitly compare the neighborhood set of every pair of vertices. To obtain $\tilde O(m \kappa)$ time, we define a different intermediate graph from the one in the proof of \Cref{lem:weightedcommonneighborhood}, open the box of sparse-neighborhood cover algorithm and apply \emph{deterministic sparse recovery sketching} from~\cite{NanongkaiS17}. We show the fast construction in \Cref{sec:sparseneighborhoodcover}.

We assume we obtain a collection of clusters $V_1,\ldots V_z \subseteq V$ satisfying the three properties stated above. In addition, the clusters $V_1,\ldots V_z$ have important structural properties.

\begin{lemma} \label{lem:additional property vc}
If $L \subseteq V_i$ for some $i$, then either $|V_i| \leq O(|S| \log n)$ or $|S \cap V_i| \leq O(|L| \log n)$.
\end{lemma}
\begin{proof}
We prove that if $|V_i|=\omega(|S|\log n)$, then $|S \cap V_i|=O(|L|\log n)$. Suppose to the contrary that $|S\cap V_i|=\omega(|L|\log n)$. We use double counting on the size of $E_G(V_i\cap S,V_i\cap R)$. we have
        \[|E_G(V_i\cap S,V_i\cap R)|\le |V_i\cap S|\cdot O(|L|\log n)\]
        since $|N_G(u)\triangle N_G(v)|=O(|L|\log n)$ according to the second property of the clusters for any $u\in V_i\cap S$ (and notice that any node in $V_i\cap R$ is not a neighbor of an arbitrary vertex $v\in L$). On the other hand, consider a node $s\in L$, since $|N_G(s)|\ge |S|$ and $N_G(s)\subseteq L\cup S$, we have $|S\backslash N_G(s)|\le |L|$. According to the second property of the clusters, for any $u\in R$, we have $|N_G(s)-N_G(u)|=O(|L|\log n)$. Thus, for any $u\in R$ we have $|(V_i\cap S)\cap N_G(u)|\ge |(V_i\cap S)\cap N_G(s)|-O(|L|\log n)\ge |V_i\cap S|-|L|-O(|L|\log n)\ge |V_i\cap S|/2$. The last inequality is due to $|S\cap V_i|=\omega(|L|\log n)$. Now we get
        \[|E_G(V_i\cap S,V_i\cap R)|\ge |V_i\cap R|\cdot|V_i\cap S|/2\]

        By combining the inequalities above, we get $|V_i\cap R|=O(|L|\log n)$. Thus, we have $$ |S| \geq |V_i\cap S|=|V_i|-|V_i\cap L|-|V_i\cap R|=\omega(|S|\log n).$$ This is a contradiction.
\end{proof}
Based on \Cref{lem:additional property vc} and the property of clusters, we have two cases:

\paragraph{Case 1:}  There is a cluster $V_i$ such that $L \subseteq V_i$ and $|V_i| = O(|S|\log n)$.  In this situation, we will use crossing families based on dispersers~\cite{TUZ01} with kernelization techniques~\cite{li_vertex_2021}.

We construct a crossing family $\mathcal{P}$ by applying \Cref{thm:crossing family} using $n = |V(G)|$ and $\alpha = \Theta (|S|/|L|)$ in $\tilde O(n \alpha) = \tilde O(n \cdot \frac{\kappa}{|L|})$ time. By \Cref{thm:crossing family}, there exists a pair $(s,t) \in \mathcal{P}$ such that $s \in L, t \in R$, and thus
\begin{align}
\kappa &= \min_{(s,t) \in \mathcal{P}} \kappa_G(s,t)   \mbox{ and } \label{eq:crossing set hit s}\\
|\mathcal{P}| &= \tilde O(n \alpha) = \tilde O(n\cdot \frac{\kappa}{|L|}) \label{eq:crossing set small}.
\end{align}   By max-flow mincut theorem, it suffices to compute $(s,t)$-vertex-capacity max-flow for each $(s,t)\in \mathcal{P}$.  Before running $(s,t)$-vertex-capacity max-flow,  we compute a \emph{reduced} instance (called \textit{kernel}) of $G$.

Our task is to compute the \textit{kernel} of $G$ with the following guarantee:
\begin{quote}
    Given $s, t$ in $G = (V,E)$, return a graph $H$ with $\tilde O(\kappa \cdot |L|)$ edges such that
    \begin{enumerate}
        \item if $s\in L,t\in R$, then $\kappa_G(s,t) = \kappa_H(s,t)$,
        \item for any $s,t$, we have $\kappa_G(s,t)\le \kappa_H(s,t)$.
    \end{enumerate}
\end{quote}
Given such kernels, we can solve Case 1 by computing $(s,t)$-max-flows on the kernel of $G$ for all $(s,t) \in \mathcal{P}$. The running time is
\begin{align} \label{eq:total kernel mk}
|\mathcal{P}| \cdot \Ohat (\kappa \cdot |L|) \overset{(\ref{eq:crossing set small})}{=} \tilde O(n \cdot \frac{\kappa}{|L|}) \cdot \Ohat(\kappa \cdot |L|) =  \Ohat (n \kappa^2) = \Ohat (m\kappa).
\end{align}

The last equality follows because the minimum degree $\delta \geq \kappa$. In this paper, we show a deterministic construction by exploiting the property of the cluster $V_i$. In \Cref{subsec:lemmadatastructure}, we prove the following: Given $s \in L, t \in R$, and a cluster $V_i$ such that $L \subseteq V_i$, we can deterministically compute a kernel of size $\tilde O(|V_i| |L|). $
By the assumption of Case 1, the kernel is of size
$ \tilde O(|V_i| |L|) = \tilde O( |S| |L| \log n) = \tilde O(\kappa \cdot |L|) $ as desired.
We remark that the kernel of size $\tilde O(\kappa |L|)$ exists by the randomized construction of \cite{li_vertex_2021}.

\paragraph{Case 2:} There is a cluster $V_i$ such that $L \subseteq V_i, |S\cap V_i| \leq O(|L|\log n) \text{ and } |V_i| = \omega(|S|\log n)$ (which implies $|R\cap V_i|\ge |S\cap V_i|$ because of \Cref{eq:LltS}).  This case is similar to the proof of~\cref{lem:balancedterminalalgorithm}, which is based on selectors (based on linear lossless condensers~\cite{guruswami2009unbalanced,Cheraghchi11}) and the isolating cut lemma~\cite{li_vertex_2021}.

Let $T = V_i$ be a terminal set. Since $|L \cap T|, |R \cap T| \geq |S \cap T|$ and $|L \cap T| = |L| \overset{(\ref{eq:LltS})}{\leq} O(|S|)$, the vertex cut $(L,S,R)$ is terminal-balanced. Therefore, applying ~\cref{lem:balancedterminalalgorithm} with $T=V_i$ on every cluster $V_i$ would obtain a vertex mincut in one of the clusters.  However, one might notice that the running time is $z\cdot \hO{m\kappa}$, which is too large for us. To address this, we will use~\cref{lem:balancedterminalalgorithm} not on the whole graph, but only on the subgraph containing all edges incident to $V_i$, which is enough for us since $L\subseteq V_i$. The details are stated in~\cref{lem:subgraphbalancedterminalalgorithm}.

\subsubsection{Overview of Terminal Reduction: \texorpdfstring{\Cref{lem:terminalreduction}}{terminal reduction lemma}}
\label{sec:overview terminal sparsification}

Fix a vertex mincut $(L,S,R)$ in $G$. We say that a terminal set $T$ is $\alpha$-\textit{balanced} if $|T \cap L|,|T\cap R|\geq \alpha |S|$.  We use $\alpha$-balanced instead of $\alpha$-strong-balanced for simplicity. The same argument goes through for $\alpha$-strong-balanced terminal set. In this overview, we outline an $\Ohat(m \kappa)$-time algorithm for the following version.

\begin{quote}
 \textbf{Terminal Reduction:} Given  a $2^{\sqrt{\log n}}$-balanced terminal set $T \subseteq V$, compute    $T' \subseteq V$ and a separator $S^*$ where $|T'| \leq 0.9|T|$ such that
 \begin{itemize} [nolistsep]
         \item $T'$ is 0.1-balanced, or
         \item $S^*$ is a minimum separator in $G$.
     \end{itemize}
\end{quote}

 We now outline the high-level ideas for the terminal reduction procedure. We define $\gamma = 2^{\sqrt{\log n}}$.
\begin{enumerate}
    \item Given $T$, we compute $(X,\mathcal{U})$ a \textit{terminal vertex expander decomposition} on graph $G$ (see \Cref{lem:vertexexpanderdecomposition} for a formal statement). That is, we obtain a vertex set $X$ and a collection of vertex disjoint sets $\mathcal{U}$ such that
    \begin{itemize}
        \item $X$ and all vertex sets in $\mathcal{U}$ form a vertex partition. That is, they are disjoint and the union of them form the vertex set $V$,
        \item There are no edges between two different sets $U_1, U_2 \in \mathcal{U}$,
        \item For each $U \in \mathcal{U}$, the induced subgraph $G[U]$ is an expander with respect to $T$, and
        \item $|X| \leq 0.1|T|$.
    \end{itemize}
      \textbf{Simplification.} In this technical overview, we assume that for each vertex set $U \in \mathcal{U}$,  $U \subseteq L$, $U \subseteq S$, or $U \subseteq R$.
      This setting conveys most of the main ideas.
      This assumption is close to be true because each vertex set $U$ is an expander with respect to $T$, the mincut $(L,S,R)$ cannot cross $U$ in a balanced way, i.e., either most terminals in $U$ are contained in $L$, $S$, or $R$, which is similar to our assumption.
    \item \textbf{Main goal.} We will construct three  terminal sets $T_{\text{big}}, T_{\text{small}}, T_{\bar X}$ so that the terminal set $T' := X \cup T_{\text{big}} \cup T_{\text{small}}\cup T_{\bar X}$ is of size at most $0.9|T|$ and compute a vertex cut $S^*$ of $G$ so that either
    \begin{itemize}
        \item [(i)] The new terminal set $T'$ is $0.1$-balanced, i.e., it satisfies $|L \cap T'|,| R\cap T'|\geq 0.1|S|$, or
        \item [(ii)] $|S^*| = \kappa$.
    \end{itemize}
    For simplicity, in the following analysis, our main goal will be to show that we will only  $$|L \cap T'|\geq 0.1|S|,$$ the other side $|R\cap T'|\ge 0.1|S|$ can be achieved by the symmetric argument.

    For each $U \in \mathcal{U}$, we denote $T_U := U \cap T$ as the terminal set inside $U$. Let $\mathcal{U}_{\text{small}} := \{ U \in \mathcal{U} \colon |T_U|=1\}$, and $\mathcal{U}_{\text{big}} := \{ U \in \mathcal{U} \colon  |T_U| > 1\}$. Below, we will describe how to construct $T_{\text{small}}$ and $T_{\text{big}}$ based on $\mathcal{U}_{\text{small}}$ and $\mathcal{U}_{\text{big}}$.

    \item \textbf{We Can Focus on  $\mathcal{U}_{\text{small}}$ by   $T_{\text{big}}$.} We construct $T_{\text{big}}$ as follows: For each $U \in \mathcal{U}_{\textbig}$, add to $T_{\text{big}}$ arbitrary $\lceil |T_U|/2 \rceil$ terminals from $T_U$. It is easy to see that $|T_{\textbig}|\le (2/3)|T|$. The following claim shows the reason why we define $T_{\text{big}}$ in this way.

    \begin{claim}
        If most terminals in $L$ are contained in $\cU_{big}$, i.e., $\sum_{U\in\mathcal{U}_{\textbig}}|T_U\cap L|\ge |T\cap L|/2$, then $|L\cap T_{\textbig}|\ge |T\cap L|/4\ge 0.1|S|$.
    \end{claim}

    This is because at least half of the terminals in $T_U\cap L$ are included to $T_{\textbig}$ according to our simplification that $T_U\subseteq L$ if $T_U\cap L\not=\emptyset$. The inequality $|T\cap L|/4\ge 0.1|S|$ is because $T$ is $\gamma$-balanced.

    Let $L_C:=\{U\in\cU_{\textsmall}\mid U\subseteq L\}$. From now, we assume most terminals in $T\cap L$ are contained in $\cU_{\text{small}}$, i.e., $\sum_{U \in \mathcal{U}_{\textsmall}}| T_U \cap L| \geq 1/2|T\cap L|\ge \frac{\gamma}{2} \cdot |S|$. Since $\sum_{U\in \mathcal{U}_{\textsmall}}| T_U \cap L| =|L_C|$, we have
    \begin{align} \label{eq:Lc large wrt S}
     |L_C| \geq \frac{\gamma}{2} \cdot |S|.
    \end{align}

    Our next goal becomes this: pick a small number of clusters in $\cU_{\textsmall}$ (at most $0.01|T|$) so that they hit $L_C$ a lot (at least $|L_C|/\gamma$) and include all the terminals in these clusters into a terminal set $T_{\textsmall}$.
    \item \textbf{An Easy Case for $X$.} Before we describe how to pick a small number of clusters in $\cU_{\textsmall}$, it is instructive to first solve an easy case of $X$ where
   $|X\cap S| = \tilde O(|X \cap L|).$ Here, we assume that $|X \cap L| \leq 0.1|S|$ (otherwise $|T' \cap L| > 0.1|S|$ and we are done).

  In this case, we can almost apply~\cref{lem:balancedterminalalgorithm} on $X$ as a terminal set except that it requires $|X \cap S|\le |X\cap R|$, but $X$ may not intersect with $R$ at all. To handle this situation, we define $T_{\bar{X}}$ that contains an arbitrary $0.01|T|$ nodes from $T-X$. If $|T_{\bar{X}} \cap L| > 0.1|S|$, then we are done. Assume $|T_{\bar{X}} \cap L| \leq 0.1|S|$, and thus $|T_{\bar{X}} \cap R| \geq |S|$ since $|T|\ge \gamma|S|$ and $|T\cap X|\le 0.1|T|$. We claim that  $|T_{\bar{X}} \cap S| \leq |X \cap L|$. If true, then  the terminal $T'' := X \cup T_{\bar{X}}$ satisfies the conditions in~\cref{lem:balancedterminalalgorithm} and we can apply~\cref{lem:balancedterminalalgorithm}  on $T''$ to get a vertex mincut $S^*$ in $\hO{m \cdot |T'' \cap L|} \leq \hO{m \cdot (|X \cap L| + |T_{\bar{X}} \cap L|)} \leq \hO{m \kappa}$ time.

 We now prove $|T_{\bar{X}} \cap S| \leq |X \cap L|$. Fix a vertex set $U \in L_C$. By definition of the expander decomposition, $N_G(U) \subseteq X \cap (L \cup S)$. So, $|N_G(U)| \leq |X \cap S| + |X \cap L|. $ On the other hand, since $N_G(U)$ is a valid vertex cut, $|N_G(U)| \geq \kappa = |S|$. So $|S| \le |X \cap S| + |X \cap L|$, which is equivalent to $|S - X| \leq |X \cap L|$. The claim follows as $|T_{\bar{X}} \cap S| \le |S - X|$.

     \item \textbf{Pruning $X$ or Concentrated $T_{\textsmall}$.} To handle the general case of $X$, we define $T_{\bar X}$ using $X$ in the same way, and then we define $X' \subseteq X$ by removing vertices in $X$ in such a way that we never remove any vertex in $X \cap L$, and at the end, all but $\tilde O(|X \cap L|)$ vertices from $X \cap S$ have been removed. The process of carving out nodes of $X$ into $X'$ is called \textit{the pruning step}. To implement the pruning step, we will construct $T_{\textsmall} \subseteq T$ of size at most $0.1|T|$ and $X' \subseteq X$ such that either
     \begin{align} \label{eq:pruning X to X'}
         |T_{\textsmall} \cap L| \geq 0.1|S| \text{ or } X'\cap L = X \cap L, |X'\cap S| \leq \tilde O(|X \cap L|).
     \end{align}
     That is, either we find a small number of clusters in $\calU_{\textsmall}$ so that they hit $L_C$ a lot (and include all terminals in these clusters in $T_{\textsmall}$) or we manage to prune $X$ into $X'$ satisfying \Cref{eq:pruning X to X'} (we say that $T_{\textsmall}$ is \textit{concentrated} if $|T_{\textsmall} \cap L| \geq 0.1|S|$). Therefore,  we can add  $T_{\textsmall}$ to the final terminal set and run the algorithm for the easy case on $X'$ (using $T_{\bar X}$ defined by $X'$), and we are done.

      \begin{quote}
         The main challenge for the pruning step is to characterize the conditions when a vertex $v \in X$ is in $X \cap L$ and in $X \cap S$ \textit{without the knowledge} of set $L$ and $S$.
     \end{quote}

The rest of the section is devoted to the pruning step.
\item  \textbf{Common-neighborhood Clustering on the Expander Decomposition.} The solution, which is the crux of our algorithm, is that we run the sparse neighborhood cover algorithm on the common neighborhood graph obtained from the bipartite graph between $X$ and $\calU$ obtained by contracting each component in $\calU$. More precisely, we define a bipartite graph $H := (X \cup \calU_{\textsmall}, E')$ where there is an edge $(x,U) \in E'$ if and only if there is an edge between $x \in X$ to $U \in \calU_{\textsmall}$ in $G$.
    To see the intuition why it helps, we have the following claim. Let $a=|X\cap L|<0.1|S|$ (if $|X \cap L| \geq 0.1|S|$, then we are done). We assume the algorithm knows $a$ (we can estimate $a$ up to a constant factor by guessing the size $O(\log n)$ times by the powers of $2$).
      \begin{claim}
          For any $U_1,U_2\in L_C$, we have $|N_H(U_1)\triangle N_H(U_2)|\le 2a$.
      \end{claim}
      \begin{proof}
          Notice that $N_H(U_1)\subseteq X$ is a vertex cut in the original graph $G$, since different clusters in $\cU$ do not have edges connecting them. Moreover, $N_H(U_1)\subseteq (L\cup S)\cap X$. The same holds for $U_2$. Thus, we have $\kappa\le |N_H(U_1)|,|N_H(U_2)|\le \kappa+a$ and $|N_H(U_1)\cup N_H(U_2)|\le \kappa+a$. This gives us $|N_H(U_1)\triangle N_H(U_2)|\le 2a$.
      \end{proof}

      By running common-neighborhood clustering on $H[\cU_{\textsmall}]$ (i.e., running sparse neighborhood cover on the common-neighborhood graph where there is an edge between $U_1,U_2$ if $|N_H(U_1)\triangle N_H(U_2)|\le 2a$), we can get $V_1,...,V_z$ with the following properties similar to~\cref{lem:comm clustering mk}

      \begin{claim}\label{claim:clustering2} We can obtain clusters $V_1,\ldots V_z \subseteq \cU_{\textsmall}$ satisfying
                \begin{enumerate}
            \item[(a)] $L_C \subseteq V_i$ for some cluster $V_i$,
            \item[(b)] for every cluster $V_i$ and every two vertices $U_1,U_2 \in V_i$, $|N_H(U_1) \triangle N_H(U_2)| \leq O(a\log n)$.
        \end{enumerate}

      \end{claim}

\begin{remark}
    We need an $\hO{m\kappa}$ algorithm for computing $V_1,...,V_z$.
    However, we cannot apply the same intermediate graph because  $H[L_C]$ is an independent set.  To fix this, the idea is to build new virtual edges between $\cU_{\textsmall}$ carefully so that, the total number of edges is bounded by $O(m)$, and after adding those edges, $L_C$ is connected. We refer to \Cref{lem:shaving} for details.
\end{remark}

    \item \textbf{Proof Strategy for the Pruning Step.}  We  exploit the structure of the common-neighborhood clustering on $H[\cU_{\textsmall}]$ to help us identify the nodes in $X \cap S$ while avoiding the nodes in $X \cap L$ as follows.  We denote $V^*$ as a cluster that contains $L_C$. Note that $V^*$ exists by \Cref{claim:clustering2}(a).

At high level, we argue that all but $\tilde O(a)$ vertices in $X \cap S$  have at least $\tau$ neighbors in $V^*$ in graph $H$ (where $\tau$ is a parameter chosen to be closer to $|V^*|$) using the properties of the common-neighborhood clustering on $H[\cU_{\textsmall}]$. At the same time, we ensure that every node in $X \cap L$ has fewer than $\tau$ neighbors in $V^*$ via $T_{\textsmall}$ (to be defined). Therefore, deleting all nodes in $X$ having at least $\tau$ neighbors in $V^*$ would complete the pruning step. Since the algorithm does not know the set $V^*$, we will try all clusters $V_1,\ldots,V_z$ of a certain size and argue that we never delete any vertex in $X\cap L$.

We next formalize the ideas.

\item  \textbf{Structural Property of $V^*$.} We start with a key property from the common-neighborhood clustering on $H[\cU_{\textsmall}]$. For any vertex $x \in X$ and $U \subseteq \cU_{\textsmall}$, we denote the degree of $x$ in $U$ in graph $H$ as $\textdeg_{H}(x,U) := |N_H(x) \cap U|$.

\begin{lemma} \label{lem:almost all X cap S removed}
    For all $\delta \in (0,1)$, all but $O(a \cdot \delta^{-1} \log n)$ vertices in $X \cap S$ have at least $(1-\delta )|V^*|$ neighbors in $V^*$ in graph $H$. That is, $|Z_\delta| \leq  O(a \cdot \delta^{-1} \log n)$  where  $Z_{\delta} := \{ x \in X \cap S \colon \textdeg_{H}(x,V^*) < (1-\delta)|V^*|\}$ .
\end{lemma}
\begin{proof}
    Fix $U \in L_C \subseteq V^*$. We have $|N_H(U) \cap (X\cap S)| \geq |S| - |X \cap L| \geq |S \cap X| - a$. By \Cref{claim:clustering2}(ii), $|N_H(U) \triangle N_H(U')| \leq O(a \log n)$ for all $U' \in V^*$. Therefore, for all $U' \in V^*$, $|N_H(U') \cap (X\cap S)| \geq |S \cap X| - O(a\log n)$. In other words, for all $U \in V^*$, the number of non-edges from $U$ to $X\cap S$ is at most $O(a \log n)$. Define $\mathtt{NN}_H(X\cap S, V^*)$ to be the number of non-edges between $X\cap S$ and $V^*$ in $H$. We have the following.
    $$ \delta \cdot |V^*| \cdot  |Z_{\delta}| \leq \mathtt{NN}_H(X\cap S, V^*) \leq |V^*| \cdot O(a \log n). $$
    The lower bound follows since each vertex in $Z_\delta$ has at least $\delta \cdot |V^*|$ non-edges to $V^*$. Therefore, $|Z_{\delta}| = O(a \cdot \delta^{-1} \log n).$
\end{proof}

  Suppose we are given the set $V^*$. How to design an algorithm that removes all but $\tilde O(a)$ vertices in $X \cap S$ and none in $X \cap L$? We almost have the pruning step except that we need the following property for some $\delta \in (0,1)$ to be specified:
  \begin{align}\label{eq: nice prop} (1+\frac{2}{\gamma})|L_C| < (1 - \delta)|V^*|. \end{align}

  Given \Cref{eq: nice prop}, we can characterize $X \cap L$ and almost all nodes in $X\cap S$ using  $\tau := (1 - \delta)|V^*|$. Indeed,  for all $y \in X \cap L$ and for all but $\tilde O(a \cdot \delta^{-1})$ vertices $x \in X \cap S$,  we have
  $$  \textdeg_H(y, V^*) \leq (1+\frac{2}{\gamma})|L_C| \overset{(\ref{eq: nice prop})}{<} \tau = (1 - \delta)|V^*| \leq \textdeg_{H}(x,V^*).$$
The first inequality follows since  every vertex in $X \cap L$ has degree at most $|L_C| + |S| \overset{(\ref{eq:Lc large wrt S})}{\leq}  (1+ \frac{2}{\gamma}) |L_C|$ in $H$, and thus it has at most $(1+\frac{2}{\gamma})|L_C|$ neighbors in $V^*$ (w.r.t. $H$).  The last inequality follows since, by \Cref{lem:almost all X cap S removed}, for all $\delta \in (0,1)$, all but $\tilde O(a \cdot \delta^{-1})$ vertices in $X \cap S$ have a lot of neighbors, i.e., at least $(1-\delta)|V^*|$ in $V^*$ (w.r.t. $H$).

 Therefore, by deleting all nodes $x \in X$ such that   $ \textdeg_{H}(x,V^*) \geq \tau$, we delete all but $\tilde O(a \cdot \delta^{-1})$ vertices in $X \cap S$ while we do not delete any node in $X \cap L$.

Next, we establish \Cref{eq: nice prop} by using $T_{\textsmall}$ and we will set $\delta^{-1} = O(\log^2n)$.

\item \textbf{Small $L_C$ by $T_{\textsmall}$.}  We construct $T_{\textsmall}$ as follows. For each $V_i$, we pick arbitrary $|V_i|/\log^2n$ elements in $V_i$. Recall that an element in $V_i$ corresponds to an expander in $\cU_{\textsmall}$. For each of these expanders, include its unique terminal into $T_{\textsmall}$.

The size of $T_{\textsmall}$ is negligible compared to the size of $T$ because $|T_{\textsmall}|\le O(\log n)\cdot (|T|/\log^2n) = o(|T|)$ where the $O(\log n)$ factor is from the fact that each element is contained in at most $O(\log n)$ clusters $V_i$.

If $|L_C|\geq (1-\frac{1}{2\log^2n})|V^*|$, then $|T_{\textsmall}\cap L|\ge |L_C| \cdot \frac{1}{2\log^2n}$. Since $|L_C| \overset{(\ref{eq:Lc large wrt S})}{\ge}  \frac{\gamma}{2}\cdot |S|$, we have $|L \cap T'| \ge 0.1|S|$ as desired. From now on, we assume
\begin{align} \label{eq:small Lc by Tsmall}
    |L_C|<(1-\frac{1}{2\log^2n})|V^*|.
\end{align} Therefore,
for all $y \in X \cap L,$
  $$  \textdeg_H(y, V^*) \leq (1+\frac{2}{\gamma})|L_C| < (1 - \frac{1}{3\log^2n})|V^*|.$$
By setting $\delta = \frac{1}{3\log^2n}$ in \Cref{lem:almost all X cap S removed}, we have for all $x \in (X\cap S) - Z_{\delta}$,
  $$ (1 - \frac{1}{3\log^2n})|V^*| \leq \textdeg_{H}(x,V^*).$$

  We now fix $\delta := \frac{1}{3\log^2n}$ and let $Z_{\delta}$ be the set defined in \Cref{lem:almost all X cap S removed} and define $\tau :=(1 - \frac{1}{3\log^2n})|V^*|$.  We conclude with the following lemma.
\begin{lemma} \label{lem:almost all X cap removed setting delta}
   For all $y \in X \cap L$, and $x \in (X \cap S) -  Z_{\delta}$,   $\textdeg_H(y, V^*) < \tau \leq \textdeg_H(x, V^*)$.
\end{lemma}

Therefore, deleting all nodes in $X$ which have at least $\tau$ neighbors in $V^*$ would complete the pruning step. The remaining piece of the puzzle is that the algorithm does not know the cluster $V^*$, which we describe how to deal with it next.

\item \textbf{Finalize the Pruning Step.}
For simplicity, let us assume that $V_1,...,V_z$ is a partition of the graph.\footnote{In fact, a sparse neighborhood cover is a collection of $O(\log n)$ partitions of the graph. To remove this assumption, it is enough to repeat the algorithm for all the $O(\log n)$ partitions.}
The procedure for deleting vertices in $X$ is as follows:  For every cluster $V_i$ from \Cref{claim:clustering2}, delete all vertices $x \in X$ such that $$ \textdeg_H(x,V_i) \geq \max \biggl \{ (1-\frac{1}{3\log^2n})|V_i|,2|S| \biggr\}.$$
Let $X'$ be the remaining vertices in $X$. We now prove the two properties of $X'$.

    \begin{enumerate}

        \item \textbf{$|X'\cap S| \leq \tO{|X \cap L|}$.} Since we try for all $V_i$, we show that the iteration $V_i = V^*$ will delete a lot of vertices in $X$. If $V_i = V^*$, then
        $$ \tau = (1-\frac{1}{3\log^2n})|V^*|  \overset{(\ref{eq:small Lc by Tsmall})}{>} |L_C| \overset{(\ref{eq:Lc large wrt S})}{\ge} \frac{\gamma}{2}|S| \geq 2|S|.$$ Therefore,  \Cref{lem:almost all X cap removed setting delta} and \Cref{lem:almost all X cap S removed} imply that $|X' \cap S| \leq |Z_{\delta}| \leq O(a \log^3n) = \tO{a}.$

        \item \textbf{$X'\cap L=X\cap L$.} We   prove that no vertices in $X\cap L$ are deleted. If $V_i = V^*$, then  $V^*$ cannot delete vertices in $X \cap L$ by  \Cref{lem:almost all X cap removed setting delta}. If $V_i \neq V^*$, then notice that $X\cap L$ has at most $|S|$ neighbors in $V_i$. This is because $X\cap L$ can only be adjacent to vertex sets in $L_C$ or vertex set that intersects with $S$. Moreover, the former case cannot happen since $L_C \subseteq V^*$ which is disjoint from $V_i$.
        Thus, $X\cap L$ has at most $|S|$ neighbors in $V_i$ and so $V_i$ cannot delete vertices in $X\cap L$.

    \end{enumerate}

\end{enumerate}

\subsection{Terminal-Balanced Case: Proof of  \texorpdfstring{\Cref{lem:balancedterminalalgorithm}}{balanced terminal algorithm}}
\label{sec:base cases alg}

In this section, we show the algorithm for the case when a vertex mincut is balanced with respect to a terminal set (\Cref{lem:balancedterminalalgorithm}) and its extension that will be used later (\Cref{lem:subgraphbalancedterminalalgorithm}).
Our strategy is to employ pseudorandom objects (i.e.,~\cref{thm:crossing family,thm:selector}).

We first need to introduce the vertex version of the isolating cuts lemma, which is from \cite{li_vertex_2021}.

\begin{lemma}[Lemma 4.2 of \cite{li_vertex_2021}]\label{isolatingvertexcuts}
There exists an algorithm that takes an input graph $G=(V, E)$ and an independent set $I \subset V$ of size at least 2, and outputs for each $v \in I$, a $({v}, I\setminus {v})$-min-separator $C_v$.
The algorithm makes $(s,t)$ max-flow calls on unit-vertex-capacity graphs with $O(m\log |I|)$ total number of vertices and edges and takes $O(m)$ additional time.
\end{lemma}

We are now ready to prove \Cref{lem:balancedterminalalgorithm}. The proof is based on an application of selector~\cref{thm:selector} and  crossing family~\cref{thm:crossing family},  which shows $\kappa$ calls to max flow are possible when $|L|$ is roughly the same as $|S|$ and $|R|$ larger than $|S|$.

\begin{proof}[Proof of \Cref{lem:balancedterminalalgorithm}]

    Suppose $|T\cap L|\le k\cdot n^{g(n)}$ and $|T\cap L|,|T\cap R|\ge n^{-g(n)}\cdot |T\cap S|$ where $g(n)=o(1)$. Let $\eps=n^{-2g(n)}$ and let $\cS$ be a $(|T|,k/\epsilon,\epsilon)$-selector which is a family of subsets of $[|T|]$. Here we map each number in $1,...,|T|$ to a vertex in $T$, and with a little abuse of notation, each $X\subseteq\cS$ is a vertex subset of $T$. Here we assume $k/\epsilon\le |T|^{1-o(1)}$ for any function $o(1)$. the case when $k$ is large will be solved in the next paragraph. According to~\cref{thm:selectoreasy}, we can construct $\cS$ of size $|\cS|=k\cdot n^{o(1)}$ in $\hO{n}$ time, and every set $X\subseteq\cS$ has size at least $2$. Notice that $k<|T\cap L|\le k/\epsilon$ and $|T\cap S|\le k/\epsilon$. Now according to the definition of selector, there exists $X\subseteq \cS$ such that $|X\cap L|=1$ and $|X\cap S|=0$ and $|X|\ge 2$. According to~\cref{isolatingvertexcuts}, by setting the independent set $I$ to be a maximal independent set of $X$, which must satisfy $|I\cap L|=1,|I\cap S|=0$ and $|I|\ge 2$, we can get a vertex cut of size at most $S$ since $S$ itself is an isolating cut. Notice that the running time is $\widehat{O}(m|\cS|)=\widehat{O}(mk)$.

    Now suppose $k/\epsilon> |T|^{1-o(1)}$, which means $k\cdot n^{o(1)}\ge |T|$. We construct a $(|T|,1/\epsilon)$-crossing family $\cP$ according to~\cref{thm:crossing family}. Similarly, we abuse the notation a bit and assume each $P\subseteq\cP$ is a subset of $T\times T$. According to the definition of crossing family, there exists $(s,t)\in P$ such that $s\in L,t\in R$. Thus, by running $s-t$ min cut over all $(s,t)\in P$ and find the minimum one, we get the minimum vertex cut of $G$. The algorithm runs in time $\hO{m|T|}$ which is fine since $k\cdot n^{o(1)}\ge |T|$.
\end{proof}

The following lemma is a more fine-grained version of~\cref{lem:balancedterminalalgorithm}, which works on a subgraph.
The proof is similar to that of \Cref{lem:balancedterminalalgorithm}.
This will be needed later for the unbalanced algorithm in \Cref{sec:proof of unbal assume ds}. It is incomparable to~\cref{lem:balancedterminalalgorithm} because it further requires $L$ to be a subset of $T$, where~\cref{lem:balancedterminalalgorithm} does not have this requirement.

\begin{lemma}\label{lem:subgraphbalancedterminalalgorithm}
    There is a deterministic algorithm \textsc{SubgraphBalancedTerminalVC}$(G,T,k)$ that takes as inputs an undirected graph $G=(V,E)$, a terminal set $T\subseteq V$ and a cut parameter $k$, outputs a vertex cut of $G$. If there exists a vertex cut $(L,S,R)$ of $G$ such that $L\subseteq T$ and $|L|,|T\cap R|=\widehat{\Omega}(|T\cap S|)$ and $|L|=\widehat{O}(k)$, then the output cut is a minimum vertex cut of $G$. The algorithm runs in $\widehat{O}(k\sum_{t\in T}|N_{G}(t)|)$ time.
\end{lemma}

\begin{proof}
    The proof is similar to the proof of~\cref{lem:balancedterminalalgorithm}. Suppose $|T\cap L|=|L|\le k\cdot n^{g(n)}$ and $|L|,|T\cap R|\ge n^{-g(n)}\cdot |T\cap S|$ where $g(n)=o(1)$ and define $\epsilon$ and $\cS$ and $I$ in the same way as in the proof of~\cref{lem:balancedterminalalgorithm}. We will explain what to do when $k$ is large later. Notice that $I$ has the property $|I\cap L|=1,|I\cap S|=0$ and $|I|\ge 2$. Now we cannot afford to run isolating cuts on the whole graph because we want the running time to be smaller. Instead, we construct a graph
    \[G'=(T\cup N_G(T)\cup \{s\},E(T,T\cup N_G(T))\cup \{(s,u)\mid u\in N_G(T)\}\]

    Define $I'=I\cup\{s\}$. Notice that $I\cup\{s\}$ is still an independent set because $I\subseteq T$.  The algorithm use~\cref{isolatingvertexcuts} on $I',G'$. The running time is bounded by $\widehat{O}(k\sum_{t\in T}|N_{G}(t)|)$ where $O(\sum_{t\in T}|N_{G}(t)|)$ is the size of $G'$. Now we prove that the minimum isolating cut of $I'$ (denoted by $C_v$ which is a $(v,I'\backslash v)$-min-separator for some $v\in I'$) on $G'$ is equal to $|S|$. Notice that $S=N_G(L)=N_{G'}(L)$ since $L\subseteq T$. Also, notice that $S\cap I'=\emptyset$. Thus, $S'$ is an isolating cut of $I'$, which means $|S|\ge |C_v|$ for some $v \in I'$. On the other way, suppose the connected component containing $v$ after deleting $C_v$ is $L_v$. We have $L_v\subseteq T$, otherwise $s\in N(L_v)$. Thus, we have $N_G(L_v)=N_{G'}(L_v)$. Notice that $N_G(L_v)\cup L_v\not=V$ since $|I|\ge 2$, and there must exists a vertex in $I$ other than $v$ inside $T$ outside $N_G(L_v)\cup L_v$. Thus, $C_v$ is a vertex cut of $G$, which implies $|S|\le |C_v|$.

    Now suppose $k/\epsilon>|T|^{1-o(1)}$, which means $k\cdot n^{o(1)}\ge |T|$. Similar to the proof of~\cref{lem:balancedterminalalgorithm}, we construct an $(|T|,1/\epsilon)$-crossing family $\cP$. There exists $(a,b)\in P$ such that $a\in L,b\in R$. Instead of running $(a,b)$ min cut on $G'$ over all $(a,b)\in P$, we run $(a,b\cup\{s\})$ min cut on $G'$. Any $(a,b\cup\{s\})$ min cut $(L',S',R')$ satisfying $L'\subseteq T$ since $s$ connects to all nodes in $N_G(T)$. Thus, $N_G(L')=N_{G'}(L')$ and $t\in R'$ where $t\in V$. Thus, $S'$ is a vertex cut of $G$. On the other hand, the minimum vertex cut $S$ of $G$ is also a vertex cut of $(a,b\cup\{s\})$ for some $(a,b)\in P$ since $L\subseteq T$. The total running time is $\widehat{O}(k\sum_{t\in T}|N_{G}(t)|)$.
\end{proof}

\subsection{Faster Common-Neighborhood Clustering Algorithm}\label{sec:sparseneighborhoodcover}

In~\cref{subsec:weightedunbalanced} step 1, we presented a common-neighborhood clustering algorithm that runs in $O(mn)$ deterministic time. This running time is too slow as our goal for this section is $\hO{m\kappa}$. Thus, we present a faster algorithm in this section. We restate the common-neighborhood clustering lemma as in~\cref{lem:weightedcommonneighborhood} that fits our purpose here.

In the following lemma, the symmetric distance oracle $\cAd:V\times V\to\bbN$ can be viewed as an algorithm computing the value of $|N_G(u)\triangle N_G(v)|$. The (sparse) property further requires $\cC$ to be a collection of $O(\log n)$ partitions of the graph, which is crucial for our usage. The (cover) property can cover all $L\subseteq V$ that is connected and has a low neighborhood symmetric distance, which covers the minimum cut $L$ as we want.

\begin{lemma}\label{lem:sparseneighborhoodcover}
    There is a deterministic algorithm \textsc{CNC}$(G,\cAd,d)$ (\cref{alg:snc}) that takes as inputs an $n$-vertex $m$-edge undirected graph $G=(V,E)$, a symmetric distance oracle $\cAd:V\times V\to \bbN$ and $d\in\bbN^+$ which satisfies triangle inequality $\cAd(u,v)+\cAd(v,w)\ge \cAd(u,w)$, outputs a set of partitions of $V$ denoted by $\cP$ such that
    \begin{enumerate}
        \item (sparse) $|\cP|=O(\log n)$, for every $P\in\cP$, $P\subseteq 2^V$ is a partition of $V$, which means $\cup_{C\in P}C=V$ and $C_1\cap C_2=\emptyset$ for $C_1,C_2\in P$ with $C_1\not=C_2$, we call $C$ a \emph{cluster},
        \item (cover) for any $L\subseteq V$ where $G[L]$ is connected and $\cAd(u,v)\le d$ for any $u,v\in L\cap U$, there exists a cluster $C$ such that $L\subseteq C$,
        \item (common-neighborhood) for any cluster $C$, for any $u,v\in C$, we have $\cAd(u,v)=O(d\log n)$.
    \end{enumerate}
    The algorithm uses $\tO{m}$ calls to $\cAd$ and in addition $\tilde{O}(n\max_{v\in V}|N_G(v)|)$ time.
\end{lemma}

\begin{algorithm}[H]
    \DontPrintSemicolon
    \caption{\textsc{CNC}$(G,\cAd,d)$}
    \label{alg:snc}
    Let $\cP\leftarrow \emptyset,U\leftarrow V$\;
    Let $G'=(V,H=\{(u,v)\in E\mid \cAd(u,v)\le d\})$\;
    \While{$U\not=\emptyset$\label{snclinewhile1}}
    {
        $\cC\leftarrow \emptyset$\;
        $U_{rem}\leftarrow U,U_{next}\leftarrow \emptyset$\;
        \While{$U_{rem}\not=\emptyset$\label{snclinewhile2}}
        {
            Let $u$ be an arbitrary node in $U_{rem}$\;
            Set $T_0 = \emptyset,T_1=\{u\},i=1$\;
            \While{$|N_{G'[U_{rem}]}(T_i)|\ge |T_{i-1}|$ or $|T_i|\ge 1.1|T_{i-1}|$\label{snclinewhile3}}
            {
                $i\leftarrow i+1$\;
                Build BFS tree $T_{i}$ on $G'[U_{rem}]$ with root $u$ and only include $v$ to $T$ if $\cAd(u,v)\le 5id$\label{snclineTi}\;
            }
            Build BFS tree $T$ on $G'[U_{rem}]$ with root $u$ and only include $v$ to $T$ if $\cAd(u,v)\le (5i-3)d$\label{snclineTB}\;
            $\cC \gets \cC\cup\{T\}$\label{snclineT}\;
            $U_{rem} \gets U_{rem}-T$\label{snclinerem}\;
            $U_{next} \gets U_{next}\cup (T_i-T_{i-1})$\label{snclineU}\;
        }
        $U\leftarrow U_{next}$\;
        $\cP\leftarrow\cP\cup\{\cC\}$
    }
    \Return{$\cP$}
\end{algorithm}

\begin{proof}[Proof of~\cref{lem:sparseneighborhoodcover}]
    We first prove \emph{(sparse)}. Notice that in the same while loop (\cref{snclinewhile2}), the tree nodes included in $\cC$ (\cref{snclineT}) are disjoint since~\cref{snclineTB} creates $T$ on $G'[U_{rem}]$ where $T$ is deleted from $U_{rem}$ (\cref{snclinerem}). They form a partition of $V$ since the loop ends if $U_{rem}=\emptyset$. Thus, we only need to prove that the outer while loop (\cref{snclinewhile1}) only contains $O(\log n)$ loops. We will prove that in each loop, $|U_{next}|\le 0.1|U|$, which implies the $O(\log n)$ loop numbers since we set $U\leftarrow U_{next}$ at the end of each loop. Notice that $U_{next}$ includes $T_i-T_{i-1}$ (\cref{snclineU}) where $T_{i-1}\subseteq T$ which is deleted from $U_{rem}$, and from the condition of~\Cref{snclinewhile3} we have $|T_{i}|<1.1|T_{i-1}|$ which implies $|T_i-T_{i-1}|\le 0.1|T|$. Thus, we can charge $|T_i-T_{i-1}|$ into $0.1|T|$. Notice that the sum of all $|T|$ is at most $|U|$ (\cref{snclinerem}), we have $|U_{next}|\le 0.1|U|$.

    Next, we prove \emph{(cover)}. Suppose $L\subseteq V$ where $G[L]$ is connected and $\cAd(u,v)\le d$ for any $u,v\in L$. We will prove that there exists $C\in\cC\in\cP$ such that $L\subseteq C$. For this purpose, we only need to prove the following loop invariance: in the while loop~\cref{snclinewhile1}, if $L\subseteq U$ in the beginning, then at the end of this loop, either there exists $C\in\cC$ such that $L\subseteq C$ or $L\subseteq U_{next}$. Since $U_{next}$ will be $U$ in the next loop, the loop invariance will imply \emph{(cover)}. Now suppose $L\subseteq U$. Notice that the trees $T$ included in $\cC$ (\cref{snclineT}) in this while loop form a partition of $U$, there must exist $T$ such that $L\cap T\not=\emptyset$. Let $T$ be the first one in the loop of~\cref{snclinewhile2} that $L\cap T\not=\emptyset$ (which means $L\subseteq U_{rem}$ at this point). If $L\subseteq T$ then we are done since $T$ is included in $\cC$. Suppose $L\not\subseteq T$, we claim that $L\cap T_{i-1}=\emptyset$. Otherwise, there exists $w\in L\cap T_{i-1}$ such that $\cAd(u,w)\le (5i-5)d$, which implies $\cAd(u,v)\le (5i-4)d$ for any $v\in L$. Also notice that for any edge $(u,v)\in G[L]$ we have $\cAd(u,v)\le d$, thus, $(u,v)\in H$ and $G[L]=G'[L]$ which is connected. Thus, we get $L\subseteq T$ since $T$ is a BFS tree that includes a node $v$ if $v\in N_{G'[U_{rem}]}(T)$ and $\cAd(u,v)\le (5i-3)d$ (\cref{snclineTB}), a contradiction. Now we have $L\cap T_{i-1}=\emptyset$. Let $w\in L\cap T$, we have $\cAd(u,w)\le (5i-3)d$ and $\cAd(u,v)\le d$ for any $v\in L$, which implies $\cAd(u,v)\le (5i-2)d$ for any $v\in L$.  Therefore, $L\subseteq T_i$ since $T_i$ is a BFS tree with root $u$ which includes all nodes $v$ with $\cAd(u,v)\le 5id$. Therefore, $L\subseteq T_i-T_{i-1}$, which means $L\subseteq U_{next}$.

    Next, we prove \emph{(common-neighborhood)}. This is because $\cC$ only includes $T$ where $T$ contains all nodes $v$ satisfying $\cAd(u,v)\le (5i-3)d$ for a specific node $u$. If we can prove that $i=O(\log n)$, then we are done. For that purpose, we need to prove the while loop~\cref{snclinewhile3} only contains $O(\log n)$ loops. There are two cases for the while loop to continue. One is $|T_i|\ge 1.1|T_{i-1}|$, which can happen at most $O(\log n)$ times since $|T_i|\le n$ for any $i$. In the case that $|N_{G'[U_{rem}]}(T_i)|\ge |T_{i-1}|$, we will prove that $|T_{i+1}|\ge |T_{i}|+|T_{i-1}| \ge 2|T_{i-1}|$. Indeed, for every vertex $v\in N_{G'[U_{rem}]}(T_i)$, there exists $w\in T_i$ such that $(w,v)\in H$, which means $\cAd(w,v)\le d$, thus, $\cAd(u,v)\le \cAd(u,w) + \cAd(w,v)\le 5(i+1)d$ and $v\in T_{i+1}$. Thus, $|T_{i+1}|\ge |T_{i}|+|T_{i-1}|$ as claimed.
    In both cases, each loop of  \cref{snclinewhile3} implies that $|T_{i+1}|\ge 1.1|T_{i-1}|$, which can happen at most $O(\log n)$ times.

    Next, we prove the running time. To construct $G'$, we need $m$ calls to $\cAd$ and $O(m)$ running time. As proven above, there are at most $O(\log n)$ loops of~\cref{snclinewhile1}. In each loop, as proven above, there are at most $O(\log n)$ loops of~\cref{snclinewhile3}. Each of them builds a BFS tree with root $u$, while building the tree $T_i$ will cost $|T_i\cup N_{G'[U_{rem}]}(T_i)|$ many $\cAd$ calls (since we calculate the distance of each of them with $u$) and in addition the number of edges inside $T_i\cup N_{G'[U_{rem}]}(T_i)$ time, which is at most $|T_i\cup N_{G'[U_{rem}]}(T_i)|\cdot \max_{v\in V}|N_{G'}(v)|$. Notice that $T_{i-1}$ is deleted from $U_{rem}$ and $|T_i\cup N_{G'[U_{rem}]}(T_i)|=O(|T_{i-1}|)$ at the end of loop~\cref{snclinewhile3}. Thus, we can charge the oracles and running time to $\tilde{O}(|T_{i-1}|)$. Since $|U_{rem}|\le n$, the number of oracle calls is bounded by $\tilde{O}(n)$ and the running time is bounded by $\tilde{O}(n\max_{v\in V}|N_{G}(v)|)$. In total, the number of oracle calls is $m+\tilde{O}(n)=\tO{m}$.
\end{proof}

\subsection{Oracle for the Unbalanced Case}
\label{subsec:lemmadatastructure}

For the unbalanced case, we wish to use the crossing family which gives roughly $n/\ell$ pairs of nodes where $\ell$ is roughly the size of $L$, and one of them will cross $L,R$. We then want to know $\kappa(s,t)$ for each $(s,t)$ pair of them. However, the trivial way gives use a running time of $\hO{m}\cdot n\kappa/\ell$, which is far from our goal of $\hO{m\kappa}$. The following data structure gives us a way to answer each $\kappa(s,t)$ in $\hO{\delta\ell}$ time, which is sufficient for us as $\hO{\delta\ell}\cdot n\kappa/\ell=\hO{m\kappa}$. We will show how to use this data structure in the next section.

\begin{lemma}\label{lem:datastructure}
There is a deterministic algorithm that, given an undirected graph $G$ and an integer $\ell$, in $\widehat{O}(m\delta)$ time returns
\begin{enumerate}
    \item vertex sets $V_{1},\dots,V_{z}$ where each vertex $v\in V$ is contained in at most $O(\log n)$ many $V_i$ for different $i$,
    \item a data structure $\cD$ that, given any vertices $s$ and $t$, it returns an integer $\tka\ge\kappa(s,t)$ in $\hO{\delta \ell}$ time.
\end{enumerate}

Suppose $(L,S,R)$ is a minimum vertex cut of $G$. If $|L|\le |R|$ and $\ell=[|L|,2|L|]$, then one of the following two events happens.

\begin{enumerate}
    \item $\exists i,L\subseteq V_i,|V_i\cap S|=\tilde{O}(|L|),|V_i\cap R|=\tilde{\Omega}(|L|)$.
    \item If $\cD$ is given $s\in L,t\in R$, it will return $\tka=\kappa(G)$.
\end{enumerate}
\end{lemma}

We need the following sparse recovery tool from~\cite{NanongkaiS17}.

\begin{lemma}[Deterministic sparse recovery, corollary of Theorem 4.14~\cite{NanongkaiS17}]\label{lem:sparserecovery}
    There are deterministic algorithms $\cA_{construct}(G,x)$ and $\cA_{recover}(\vec{v},\vec{u})$ such that
    \begin{itemize}
        \item $\cA_{construct}$ takes as inputs an undirected graph $G=(V,E)$ and a positive integer $x$  outputs a vector $\vec{v}$ with length $\hat{O}(x)$ for each $v\in V$ in $\hat{O}(m)$ time,
        \item $\cA_{recover}$ takes as inputs two vectors $\vec{v},\vec{u}$ computed by $\cA_{construct}$, outputs $N_G(u)\triangle N_G(v)$ in $\hat{O}(x)$ time if $|N_G(u)\triangle N_G(v)|\le x$.
    \end{itemize}
\end{lemma}

\renewcommand{\vec}[1]{#1}
\begin{proof}
    In Theorem 4.14~\cite{NanongkaiS17}, we set $k=x$ and $n=|V|$, we can construct $\Phi\in\{0,1\}^{d\times n}$ ($d=\hat{O}(x)$) implicitly in $\hat{O}(n)$ time by knowing all the non-zero entities. Now we abuse the notation a bit and define $\vec{N_G(u)}\in\{0,1\}^n$ for each $u$ to be the identity vecter where the $i$-th entry is $1$ iff the $i$-th node in $V$ is in $N_G(u)$. $\cA_{construct}$ will output $\vec{v}=\Phi\vec{N_G(u)}$ for any $u\in V$, which takes time $\hat{O}(m)$ in total since $\Phi\vec{N_G(u)}$ takes time $\hat{O}(|N_G(u)|)$ (each column of $\Phi$ has $\hat{O}(1)$ non-zero entries).

    $\cA_{recover}$ will take $\Phi\vec{N_G(u)},\Phi\vec{N_G(v)}\in\{0,1\}^d$ and compute $\Phi(\vec{N_G(u)}-\vec{N_G(v)})$. According to Theorem 4.14~\cite{NanongkaiS17}, if $\vec{N_G(u)}-\vec{N_G(v)}$ contains at most $x$ non-zero entries (i.e., $|N_G(u)\triangle N_G(v)|\le x$), then all the non-zero entries are returned, which give us $N_G(u)\triangle N_G(v)$.
\end{proof}

\renewcommand{\vec}[1]{\textbf{#1}}
We show our algorithm for~\cref{lem:datastructure} as follows, i.e., how to compute the vertex sets $V,V_1,...,V_z$ and the memory space of the data structure $\cD$. We will show how $\cD$ uses the memory space to answer a query later.

\paragraph{Algorithm for~\cref{lem:datastructure}.} The algorithm takes inputs an undirected graph $G = (V,E)$ and an integer $\ell$. The following procedure will find the data structure $\cD$.

\newcommand{\Visize}{O(\delta\log n)}

\begin{enumerate}
    \item Compute the minimum degree of the graph, denoted by $\delta$. Let $V_{low}$ contain all vertices $v$ with $|N_G(v)|=O(\delta)$.
    \item Define distance $\cAd(u,v)$ for every $u,v\in V_{low}$ as $\cAd(u,v)=|N_G(u)\triangle N_G(v)|$. Notice that $\cAd(u,v)$ can be computed in time $\tilde{O}(\delta)$.
    \item Let $\cP\leftarrow$\textsc{CNC}$(G[V_{low}],\cAd,4\ell)$. Let $V_1,...,V_z$ be all vertex sets in $\cC$ for all $\cC\in\cP$, i.e., $z=\sum_{\cC\in\cP}|\cC|$.
    \item Use $\cA_{construct}(G,\ell\log^2n)$ in~\cref{lem:sparserecovery} to get a vector $\vec{v}$ for every $v\in V$.
    \item Data structure $\cD$ memorizes (i) $V_1,...,V_z$ and the sizes of them, (ii) for every $u\in V$, the index set $I_u=\{i\mid u\in V_i\}$, (iii) the vectors $\vec{v}$ for every $v\in V$.
\end{enumerate}

The algorithm returns vertex sets $V,V_1,...,V_z$ and the data structure $\cD$. We will show the query algorithm for $\cD$ later. The following claim is a crucial property.

\begin{claim}\label{cla:nearneighbor}
    For every $u,v\in V_i$, $|N_G(u)\triangle N_G(v)|=O(\ell\log n)$.
\end{claim}
\begin{proof}
    From~\cref{lem:sparseneighborhoodcover}, we have $\cAd(u,v)=O(\ell\log n)$. According to the definition of $\cAd$, we have $|N_G(u)\triangle N_G(v)|=O(\ell\log n)$.
\end{proof}

The following lemma shows the running time of the algorithm.

\begin{lemma}\label{lem:constrcuttime}
    The algorithm for constructing vertex sets and $\cD$ described above runs in time $\tO{m\delta}$.
\end{lemma}
\begin{proof}
    We analyze each step.
    \begin{enumerate}
        \item Computing the minimum degree costs $O(m)$.
        \item The distance oracle uses time $\tilde{O}(\delta)$ for each query.
        \item By~\cref{lem:sparseneighborhoodcover}, the running time is $\tilde{O}(m\delta)$.
        \item By~\cref{lem:sparserecovery}, the running time is $\hat{O}(m)$.
    \end{enumerate}
\end{proof}

\paragraph{Intuition for Kernel Graphs.} Suppose we are given $s \in L, t \in R$. Let $V'$ be one of the clusters such that $L \subseteq V'$.  Our goal is to compute a graph $H$ such that $\kappa_G(s,t) = \kappa_H(s,t)$. At the same time, we will guarantee that $\kappa_G(s,t)\le \kappa_H(s,t)$ always holds even if $s,\in L,t\in R$ and $L\subseteq V'$ do not hold.

WLOG, we assume that $t \not \in N_G[V']$. If $t\in N_G[V']$, we delete all neighbors of $t$ and $t$ itself from $V'$, as they cannot be in $L$, after which we still have $L\subseteq V'$ and the common neighborhood guarantee of $V'$. We define the graph $H$ in the following way. Define $E_G(A,B)=\{(u,v)\in E\mid u\in A,v\in B\}$ for two vertex sets $A,B\subseteq V$.

\begin{quote}
    Let $H'$ be the graph with vertex set $N_G[V']\cup \{t\}$ and edge set $E_G(V',N_G[V'])\cup \{(v,t)\mid v\in N_G(V')\}$, i.e., all edges in $G$ incident to vertices in $V'$ plus new edges from every node in $N_G(V')$ to $t$.

    Let $H$ be obtained from $H'$ as follows. Vertex $s$ and $t$ mark all of their incident edges. Other vertices in $V'\setminus \{s\}$ mark edges not incident to $N_G(s)$. Then, we delete all unmarked edges from $H'$.
\end{quote}

The following claim is crucial to the correctness and shows the reason for constructing $H$ in the above manner.
\begin{claim}\label{cla:L}
    For any vertex set $A\subseteq V'$ with $s\in A$, we have $N_H(A)=N_G(A)$.
\end{claim}

\begin{proof}

    We first show that $N_H(A)=N_{H'}(A)$. It is easy to see that $N_{H}(A)\subseteq N_{H'}(A)$. Now, suppose for contradiction that there exists $v\in N_{H'}(A) \setminus N_{H}(A)$. Thus, there is $u \in A \subseteq V'$ where $u\neq s$ and $(u,v)$ is unmarked. But this happens only when $v \in N_G(s)$. Since $s \in A$, we have $v \in N_{H}(A)$, a contradiction.

    Now we show $N_{H'}(A)=N_G(A)$. This follows from the fact that all edges incident to $A\subseteq V'$ are preserved in $H'$ compared to $G$.
\end{proof}

\paragraph{Correctness.} We need to prove that the $s-t$ minimum vertex cut, denoted by $(L',S',R')$ where $s\in L',t\in R'$ in $H$ has the same size as $(L,S,R)$, which is the minimum vertex cut of $G$ with $s\in L,t\in R$.
We need to prove two directions.

\begin{enumerate}
    \item ($|S|\ge |S'|$) for that we need to prove $S$ is also a $s-t$ vertex cut in $H$. Remember that $L\subseteq V'$ and $N_G(L)=S$. Thus, $N_G(L)\subseteq N_{H}[V']$, which does not contain $t$, is an $s-t$ vertex cut.

    \item ($|S|\le |S'|)$ for that we need to prove $S'$ is also a $s-t$ vertex cut in $G$. Since $t$ has an edge to every vertex in $N_G(V')$, we have $L'\subseteq V'$, otherwise $t\in N_H(L')$. Thus, $S'=N_H(L')=N_G(L')$ is a $s-t$ vertex cut in $G$.
\end{enumerate}

Moreover, the second argument holds even if $L$ is not a subset of $V'$ or $s\in L,t\in R$ does not hold. Thus, we always have $\kappa_G(s,t)\le \kappa_H(s,t)$.

\paragraph{Query algorithm for the data structure $\cD$.} Before stating our algorithm, we first define our \emph{kernel graph} with respect to some $V_i$ and $s\in V_i,t\in V$ as follows.

\newcommand{\Vist}{V^i_{s,t}}
\newcommand{\Eist}{E^i_{s,t}}
\newcommand{\Gist}{G^i_{s,t}}

\begin{definition}\label{def:gist}
    For $i\in [z],s\in V_i,t\in V$, let $V'_i=V_i\setminus N_G[t]$, define $\Gist=(\Vist,\Eist)$ where
    \begin{align*}
    \Vist= & N_G[V'_i]\cup\{t\}\\
        \Eist= & \{(s,v)\mid v\in N_G(s)\} \\
    &\cup\{(u,v)\mid u\in V'_i,v \in N_G(u) \setminus N_G(s)\}\\
    &\cup \{(u,t)\mid u\in N_G(V'_i), u\not=t\}
    \end{align*}

\end{definition}

Note that the construction of $G_{s,t}^i$ is the same as the one described in the overview. We state the following key properties of $G_{s,t}^i$ whose proofs are essentially the same as shown above.

\begin{lemma}\label{lem:atleastkappaG}
     $\kappa_{\Gist}(s,t) \geq \kappa_G(s,t)$.
\end{lemma}

\begin{lemma}\label{lem:whensinLtinR}
    Let $(L,S,R)$ be a minimum vertex cut of $G$. If  $s\in L,t\in R, L\subseteq V_i$, then
    $\kappa_{\Gist}(s,t) = \kappa(G)$.
\end{lemma}

Now we are ready to state our query algorithm for $\cD$.

	\par\addvspace{.5\baselineskip}
	\framebox[16cm][l]{
		\noindent
		\parbox{15.5cm}{
			\par\addvspace{.2\baselineskip}
			\hypertarget{alg:static}{\underline{\emph{Query Algorithm for $\cD$}}}
			\par\addvspace{.3\baselineskip}
			\begin{enumerate}
			    \item Let $I_s$ be all indices $i$ with $s\in V_i$. If $I_s=\emptyset$, \textbf{return} $n$.
                    \item   $I'_s=\{i\in I_s\colon |V_i|\le\Visize\}$. If $I'_s=\emptyset$, \textbf{return} $n$.
                    \item For every $i\in I'_s$, construct a graph $\Gist=(\Vist,\Eist)$ defined in~\cref{def:gist}.

                    \item Let $S'$ be the minimum $(s,t)$-separator in $\Gist$ (which can be done by solving a vertex-capacity max flow on $\Gist$).
                    \item \textbf{return} $|S'|$.

			\end{enumerate}
			\par\addvspace{.2\baselineskip}

		}
	}
	\par\addvspace{.5\baselineskip}

Next, we prove that $\Gist$ can be constructed efficiently.

\begin{lemma}\label{lem:sizeofundirectedkernel}
    For every $i \in I'_s$, $\Gist$ can be constructed in $\widehat{O}(|V_i| \ell) = \widehat{O}( \delta \ell)$ time.
\end{lemma}
\begin{proof}
We can find $V'_i=V_i\backslash N_G[t]$ in time $\tilde O(|V_i|)$ by checking whether each node in $V_i$ is in $N_G[t]$ or not. To find $N_G[V'_i]$, for each $u\in V'_i$, list $N_G(u) \setminus N_G(s)$ and include them in $\Vist$ and also add $(u,v)$ for all $v \in N_G(u) \setminus N_G(s)$ to the set of edges.
To list $N_G(u)\setminus N_G(s)$ efficiently, notice that according to~\cref{cla:nearneighbor}, $N_G(u)\setminus N_G(s)$ is guaranteed to have size $O(\ell\log n)$, and according to~\cref{lem:sparserecovery}, we can use $\cA_{recover}$ to find $N_G(u)\setminus N_G(s)$ in $\widehat{O}(\ell)$ time.
In total, we spent time $\widehat{O}(\ell|V_i|)$ on $\cA_{recover}$.
Finally, add all the edges from $t$ to $N_G(V_{s,t}^i)$. Each step takes $\hO{|V_i|\ell}$ time, so the lemma is proved.
\end{proof}

We are ready to prove \Cref{lem:datastructure}.

\begin{proof}[Proof of~\cref{lem:datastructure}]
    According to~\cref{lem:sparseneighborhoodcover}, each $v\in V$ is contained in at most $O(\log n)$ many $V_i$ for different $i$. The algorithm also returns a data structure $\cD$, according to~\cref{lem:constrcuttime}, it costs $\widehat{O}(m\delta)$ time to construct. Given $(s,t)$, according to~\cref{lem:sizeofundirectedkernel}, $\cD$ takes $ \widehat{O}(\delta\ell)$ time to construct the kernel $G^i_{s,t}$ for any $i\in I'_s$ (recall that $|I'_s|=O(\log n)$ according to~\cref{lem:sparseneighborhoodcover}) and we run vertex-capacity max flow calls on $G^i_{s,t}$. According to~\cref{lem:atleastkappaG}, the returned value of $\cD$ on each query $s,t$ is at least $\kappa(G)$.

    Now suppose $(L,S,R)$ is a minimum vertex cut of $G$ where $|L|\le |R|$ and $\ell\in[|L|,2|L|]$. Notice that if $|L|=\Omega(|S|)$, then the first event happens by taking $V_i=V$. Thus, in what follows we assume $|L|=o(|S|)$. We need the following claim.

    \begin{claim}
        For any $u,v\in L$, $|N_G(u)\triangle N_G(v)|\le 2|L|$.
    \end{claim}
    \begin{proof}
        We have $|N(u)\triangle N(v)|=|N(u)\cup N(v)|-|N(u)\cap N(v)|$, where $|N(u)\cup N(v)|\le |L|+|S|$, and $|N(u)\cap N(v)|=|N(u)|+|N(v)|-|N(u)\cup N(v)|$. Here $|N(u)|\ge |S|$ since $|S|$ is a minimum separator, same for $|N(v)|$. Thus, $|N(u)\cap N(v)|\ge 2|S|-(|L|+|S|)$. We have $|N(u)\triangle N(v)|\le (|L|+|S|)-(2|S|-(|L|+|S|))=2|L|$.
    \end{proof}

    We can assume WLOG that $G[L]$ is connected; otherwise, by taking a connected component of $G[L]$ we get another minimum vertex cut such that $G[L]$ is connected. According to~\cref{lem:sparseneighborhoodcover} property (cover), $L\subseteq V_i$ for some $i\in[z]$.

      We claim that if $|V_i|=\Visize$, then the second event happens. Suppose $\cD$ is given $s\in L,t\in R$, we have $|N_G(s)|=O(|L|+|S|)=O(\delta)$ since $|L|=o(|S|)$ and $|S|\le \delta$. Besides, we have $i\in I_s$ and $|V_i|\le \Visize$. Thus, according to the query algorithm for $\cD$, $\Gist$ is constructed. According to~\cref{lem:whensinLtinR}, $\cD$ will return exactly $|S|=\kappa(G)$.

        Now assume $|V_i|=\omega(\delta\log n)$, we will prove that the first event happens, i.e., $|L|=\tOm{|V_i\cap S|},|V_i\cap R|=\tOm{|V_i\cap S|}$ (remember that $L\subseteq V_i$ is already proved). If $|V_i\cap S|=O(\ell\log n)=O(|L|\log n)$, then we are done since $|V_i\cap R|=|V_i|-|L|-|V_i\cap S|=\Omega(\delta\log n)$ (remember that $\ell=\Theta(|L|)=o(|S|)=o(\delta)$).

        Now we will prove that $|V_i\cap S|=\omega(|L|\log n)$ cannot happen. We use double counting on the size of $E_G(V_i\cap S,V_i\cap R)$. We have
        \[|E_G(V_i\cap S,V_i\cap R)|\le |V_i\cap S|\cdot O(\ell\log n).\]
        In fact, every $u \in V_i \cap S$ has at most $O(\ell\log n)$ neighbors in $V_i\cap R$. This is because $u$ is a neighbor $v \in L$, which has no neighbor in $V_i \cap R$, but the neighborhood of $u$ and $v$ are similar; more precisely,
        $|N_G(u)\triangle N_G(v)|=O(\ell\log n)$ according to~\cref{cla:nearneighbor}.

        On the other hand, consider a node $v\in L$, since $|N_G(v)|\ge |S|$ and $N_G(v)\subseteq L\cup S$, we have $|S\backslash N_G(v)|\le |L|$. Thus, we have $|V_i\cap N_G(v)|=\Theta(|V_i\cap S|)$. Now we get
        \[|E_G(V_i\cap S,V_i\cap R)|\ge \Omega(|V_i\cap R|\cdot|V_i\cap S|)\]
        since $|N_G(u)\triangle N_G(v)|=O(\ell\log n)$ for any $u\in V_i\cap R$ and $|V_i\cap N_G(v)|=\Theta(|V_i\cap S|)=\omega(\ell\log n)$. By combining the inequalities above, we get $|V_i\cap R|=O(\ell\log n)$. Thus, we have $|V_i\cap S|=|V_i|-|V_i\cap L|-|V_i\cap R|=\omega(\delta\log n)$. This is a contradiction since $|S|\le \delta$.
\end{proof}

\subsection{Unbalanced Case: Proof of  \texorpdfstring{\Cref{lem:unbalanced}}{lemma}}
\label{sec:proof of unbal assume ds}

We are now ready to prove \Cref{lem:unbalanced}.
   We describe an algorithm and analysis below.

   \paragraph{Algorithm.}  For every $\ell=2^i$ from $i=1$ to $i=\lceil\log (\delta\log n)\rceil$ (which serves as an estimate of $|L|$), the algorithm uses the following procedure (called the $i$-th iteration) to find the separator $S^{(i)}$ as follows.

    The algorithm firstly uses \cref{thm:crossing family} with parameters $\alpha = \frac{2\delta}{\ell}$ to compute an $(n,\alpha)$-crossing family $P$. Then, it uses   \cref{lem:datastructure} on $G,\ell$ to get (i) $V_1,V_2,...,V_z$, (ii) the data structure $\cD$. Then the algorithm uses~\cref{lem:subgraphbalancedterminalalgorithm} on $G,V_i,\delta\log n$ for every $i\in[z]$ to get separators $S_1,...,S_z$, and uses the data structure $\cD$ with queries $(s,t)\in P$ to get $\tka_{s,t}\ge \kappa(s,t)$ for every $(s,t)\in P$.
     If among $|S_1|,...,|S_z|$ and $\tka_{s,t}$ for any $s,t\in P$, the minimum value is $|S_{j_i}|$, then set $S^{(i)}=S_{j_i}$; otherwise, suppose the minimum value is $\tka_{s,t}$ for some $(s,t)\in P$, we set $S^{(i)}$ to be the minimum separator between $s,t$, denoted by $S_{s,t}$, which can be found by one call to unit-vertex-capacity max flow on the whole graph.

     Finally,  return the minimum $S^{(i)}$ among all $i$ as the output.

    \paragraph{Correctness.} According to~\cref{lem:subgraphbalancedterminalalgorithm}, the returned cuts $S_1,...,S_z$ are always a valid separator. Besides, $S_{s,t}$ must be a valid separator since it is deduced from a vertex-capacity max flow from $s$ to $t$. Thus, we just need to prove that one of $S^{(i)}$ is a minimum separator. Suppose $(L,S,R)$ is one of the minimum separators of $G$ with $|L|\le |R|$ and $|L|,|S|=O(\delta)$. There must exist $i$ such that $\ell=2^i\in[|L|, 2|L|]$. We will prove that $S^{(i)}$ is a minimum cut. In order to prove this, we just need to show that in the $i$-th iteration, either one of $S_1,...,S_z$ is a minimum cut, or one of $\tka_{s,t}$ for $(s,t)\in P$ satisfies $\tka_{s,t}=\kappa(G)$. If this is true, then $S^{(i)}$ is either $|S_j|$ for some $j$ such that $|S_j|=\kappa(G)$, or some $s,t$ minimum separator for some $(s,t)\in P$ with $\tka_{s,t}=\kappa(G)\ge \kappa(s,t)$ (in which case $\kappa(s,t)=\kappa(G)$).

    According to~\cref{lem:datastructure}, there are two events to consider.

    \textbf{Case 1: $\exists i,L\subseteq V_i,|V_i\cap S|=\tilde{O}(|L|),|V_i\cap R|=\tilde{\Omega}(|L|)$.} In this case, since $|V_i\cap L|\le |L|\le \delta\log n$, by using~\cref{lem:subgraphbalancedterminalalgorithm} on $G,V_i,\delta\log n$, we can get $S_i$ which is a minimum separator of $G$.

    \textbf{Case 2: $\cD$ is given $s\in L,t\in R$.} We prove that it will return $\tka=\kappa(G)$. Indeed,

    \begin{lemma}\label{lem:stinP}
        there exists $(s,t)\in P$ such that $s\in L,t\in R$.
    \end{lemma}

    \begin{proof}
        Let $(L,S,R)$ be a vertex mincut in $G$ where $|R| \geq |L|$.
        Since $S$ is a minimum separator, $|S|\le \delta$ (otherwise the neighborhoods of the vertex with a minimum degree should be the minimum separator). We have $|S|/|L|\le (\delta)/(\ell/2) = \alpha$. According to~\cref{thm:crossing family}, since we have $|R| \geq |L| \geq |S|/\alpha$, there exists $(s,t)\in P$ such that $s\in L,t\in R$.
    \end{proof}

    By~\cref{lem:datastructure}, $\cD$ will return $\tka_{s,t}=\kappa(G)$ on such $s\in L,t\in R$.

    \paragraph{Time complexity.} For each iteration, \cref{thm:crossing family,lem:datastructure} use $\tO{m\delta}$ time. Then \cref{lem:subgraphbalancedterminalalgorithm} on each $V_i$ takes time $\hO{\sum_{v\in V_i}|N_{G}(v)|}$, which in total takes time $\hO{m\delta}$ since each vertex in $V$ is included in at most $O(\log n)$ $V_i$'s. One call to $s,t$ vertex-capacity max flow takes time $\hO{m}$. It remains to bound the running time of queries to $\cD$, which is
    \[ |P| \cdot \widehat{O}( \delta \ell) = O(n \alpha) \cdot \widehat{O}( \delta \ell) = \widehat{O}( n \delta^2) = \widehat{O}( m \delta).\]
The first equality follows from \cref{thm:crossing family,lem:datastructure}.

\subsection{Terminal Reduction: Proof of  \texorpdfstring{\Cref{lem:terminalreduction}}{lemma}}
\label{sec:terminal sparse alg}

This section is devoted to proving ~\cref{lem:terminalreduction}. We first restate the omitted bottlenecks in the overview and explain how we solve it.

In the overview, we mentioned that common-neighborhood clustering is not efficient as trivially it needs $\tO{mn}$ time. The algorithm presented in~\cref{sec:sparseneighborhoodcover} is efficient enough as logn as the input graph has degree $\kappa$ which we can guarantee because all nodes with degrees much larger than $\kappa$ must not be in $L$ and can be deleted. However, according to~\cref{lem:sparseneighborhoodcover}(2), it requires $L$ to be connected in $G$. It is not guaranteed that $L_U$ is a connected induced subgraph. Remember the definition of $L_U$: all clusters that are completely inside $L$. $G[L]$ is connected indeed, but $G[L]$ contains three parts (i) $L_U$, (2) $X\cap L$, (3) nodes in other clusters crossing $S$ but not in $L_U$. Deleting the second and third parts could cause $L_U$ to be disconnected.

\paragraph{Our solution.} We try to add (not too many) extra edges to make $L_U$ connected. This is reflected from~\cref{trlineGp} to~\cref{trlineGpp} in the algorithm. The idea is: for every $x\in X\cap L$, we want to make sure that all neighbors of $x$ inside $L_U$ are connected by extra edges. A trivial way to do this is by adding a spanning subgraph on neighbors of $x$.
However, this simple strategy does not solve the problem: assume neighbors of $x$ are $v_1,...,v_k$, after adding a spanning subgraph connecting them (for example, a star rooted at $v_1$), if $v_1$ is not in $L_U$, then the star helps nothing for connecting $L_U$: there are no edges added to the induced subgraph on $L_U$! We solve this issue by the shaving lemma~\cref{lem:shaving}. At a high level, it helps to shave nodes out of the neighbors of $x$, so that the remaining nodes are "roughly in $L_U$" in the sense that they have a small neighborhood symmetric difference just as the nodes in $L_U$, which is enough for us.
For this shaving lemma, think of $A$ as the neighbors of $x$, and $A'$ as the neighbors of $x$ inside $L_U$, we want to find $R$ covering $A'$ such that $R$ itself has a small neighborhood symmetric difference.

\begin{lemma}\label{lem:shaving}
    There is a deterministic algorithm \textsc{Shaving}$(G,A,a)$ that takes as inputs an undirected graph $G=(V,E)$, a vertex set $A$ and a integer $a$, return $R\subseteq A$. Moreover, if there exists $B\subseteq V,A'\subseteq A$ satisfying the following conditions
    \begin{itemize}
        \item $|A'|\ge 0.99|A|$,
        \item for any $u\in A'$ we have $N_{G}(u)\subseteq B$ and $|B-N_G(u)|\le a$,
    \end{itemize}
    then the returned set $R$ further satisfies
    \begin{itemize}
        \item $A'\subseteq R$,
        \item $|N_G(u)\triangle N_G(v)|\le 5a$ for any $u,v\in R$.
    \end{itemize}
    The algorithm runs in time $O(\sum_{u\in A}|N_G(u)|)$.
\end{lemma}

A careful reader might notice that even if we connect all neighbors of $x$ inside $L_U$ for every $x\in X\cap L$, $L_U$ still could be disconnected: remember $L$ has three parts, and the third part (nodes in clusters crossing $S$ but not in $L_U$) is also removed from $L$, which could disconnect $L_U$. The solution is to give a better analysis: if after removing the third part $L$ is disconnected into many components, instead of trying to cover the whole $L_U$ using common-neighborhood clustering, we can try to cover each of the components, which are indeed connected. This does not change the algorithm but puts more challenges on the analysis.

Now we start describing the algorithm.
We first introduce the notion of vertex expanders.

\begin{definition}\label{def:expander}
    We define the \emph{terminal vertex expansion} $h_T$ of a vertex cut $(L, S, R)$ in graph $G$ with respect to the terminal set $T$ as follows:
\[h_T(L, S, R) = \frac{|S|}{\min\{|T\cap (L\cup S)|,|T\cap (R\cup S)|\}}.\]
We say the cut $(L, S, R)$ is $(T,\phi)$\emph{-vertex expanding} in $G$ if $h_T(L, S, R)\geq \phi$ and $(T,\phi)$\emph{-vertex sparse} if $h_T(L, S, R) < \phi$. We say that a graph $G$ is a $(T, \phi)$\emph{-vertex expander} if every vertex cut $(L, S, R)$ such that $\min{(|T\cap (L\cup S)|,|T\cap (R\cup S)|)} > 0$ in the graph $G$ is $(T,\phi)$-vertex expanding.
\end{definition}

We need the following lemma as a subroutine of our algorithm~\cref{alg:terminalreduction}.

\begin{lemma}[Vertex expander decomposition, \cite{LongS22}]\label{lem:vertexexpanderdecomposition}
    There is a deterministic algorithm \textsc{ExpanderDecomposition}$(G,T)$ that takes as inputs an $n$-vertex $m$-edge undirected graph $G=(V,E)$ and a terminal set $T\subseteq V$, return $(X,\mathcal{U})$ where $X\subseteq V$ and $\mathcal{U}$ is a set of subsets of $V$, such that
    \begin{enumerate}
        \item $X$ along with all sets in $\mathcal{U}$ form a partition of $V$, i.e., they are disjoint and the union of them is $V$,
        \item there are no edges between different sets in $\mathcal{U}$, i.e., $E(U_1,U_2)=\emptyset$ for any $U_1,U_2\in \mathcal{U}$ with $U_1\not=U_2$,
        \item $G[U]$ is a $(T,2^{-\log^{0.1} n})$-vertex expander for any $U\in\mathcal{U}$,
        \item $|X|<0.01|T|$.
    \end{enumerate}
    The algorithm runs in time $\hat{O}(m)$.
\end{lemma}

\begin{proof}
    We first calculate $N_G(u)\cap A$ for any $u\in V$. This can be done by the following way in time $O(\sum_{u\in A}|N_G(u)|)$: for every $x\in A$, we increase the count for each $v\in N_G(x)$ by $1$. Now let $V'$ contain all $v$ with $|N_G(v)\cap A|\ge 0.1 |A|$. We claim that $V'\subseteq B$ and $|B-V'|\le 1.2a$. To see $V'\subseteq V$, notice that at most $0.01|A|$ nodes in $A$ have edges to $V-B$. To see $|B-V'|\le a$, notice that for a vertex in $B$ to be not included in $V'$, it needs to have more than $0.9|A|$ non-edges to $A$, which includes at least $0.89$ non-edges $A'$. However, the total number of non-edges from $A'$ to $B'$ is at most $a\cdot |A|$. Thus, there are at most $a\cdot |A|/0.9|A|\le 1.2a$ nodes in $B$ that is not in $V'$. Now we set $R$ to be all nodes $u\in A$ with $|N_G(u)\triangle V'|\le 2.2a$. For every node $u\in A'$, we have $|N_G(u)\triangle V'|\le a+1.2a\le 2.2a$, which means $A'\subseteq R$. For any $u,v\in R$, we have $|N_G(u)\triangle N_G(v)|\le 4.4a\le 5a$.
\end{proof}

To prove~\cref{lem:terminalreduction}, we describe the algorithm, shown in \Cref{alg:terminalreduction}, and analysis.

\newcommand{\Tbig}{T_{\text{big}}}
\newcommand{\Tsmall}{T_{\text{small}}}
\newcommand{\TX}{T_{\bar X}}
\newcommand{\Sstar}{S^*}
\newcommand{\cUsmalllow}{\cU^{\text{small}}_{\text{low}}}
\newcommand{\cUsmall}{\cU_{\text{small}}}
\newcommand{\cUbig}{\cU_{\text{big}}}
\newcommand{\Hgraph}{H}
\newcommand{\Hgraphlow}{H_{low}}
\newcommand{\Hconnected}{H_{connected}}
\begin{algorithm}[H]
    \DontPrintSemicolon
    \caption{\textsc{TerminalReduction}$(G,T,k)$}
    \label{alg:terminalreduction}
    Let $(X,\mathcal{U})\leftarrow$\textsc{ExpanderDecomposition}$(G,T)$ (see~\cref{lem:vertexexpanderdecomposition})\;
    Let $\Sstar{}=V,\Tbig{}=\emptyset,\Tsmall{}=\emptyset$ and $\TX{}$ be arbitrary $\lfloor|T|/100\rfloor$ nodes in $T-X$\;\label{trlinet3}
    Let $\cUsmall{}=\{U\in\cU\mid 0<|T\cap U|<5\},\cUbig{}=\{U\in\cU\mid |T\cap U|\ge 5\}$\;
    \ForEach{$U\in \cUbig{}$}
    {Include arbitrary $\lceil\frac{2}{3}|T\cap U|\rceil$ vertices in $T\cap U$ into $\Tbig{}$\; \label{trlinet1}
    }
    Construct $G'=(X\cup\cUsmall{},E')$ where $E'=\{(x,U)\in X\times \cUsmall{}\mid \exists u\in U,(x,u)\in E\}$\; \label{trlineGp}
    Let $\cUsmalllow{}\leftarrow \{U\in\cUsmall{}\mid |N_{G'}(U)|\le 2k\}$ and $\Hgraphlow{}=G'[X\cup\cUsmalllow{}]$\;

    \ForEach{integer $1\le i\le \lceil\log k\rceil$\label{trlineloopi}}
    {
            Let $a\leftarrow 2^i$\;
            Let $X_{low}\leftarrow \{x\in X\mid |N_{\Hgraphlow{}}(x)|<1000a\}$\;
            \ForEach{$x\in X$}
            {
                Let $R_x\leftarrow$\textsc{Shaving}$(\Hgraphlow{},N_{\Hgraphlow{}}(x),a)$ (see~\cref{lem:shaving})\;\label{trlineshave}
                Let $H_x$ be arbitrary $|R_x|-1$ edges (not in the original graph) that span $R_x$\;\label{trlinepath}
            }
            Build a graph $\Hconnected{}$ with vertex set $\cUsmalllow{}$, and edge set be the union of $H_x$ for all $x\in X$\;\label{trlineGpp}
            Define distance function $\cAd(u,v)=|N_{G'}(u)\triangle N_{G'}(v)|$\;
            Let $\cP\leftarrow$\textsc{CNC}$(\Hconnected{},\cAd(u,v),15a)$ (see~\cref{lem:sparseneighborhoodcover})\;\label{trlinesnc}
            \ForEach{$P\in\cP$}
            {
                Let $X'\leftarrow X$\;
                \ForEach{$C\in P$}
                {
                    For arbitrary $\lfloor\frac{|C|}{\log^2 n}\rfloor$ nodes $U\in C\subseteq \cUsmalllow{}$, include $|T\cap U|$ into $\Tsmall{}$\; \label{trlinet2}
                    \If{$|C|>10a\log^2n$}
                    {
                        \ForEach{$x\in X'\cap N_{G'}(C)$}
                        {Delete $x$ from $X'$ if $|N_{G'}(x)\cap C|>\left(1-\frac{1}{\log^3 n}\right)|C|$\;\label{trlinedelete}
                        }
                    }
                }
                Let $S_1\leftarrow$\textsc{BalancedTerminalVC}$(G,X'\cup \TX{},k)$ (see~\cref{lem:balancedterminalalgorithm})\; \label{trlinex}
                Let $S_2\leftarrow$\textsc{BalancedTerminalVC}$(G,X_{low}\cup \TX{},k)$\;\label{trlinexl}
                If $|S_1|<|\Sstar{}|$, let $\Sstar{}\leftarrow S_1$. If $|S_2|<|\Sstar{}|$, let $\Sstar{}\leftarrow S_2$\;
            }
    }
    Let $T'\leftarrow X\cup \Tbig{}\cup \Tsmall{}\cup \TX{}$

    \Return $(\Sstar{},T')$\;
\end{algorithm}

We now analyze \Cref{alg:terminalreduction}.

    \paragraph{Correctness.} We first show that $|T'|\le 0.9|T|$. Remember that $T'=X\cup \Tbig{}\cup \Tsmall{}\cup \TX{}$. We bound each term as follows.
    \begin{enumerate}
        \item According to~\cref{lem:vertexexpanderdecomposition}, $|X|<0.01|T|$.
        \item According to~\cref{trlinet1}, $\Tbig{}$ includes $\lceil\frac{2}{3}|T\cap U|\rceil$ vertices for every $U$ with $T\cap U\ge 5$, which is at most $4|T\cap U|/5$ nodes. Remember that $U$ are disjoint for different $U\in\cU$, therefore, $|\Tbig{}|\le 4|T|/5$.
        \item According to~\cref{trlinet2}, we include $T\cap U$ into $\Tsmall{}$ for at most $|C|/\log^2n$ number of $U$ for every $C\in\cC$. Since $|T\cap U|<5$, we includes at most $4|C|/\log^2n$ nodes into $\Tsmall{}$ for each $C$. According to~\cref{lem:sparseneighborhoodcover}, $\sum_{C\in\cC}|C|=O(|\cUsmalllow{}|\log n)$. Since every $U\in \cUsmall{}$ contains at least one terminal in $T$, we have $|\Tsmall{}|=O(|T|/\log n)$, which is at most $0.01|T|$ for sufficiently large $n$.
        \item According to~\cref{trlinet3}, we have $|\TX{}|\le 0.01|T|$.
    \end{enumerate}
    In total, $|T'|\le (0.01+0.8+0.01+0.01)|T|\le 0.9|T|$.

    It remains to prove either $\Sstar{}$ is a minimum vertex cut of $G$ or $T'\subseteq V$ is a \strongbalanced{$0.01$} terminal set.
    Let $S$ be an arbitrary vertex cut of $G$ and $L$ be a connected component after deleting $S$ from $G$. Since $T$ is \strongbalanced{$2^{\sqrt{\log n}/2}$} and $k\in[\kappa_G,C\kappa_G]$ for sufficient large constant $C$, according to~\cref{def:balancedterminalset}, we have $|T\cap L|\ge 2^{\sqrt{\log n}/2}k$. Our goal is to prove that if $|T'\cap L|<0.01k$, then $\Sstar{}$ (the output of~\cref{alg:terminalreduction}) must be a minimum vertex cut of $G$. Since $T'=X\cup \Tbig{}\cup \Tsmall{}\cup \TX{}$, we can assume $|X\cap L|<0.01k$, $|\Tbig{}\cap L|<0.01k$ and $|\Tsmall{}\cap L|<0.01k$. Write $T_L=T\cap L,T_S=T\cap S,T_R=T\cap(V-L-S)$. In the following lemmas, we will specify what happens if $|\Tbig{}\cap L|<0.01k$ and $|\Tsmall{}\cap L|<0.01k$, separately.

    \begin{lemma}\label{lem:T1}
        If $|X\cap L|<k/100$ and $|\Tbig{}\cap L|<k/100$, then $|\cup_{U\in\cUsmall{}}(T_L\cap U)|>2^{\sqrt{\log n}/2}k/4$.
    \end{lemma}
    \begin{proof}
        Suppose to the contrary, $|\cup_{U\in\cUsmall{}}(T_L\cap U)|\le 2^{\sqrt{\log n}/2}k/4$. Then we have $|\cup_{U\in\cUbig{}}(T_L\cap U)|\ge 2^{\sqrt{\log n}/2}k/5$ since $|T_L|\ge 2^{\sqrt{\log n}/2}k$ and nodes in $T$ can be either $U\in\cUsmall{}\cup \cUbig{}$ or $X$ where $|X|<0.01|T|$.

        Let $\tilde{\cU}$ contain all $U\in\cUbig{}$ with (i) $|U\cap T_L|\ge |U\cap T_R|$, (ii) $|U\cap T_L|\ge 2|U\cap S|$. We first prove that $|\cup_{U\in\tilde{\cU}}(T_L\cap U)|\ge 2^{\sqrt{\log n}/2}k/10$. For that purpose, we will bound the size of $|\cup_{U\in\cUbig{}-\tilde{\cU}}(T_L\cap U)|$. If (i) is violated, since $G[U]$ is a $(T,2^{-\log^{0.1}n})$-vertex expander according to~\cref{lem:vertexexpanderdecomposition} and according to the definition of expander~\cref{def:expander}, we have $\frac{U\cap S}{|U\cap (T_L\cup T_S)|}\ge 2^{-\log^{0.1}n}$ since $U\cap S$ is a vertex cut of $G[U]$. This gives us $|U\cap T_L|\le 2^{\log^{0.1}n}|U\cap S|$. Thus, the sum of $|U\cap T_L|$ for all $U$ violating (i) is at most $2^{\log^{0.1}n}k$. If (ii) is violated, then $U\cap T_L|<2|U\cap S|$. Thus, the sum of $|U\cap T_L|$ for all $U$ violating (ii) is at most $2k$. We get $|\cup_{U\in\tilde{\cU}}(T_L\cap U)|\ge 2^{\sqrt{\log n}/2}k/5-2^{\log^{0.1}n}k-2k\ge 2^{\sqrt{\log n}/2}k/10$.

        Notice that $U\cap T_L$ contains at least $0.4|U\cap T|$ nodes for $U\in\tilde{\cU}$ according to the above properties (i),(ii). According to~\cref{trlinet1}, we will include at least $2/3|T\cap U|$ many terminals in $T\cap U$ into $\Tbig{}$ for every $U\in\tilde{\cU}$, which contain at least $0.05|T\cap U|$ nodes from $T_L\cap U$. Since $|\cup_{U\in\tilde{\cU}}(T_L\cap U)|\ge 2^{\sqrt{\log n}/2}k/10$, we have $|\Tbig{}\cap L|\ge 0.05\cdot 2^{\sqrt{\log n}/2}k/10>k/100$, a contradiction.
    \end{proof}

    \begin{lemma}\label{lem:T2}
        If $|X\cap L|<k/100$ and $|\Tbig{}\cap L|,|\Tsmall{}\cap L|<k/100$, then one of the following events happens.
        \begin{enumerate}
            \item There exists a loop~\cref{trlineloopi} such that $X_{low}$ in~\cref{trlinexl} satisfies $|X_{low}\cap S|=\tilde{O}(|X_{low}\cap L|)$ and $|\TX{}\cap S|=\tilde{O}(|X_{low}\cap L|)$.
            \item There exists a loop~\cref{trlineloopi} such that $X'$ in~\cref{trlinex} satisfies $|X'\cap S|=\tilde{O}(|X'\cap L|)$ and $|\TX{}\cap S|=\tilde{O}(|X'\cap L|)$.
        \end{enumerate}
    \end{lemma}
    \begin{proof}
        Let $\cU_{inside}\subseteq \cUsmalllow{}$ contain all $U$ such that $U\subseteq L$. Let $\cU_{cross}\subseteq \cUsmalllow{}$ contain all $U$ such that $U\cap (L\cup S)\not=\emptyset$ and $U\not\subseteq L$. Notice that for any $x\in X\cap L$, we have $N_{\Hgraphlow{}}(x)\subseteq \cU_{inside}\cup\cU_{cross}$.
        \begin{claim}
            $|\cU_{inside}|>k\log^9n$.
        \end{claim}
        \begin{proof}
            For any $U\in\cUsmall{}$, if $U\subseteq L$, then $|N_{G'}(U)|\le |X\cap L|+|S|\le 1.01k$ (that is because $S$ is a vertex cut), which implies $U\in\cU_{low}$. Also notice that for any $U\in\cUsmall{}$ with $T_L\cap U\not=\emptyset$, either $U\in\cU_{inside}$ or $U\cap S\not=\emptyset$ where the latter case can only happen for at most $k$ different $U$ as they are disjoint and $|S|\le k$. Each $U$ contain at most $4$ terminals in $T$, and due to~\cref{lem:T1}, we have $|\cup_{U\in\cUsmall{}}(T_L\cap U)|>2^{\sqrt{\log n}/2}k/4$, which means terminals inside $\cU_{inside}$ is at least $2^{\sqrt{\log n}/2}k/4-4k$, and $|\cU_{inside}|>2^{\sqrt{\log n}/2}k/16-k>k\log^9n$.
        \end{proof}

        We create $O(\log n)$ buckets for elements $U\in\cU_{inside}$ according to $|N_{G'}(U)\cap (X\cap L)|$. Let $\cU^j_{inside}$ contain all $U\in\cU_{inside}$ with $|N_{G'}(U)\cap (X\cap L)|\in [2^j,2^{j+1}]$. There exists $j$ such that $|\cU^j_{inside}|>k\log^8n$. Let $j$ be such an index and write $\cU^*=\cU^j_{inside}$. Consider the loop in~\cref{trlineloopi} when $2^i\in [2^{j-2},2^{j-1}]$. We write $a^*=2^i$.

        \begin{claim}\label{cla:xlow}
            In the loop $i$ described above, $|X_{low}\cap S|\le 16a^*$.
        \end{claim}
        \begin{proof}
            For every $U\in\cU^*$, we have $N_{G'}(U)\ge |S|$ and $|N_{G'}(U)\cap (X\cap L)|\le 8a^*$ and $N_{G'}(U)\subseteq (X\cap L)\cup(X\cap S)$. Thus, $|N_{G'}(U)\cap S|\ge |S|-8a^*$, in other words, $U$ has at most $8a^*$ non-edges to $S$. If $x\in X_{low}\cap S$, then it has at most $1000a^*$ edges to $\cU^*$, which is at most $1000k$. The number of non-edges from $x$ to $\cU^*$ is at least $|\cU^*|-1000k\ge |\cU^*|/2$. Thus, $|X_{low}\cap S|\le (8a^*|\cU^*|)/(|\cU^*/2)=16a^*$.
        \end{proof}
        We also need the following useful claim.
        \begin{claim}\label{cla:Ucross}
            $|U_{cross}|\le 8a^*$ and $|\TX{}|\le 8a^*$.
        \end{claim}
        \begin{proof}
            Take an arbitrary $U\in\cU^*$, we have $N_{G'}(U)\ge |S|$ and $N_{G'}(U)\cap L\le 8a^*$. Notice that $N_{G'}(U)\subseteq (L\cup S)\cap X$. Thus, $|N_{G'}(U)\cap S|\ge |S|-8a^*$, which implies $S-X\le 8a^*$. Each $U\in U_{cross}$ must contain at least one element in $S$, and they are disjoint. Remember that $\TX{}\subseteq T-X$. Therefore, $|U_{cross}|\le 8a^*$ and $|\TX{}|\le 8a^*$.
        \end{proof}
        Now if $|X_{low}\cap L|>0.01a^*$, according to~\cref{cla:xlow} and \cref{cla:Ucross}, we have $|X_{low}\cap S|=\tilde{O}(|X_{low}\cap L|)$ and $|\TX{}\cap S|=\tilde{O}(|X_{low}\cap L|)$, which completes the proof. In what follows we assume $|X_{low}\cap L|\le 0.01a^*$. Write $X_{highL}=(X-X_{low})\cap L$ and $X_{lowL}=X_{low}\cap L$.

        We first define a graph on vertex set $X_{highL}$ where $x_1,x_2\in X_{highL}$ have an edge connecting them if $N_{G'}(x_1)\cap N_{G'}(x_2)\cap \cU_{inside}\not=\emptyset$. Suppose this graph has $\ell$ connected components denoted by $X_1,...,X_{\ell}$. We write $\cU_q=N_{\Hgraphlow{}}(X_q)\subseteq\cU_{inside}\cup\cU_{cross}$ and $\cU^{inside}_q=\cU_q\cap\cU_{inside}$.
        \begin{claim}\label{cla:cover}
            For any $q\in[\ell]$ with $|X_q|\ge a^*$, one of the loop of~\cref{trlineloopi} when $2^i$ (which we denote by $a$) satisfies $2|X_q|\le a< 4|X_q|$ (which must happen for exactly one $i$), there exists a partition $P\subseteq \cP$ (returned in~\cref{trlinesnc}) and a cluster $C\in P$ such that $\cU^{inside}_q\subseteq C$. Denote such partition as $P_q$ and such cluster as $C_q$.
        \end{claim}
        \begin{proof}
            Remember $R_x$ is defined in~\cref{trlineshave} for every $x\in X$. We first prove that for any $x\in X_q$, the condition described in~\cref{lem:shaving} is satisfied if we let $B=X_{lowL}\cup X_q\cup(X\cap S)$ and $A'=N_{\Hgraphlow{}}(x)\cap\cU_{inside}$. Notice that $A=N_{\Hgraphlow{}}(x)$ and $A-A'\subseteq \cU_{cross}$. According to~\cref{cla:Ucross}, we have $|A|-|A'|\le 8a^*$. Remember that $x\in X_q\subseteq X_{highL}$. According to the definition of $X_{highL}$, we have $N_{\Hgraphlow{}}(x)\ge 1000a^*$ (recall that the \emph{low}, \emph{high} is defined with respect to $a^*$), thus, $|A'|\ge 0.99|A|$. For any $U\in A'=N_{\Hgraphlow{}}(x)\cap\cU_{inside}$, we have $N_{G'}(U)\cap L\subseteq X_{lowL}\cup X_q$ since $X_q$ is a connected component where all neighbors of $N_{G'}(U)\cap X_{highL}$ should be in the same connected component $X_q$. Thus, $N_{G'}(U)\subseteq X_{lowL}\cup X_q\cup (X\cap S)=B$. Besides, we have $|N_{G'}(U)|\ge |S|$, which means $|B-N_{G'}(u)|\le |X_{lowL}|+|X_q|+|S|-|S|\le 0.01a^*+0.5a\le 0.51a$. Thus, all conditions are satisfied~\cref{trlineshave}. The conclusion we get is: $R_x$ satisfies $N_{\Hgraphlow{}}(x)\cap\cU_{inside}\subseteq R_x$ and $|N_{G'}(u)\triangle N_{G'}(v)|\le 5a$ for any $u,v\in R_x$.

            Write $R_q=\cup_{x\in X_q}R_x$. We will prove that (i) $\cU^{inside}_q\subseteq R_q$, (ii) $\Hconnected{}[R_q]$ is connected, (iii) $|N_{G'}(u)\triangle N_{G'}(v)|\le 15a$ for any $u,v\in R_q$. If all of them are correct, then according to~\cref{lem:sparseneighborhoodcover} property (cover), there exists $C\in P\in \cP$ such that $\cU^{inside}_q\subseteq R_q\subseteq C$ and the proof is finished.

            \paragraph{(i) $\cU^{inside}_q\subseteq R_q$.} This follows from $N_{\Hgraphlow{}}(x)\cap\cU_{inside}\subseteq R_x$ for any $x\in X_q$ and $\cU^{inside}_q=\cU_q\cap\cU_{inside}$.

            \paragraph{(ii) $\Hconnected{}[R_q]$ is connected.} Remember that in $\Hconnected{}$, there is a path on $R_x$ for every $x\in X_q$. Recall that $X_q$ is connected in the sense that $x,y\in X_q$ have an edge if they share a neighbor in $\cU_{inside}$, in which case the path $R_x,R_y$ are connected together since $N_{\Hgraphlow{}}(x)\cap\cU_{inside}\subseteq R_x$ for any $x\in X_q$. Therefore, $\Hconnected{}[R_q]$ (which is the union of paths $R_x$ for $x\in X_q$) is connected.

            \paragraph{(iii)$|N_{G'}(u)\triangle N_{G'}(v)|\le 15a$ for any $u,v\in R_q$.} For any $u\in\cU^{inside}_{q}$, we have $N_{G'}(u)\subseteq X_{lowL}\cup X_q\cup S$ and $|N_{G'}(u)|\ge |S|$, which implies $|N_{G'}(u)\triangle N_{G'}(v)|\le 2|X_{lowL}\cup X_q\cup S|-|S|-|S|\le 1.02a$. For any $w\in R_q-\cU^{inside}_q$, there exists $u\in \cU^{inside}_q$ such that they are in the same path $R_x$ for some $x\in X_q$ and $|N_{G'}(u)\triangle N_{G'}(w)|\le 5a$. Thus, according to the triangle inequality for symmetric difference, we have $|N_{G'}(u)\triangle N_{G'}(v)|\le 15a$ for any $u,v\in R_q$.
        \end{proof}

        Now we have $P_q$ and $C_q$ for every $q\in[\ell]$ with $|X_q|\ge a^*$ and $\cU^{inside}_q\subseteq C_q$, we have the following claim.
        \begin{claim}
            There exists $q\in[\ell]$ with $|X_q|\ge a^*$ and $|C_q|\ge (\log^4n)|X_q|$ where $|\cU^{inside}_q|< (1-1/2\log^2 n)|C_q|$.
        \end{claim}
        \begin{proof}
            If to the contrary, $|\cU^{inside}_q|\ge (1-1/2\log^2 n)|C_q|$ holds for any of them, then according to~\cref{trlinet2}, at least $|\cU^{inside}_q|/2\log^2 n$ nodes in $T\cap L$ is included in $\Tsmall{}$. If we can prove $\sum_{|X_q|\ge a^*}|\cU^{inside}_q|\ge k\log^5n$, then $\sum_{|X_q|\ge a^*,|C_q|\ge(\log^4n)|X_q|}|\cU^{inside}_q|\ge k\log^5n-\sum_{q}(\log^4n)|X_q|\ge k\log^5n-0.01k\log^4n$, violating the fact that $|\Tsmall{}|<0.01k$. Remember that $|\cU^*|\ge k\log^8n$ where $\cU^*$ only contain $U\in\cU_{inside}$ with $|N_{G'}(U)\cap (X\cap L)|\ge 2a^*$. Any $U\in\cU^*$ satisfies $N_{G'}(U)\cap (X\cap L)\subseteq X_{lowL}\cup X_{highL}$ where the part in $X_{lowL}$ has size at most $0.01a^*$, and the part in $X_{highL}$ must be in the same $X_q$ for some $q$, implying $|X_q|\ge 1.99a^*$. Thus, $\sum_{|X_q|\ge a^*}|\cU^{inside}_q|\ge k\log^5n$.
        \end{proof}

        Now we have $q\in[\ell]$ with $|X_q|\ge a^*$ and $|C_q|\ge (\log^4n)|X_q|$ where $|\cU^{inside}_q|< (1-1/2\log^2 n)|C_q|$. We will prove that $X'$ (in~\cref{trlinex}) satisfies $|X'\cap S|=\tilde{O}(|X'\cap L|)$ in this specific loop when $P=P_q$, which finishes the proof. We first show that $X_q\subseteq X'$, i.e.,~\cref{trlinedelete} does not delete any vertex in $X_q$. This is because for any $x\in X_q$, $N_{G'}(x)\subseteq \cU^{inside}_q\cup \cU_{cross}$ where $|\cU^{inside}_q|< (1-1/2\log^2 n)|C_q|$ and $|\cU_{cross}|\le 8a^*\le 8|X_q|\le |C_q|/\log^4n$. It remains to prove $|X'\cap S|=\tilde{O}(|X_q|)=O(|X'\cap L|)$. Let $u\in\cU^{inside}_q$ be an arbitrary element. We have $|N_{G'}(u)\cap S|\ge |S|-|X_{lowL}|-|X_q|\ge |S|-0.1a$ (remember that $|X_q|\ge a/4$ and $a\ge a^*$). Notice that for any $v\in C_q$, we have $|N_{G'}(u)\triangle N_{G'}(v)|=O(a\log n)$ according to~\cref{lem:sparseneighborhoodcover} property (common-neighborhood). Thus, we have $|N_{G'}(v)\cap S|\ge |S|-O(a\log n)$ for any $v\in C_q$. For a vertex $s\in S$ to not be deleted, it needs to have at least $|C_q|/\log^3n$ many non-edges to $C_q$. However, there are at most $O(a\log n)|C_q|$ many non-edges from $C_q$ to $S$. Thus, the number of nodes in $S$ that is not deleted is at most $O(a\log n)|C_q|/(|C_q|/\log^3n)=\tilde{O}(a)=\tilde{O}(|X_q|)=O(|X'\cap L|)$. Besides, according to~\cref{cla:Ucross}, we have $|\TX{}\cap S|=O(a^*)=O(|X'\cap L|)$, which finishes the proof.
    \end{proof}

    With~\cref{lem:T2}, further assume $|\TX{}\cap L|<k/100$ (otherwise we are done), the $S_1$ in~\cref{trlinex} or $S_2$ in~\cref{trlinexl} will be a minimum vertex cut of $G$: $|(X'\cup \TX{})\cap S|=\tilde{O}(X'\cap L)$ holds due to~\cref{lem:T2}, $|(X'\cup \TX{})\cap L|\le 0.02k$ due to $|X\cap L|<0.01k$ and $|\TX{}\cap L|<0.01k$, $|\TX{}\cap (V-L-S)|>k$ due to $|T-X|>(\log n)k$ and $|\TX{}|\ge |T|/101$, similar for $X_{low}$.

    \paragraph{Complexity.} One can verify that most lines of the algorithm can be run in time $\hat{O}(mk)$. We will explain some non-trivial parts as follows.
    \begin{enumerate}
        \item \textbf{Construction of auxiliary graphs.} $G'$ is constructed by first deleting all vertex set in $\cUbig{}$ from $G$, and then contract every vertex set in $\cUsmall{}$. $\Hgraphlow{}$ is an induced subgraph of $G'$.

        \item \textbf{Call of Shaving.} In~\cref{trlineshave}, the total running time according to~\cref{lem:shaving} is $\sum_{x\in X}\sum_{U\in N_{\Hgraphlow{}}(x)}|N_{G'}(U)|$. Here $|N_{G'}(U)|\le 2k$ since $U\in\cUsmalllow{}$. Thus, it is at most $k\sum_{x\in X}|N_{\Hgraphlow{}}(x)|=O(km)$ since the number of edges in $\Hgraphlow{}$ is at most $m$.

        \item \textbf{Size of $\Hconnected{}$ and call of CNC.} $\Hconnected{}$ is a union of path $H_x$, where $H_x$ is a path with length at most $|N_{\Hgraphlow{}}(x)|$. Thus, the number of edges in $\Hconnected{}$ is at most $m$. Moreover, the maximum degree is at most $4k$ since for every $u\in\cUsmalllow{}$, we have $|N_{G'}(u)|\le 2k$, where each $x\in N_{G'}(u)$ will add at most $2$ edges to $u$ in $\Hconnected{}$.
        The distance oracle $\cA_{dist}(u,v)$ for any $u,v\in \cUsmalllow{}$ (which is the vertex set of $\Hconnected{}$) can be computed in time $O(k)$ since $|N_{G'}(u)|,|N_{G'}(u)|\le 2k$ according to the definition of $\cUsmalllow{}$. Therefore, according to~\cref{lem:sparseneighborhoodcover}, the running time for the \textsc{CNC} call is at most $\tilde{O}(mk)$.

        \item \textbf{Delete vertex from $X'$.} The size $\cP$ is at most $O(\log n)$ according to~\cref{lem:sparseneighborhoodcover}, where each $P\in\cP$ is a partition. Now we explain the total running time of~\cref{trlinedelete}. To implement it for a specific $P\in\cP$, we first iterate over all $x\in X$ and $u\in N_{\Hgraphlow{}}(x)$, add a count $1$ to the $C\in P$ for $u\in C$, which cost $O(m)$ time. After that, we get $|N_{\Hgraphlow{}}(x)\cap C|$ for every $x\in X$ and $C\in P$. For each $C$, we can find $N_{G'}(C)$ in time the sum of the degree of every node in $C$, and decide whether to delete each node or not in $O(1)$ time. Since $P$ is a partition, it costs time $O(m)$.
    \end{enumerate}

\section{Construction of Pseudorandom Objects}

In this section, we discuss the construction of the crossing family (\Cref{thm:crossing-family-unified,thm:crossing family}) and the selectors (\Cref{thm:selectoreasy}).

\subsection{Crossing Family}
\label{sec:pseudo random objects}

The goal of this section is to prove \Cref{thm:crossing-family-unified} below. Recall the definition of the crossing family from \Cref{def:ablr-crossingfamily}.

\asymcrossingfamilythm*

The degree of an element $x \in A$ in a crossing family, $\cal{P}$ is $|\{ y \in B \mid (x,y) \in \cal{P} \}|$.

The Crossing Family uses an explicit construction of a pseudorandom object named disperser defined below, and the theorem regarding the explicit construction of disperser follows.

\begin{definition}[$(k,\eps)$-Disperser \cite{TUZ01}] A bipartite graph $G=(V,W,E)$ is a $(k,\eps)$-disperser if every subset $A\subseteq V$ of size at least $k$ has at least $(1-\eps)|W|$ distinct neighbors in $W$. A disperser is explicit if, for every vertex $v \in V$ and $i$, the $i^{th}$ neighbor of $v$ can be computed in $\tO{1}$ time.
\end{definition}

\begin{theorem}[Paraphrase of Theorem 1.4 from \cite{TUZ01}]
    \label{thm:disperser}
    For every integer $n$ and $k < n$ that are powers of $2$, and a constant $\eps \in (0,1)$, there is an explicit $(k,\eps)$-disperser $G=(V=[n],W,E)$ with left-degree $d=(\log n)^{(\cTUZ)}$ and we have $|W|\in [c_{\text{low}}kd/\log^3 n, c_{\text{high}}kd]$ for some universal constants $\cTUZ,c_{\text{low}},c_{\text{high}}$.
\end{theorem}

We extend the construction of disperser in \Cref{thm:disperser} to \Cref{lem:generalDisperser} so that, the inputs $n,k$ can be arbitrary integers and not just powers of $2$ and the left-degree of the disperser can be amplified arbitrarily to any given $d > (\log n)^{(\cTUZ)}$.

\begin{restatable}{lemma}{generaldisperserlem}
    \label{lem:generalDisperser}
    Given three positive integers $n,k,d$ with $k < n$, $0<d<n$ and a constant $\eps\in (0,1)$, we can construct an explicit $(k,\eps)$-disperser $D=(V,W,E)$ with $|V|=n$ and $c_{\text{low}}\frac{k \cdot d'}{\log^3 n} \leq |W| \leq c_{\text{high}}k\cdot d'$ where $d' = \Theta(\max(d,(\log n)^{(\cTUZ)}))$ is the left-degree of the graph $G$.
\end{restatable}

The proof of \Cref{lem:generalDisperser} is mostly by adjusting parameters appropriately and is deferred to the \Cref{sec:omittedproofs}. Given \Cref{lem:generalDisperser}, we first prove a relaxed version of \Cref{thm:crossing-family-unified} below.

\begin{lemma}
    \label{thm:ablr2}
    There exists a deterministic algorithm that, given two finite sets $A,B$ and two integers $l,r$ with $l\leq |A|, r\leq |B|$ and $l\leq r$ outputs an $(A,B,l,r)$-crossing family in $\tO{|A|\cdot|B|/l}$ time. The degree of every element in $A$ is at most $\tO{|B|/l}$.
\end{lemma}

\begin{proof}
    (\textbf{Construction}) We first construct two dispersers. Let $D_1=(V_1=B,W_1,E_1)$ be the disperser constructed using \Cref{lem:generalDisperser} with $n=|B|,k=r,d=(\log n)^{\cTUZ},\eps = 1/8$ as inputs and $D_2=(V_2=A,W_2,E_2)$ be the disperser constructed using \Cref{lem:generalDisperser} with $n=|A|,k=l,d=(r/l)\cdot (\log n)^{\cTUZ},\eps = 1/8$ as inputs.

    By \Cref{lem:generalDisperser} we have the guarantee that $|W_{1}|\in[r\cdot(\log n)^{(\cTUZ-3)},r\cdot(\log n)^{(\cTUZ)}]$ since $d'=(\log n)^{(\cTUZ)}$. Similarly, we have $d'=(r/l)\cdot (\log n)^{\cTUZ}$ for $D_2$ as $r\geq l$ thus $|W_{2}|\in[r\cdot(\log n)^{(\cTUZ-3)},r\cdot(\log n)^{(\cTUZ)}]$ as $kd'=l\cdot (r/l)\cdot(\log n)^{(\cTUZ)} = r\cdot(\log n)^{(\cTUZ)}$.

    Note that although the sizes of $W_1,W_2$ are in the same range they need not be equal. Assume $W_1\geq W_2$ (other case is similar) and $\beta=\floor{|W_1|/|W_2|} \in [1,\log^3 n]$. Each vertex in $W_2$ is duplicated $\beta$ times and then choose $|W_1|-\beta|W_2|$ arbitrary vertices in $W_2$ and duplicate one more time to reach the size $|W_1|$ exactly. Hence every vertex is duplicated at least $\beta$ times and at most $\beta+1$ times. Since the set with duplicates of $W_2$ has size exactly $W_1$, we can arbitrarily map both sets and without loss of generality assume they are same say $W$. We define $D'_2 = (V_2=A,W,E'_2)$ where $E'_2$ is obtained by duplicating incident edges when duplicating vertices in $W_2$.

    Any subset of size at least $l$ in $D_2$ has at least $(1-\eps)|W_2|$ neighbors in $W_2$, so the number of non-neighbors are at most $\eps |W_2|$. So the number of non-neighbors of the same set of size at least $l$ in $D'_2$ are at most $\eps(\beta+1)|W_2|$ and the fraction of non-neighbors in the $W$ is $\eps(\beta+1)|W_2|/|W| \leq \eps(1+1/\beta) \leq 2\eps$. The first inequality follows as $|W|\geq \beta|W_2|$, second inequality follows due to $\beta \geq 1$. Thus $D'_2$ is $(l,2\eps)$-disperser.

    We now bound the right-degree of $D_1$. The total number of edges in $D_1$ is $|B|\cdot (\log n)^{(\cTUZ)}$. The average degree of the right vertices is $|B|\cdot (\log n)^{(\cTUZ)}/|W|$. Hence, by an averaging argument there cannot be more than $|W|/4$ many vertices of degree greater $4|B|\cdot (\log n)^{(\cTUZ)}/|W|$. Therefore, at least $3|W|/4$ many vertices in $W$ have right-degree at most $4|B|\cdot (\log n)^{(\cTUZ)}/|W| \leq \tO{|B|/r}$ we call this set $W_{\text{small}}$.

    Let $N_{D}(v)$ be the set of neighbors of vertex $v$ in graph $D$. We extend this definition to sets where $N_D(S)$ is the union of set of neighbors of vertices present in $S$. Recall $D_1=(B,W,E_1)$ and $D'_2 = (A,W,E'_2)$. We define the $(A,B,l,r)$-crossing family $\mathcal{P}$ as
    \[ \mathcal{P} = \{(u,v) \mid u\in A, v\in N_{D_1}(N_{D'_2}(u)\cap W_{\text{small}})\} \]

    (\textbf{Running time and degree bound}) For each vertex $u\in A$ we have $|N_{D'_2}(u)|\leq r(\beta+1)(\log n)^{(\cTUZ)}/l$. Since we are only considering neighbors of $u$ that are also in $W_{\text{small}}$, the size of neighborhood of $N_{D'_2}(u)\cap W_{\text{small}}$ in $D_1$ is $|N_{D_1}(N_{D_2}(u)\cap W_{\text{small}})| \leq \frac{r(\beta+1)(\log n)^{(\cTUZ)}}{l}\cdot \tO{|B|/r} = \tO{|B|/l}$ which is the out-degree of $u\in A$ in $\cP$. Since each edge takes $\tO{1}$ time to compute, it takes $\tO{|A|\cdot |B|/l}$ time to compute $\cP$.

    (\textbf{Correctness}) For all $L \subset A$, $R\subset B$ which are of size at least $l,r$ respectively, we need to show that $\cP\cap (L\times R)\neq \emptyset$.
    Since $D'_2$ is a $(l,2\eps)$-disperser we have $|N_{D'_2}(L)| \geq (1-2\eps)|W| = 3|W|/4$ as $\eps = 1/8$. Similarly, $|N_{D_1}(R)| \geq (1-\eps)|W| = 7|W|/8$ as $D_1$ is a $(r,\eps)$-disperser. Since $|W_{\text{small}}| \geq 3|W|/4$, we have $N_{D'_2}(L)\cap W_{\text{small}} \cap N_{D_1}(R) \neq \emptyset$ and let $w$ be any vertex in this intersection. We conclude $\mathcal{P}\cap (L\times R) \neq \emptyset$ because $w\in N_{D'_2}(u)\cap W_{\text{small}}$ for some $u\in L$ which also has an edge to some vertex $v\in R$ as $w\in N_{D_1}(R)$ which concludes $(u,v)\in \cP$.
\end{proof}

We can extend the lemma above to get a better bound when $r \geq |B|/2$ using a simple trick and prove \Cref{thm:crossing-family-unified}.

\begin{proof}[Proof of \Cref{thm:crossing-family-unified}]
    If $r\leq |B|/2$, we know that $|B|-r \geq |B|/2$ so we apply \Cref{thm:ablr2} with $A,B,l,r$ as inputs to get crossing family $\mathcal{P}$ of degree at most $\tO{|B|/l} \leq \tO{(|B|-r)/l}$.

    If $r > |B|/2$ let $B'\subset B$ be an arbitrary subset of size exactly $2(|B|-r)$. Let $r'= |B|-r$ and $l'=\min(|B|-r,l)$, we have $r'\geq l'$, thus we can apply the \Cref{thm:ablr2} with $A,B',l',r'$ as inputs to get crossing family $\mathcal{P}$ of degree at most $\tO{(|B|-r)/l'} = \max(1,(|B|-r)/l)$ as required.

    We will now prove that $\mathcal{P}$ is also an $A, B, l, r$ crossing family. For any subset $L\subset A, R\subset B$ of size at least $l,r$ respectively, let $R'=R\cap B'$. As $R\geq r$ we have the guarantee that $|R'| = |R\cap B'| \geq |B|-r = r'$. We also have $|L|\ge l \ge l'$. Thus from \Cref{thm:ablr2} we have $\mathcal{P}\cap(L\times R') \neq \emptyset$ which implies $\mathcal{P}\cap(L\times R) \neq \emptyset$ as $R'\subset R$.
\end{proof}

Given \Cref{thm:crossing-family-unified},  \Cref{thm:crossing family} follow immediately.

\symcrossingfamily*

\begin{proof}    \emph{Algorithm:} Guess the size of $|L|,|R|$ by powers of $2$, denoted by $l,r$.
    For each guess $l,r$ satisfying (1) $l\le r$, (2) $l+r<n$, and (3) $(n-r)/l\le\alpha+1$, we run \cref{thm:crossing-family-unified} with parameters $A=V,B=V$ and the corresponding $l,r$ and take the union of all of those crossing families as $\cP$.

    \emph{Running time and degree bound:} There can be at most $\log^2 n$ such crossing families, each taking atmost $\tO{|A|\cdot (|B|-r)/l} = \tO{n(n-r)/l} = \tO{n\alpha}$ time to construct. The total running time is still $\tO{n\alpha}$. The out-degree of every vertex in each crossing family is bounded by $\tO{(n-r)/l} = \tO{\alpha}$ based on condition (3). Hence the out-degree of every element in $A$ is at most $\log^2  n\cdot \tO{\alpha} = \tO{\alpha}$.

    \emph{Correctness:} For every partition $(L,S,R)$ of $[n]$ that satisfies $|R|\ge |L|\ge |S|/\alpha$ let $l,r$ be the highest powers of $2$ that are at most $|L|,|R|$. It is easy to verify that all the above three conditions are satisfied. Thus the crossing family constructed with this $l,r$ as inputs has non-empty intersection with $L\times R$, so does the union $\cP$.
\end{proof}

\subsection{Selectors}\label{sec:selector construction}

We mainly prove \Cref{thm:selectoreasy} in this sub-section, which follows from \Cref{thm:selector} stated below. Recall the definition of the selector.

\selectordefn*

Let $T\in \cS$ be the set such that $T\cap L = \{x\}$ and $T\cap S = \emptyset$ then we say that $T$ \textit{selects} $x$ from $L$ avoiding $S$.

\begin{restatable}{theorem}{selectorthm}
    \label{thm:selector}
    There exists an algorithm that, given integers $n,k$ and parameters $\eps \in [0,1],\alpha \in (1/n,1)$ where $k \le (\frac{n}{2})^{\frac{1}{1+\alpha}}\cdot (\frac{\eps}{\log^2 n})^{\frac{2}{\alpha}}$, return a $(n,k,\epsilon)$-selector $\cS$ of size at most $|\cS| \leq k^{1+\alpha}\cdot (\frac{\log n \log k}{\eps})^{3+3/\alpha}$ in $\tO{n}\cdot (\frac{\log n \log k}{\eps})^{1+1/\alpha}$ running time and every set $T \in \cS$ has size at least $n/(k^{1+\alpha}\cdot (\frac{\log n \log k}{\eps})^{2+2/\alpha})$.
\end{restatable}

We give the proof of \Cref{thm:selectoreasy} (recalled below) assuming \Cref{thm:selector}.

\selectoreasy*

\begin{proof}
    Let $\psi(n)=\log(1/\eps)/\log n = o(1)$ as $\eps \geq 1/n^{o(1)}$ and $\alpha = \max(\sqrt{\psi(n)},1/\sqrt{\log n})$. Let $\tilde{k}=n^{1-f(n)}\geq n^{1-o(1)}$ with $f(n)= O(\alpha + \frac{\psi(n)}{\alpha}+\frac{\log \log n}{\alpha\log n}) = o(1)$ based on our choice of parameters. Our choice of $f(n)$ satisfies the premise of \Cref{thm:selector} that is $\tilde{k} = n^{1-f(n)} \leq (\frac{n}{2})^{\frac{1}{1+\alpha}}\cdot (\frac{\eps}{\log^2 n})^{\frac{2}{\alpha}}$ as $k\le \tilde{k}$. Apply \Cref{thm:selector} with $n,k,\eps,\alpha$ as input and return the output $(n,k,\eps)$ -selector $\cS$.

    Based on our choice of parameters we have
    \begin{align*}
    \label{eq:upperbound}
        (\frac{\log n \log k}{\eps})^{1+1/\alpha} &\leq (\frac{\log^2 n}{\eps})^{2/\alpha}\\
        &\leq n^{(\psi(n)+\frac{2\log\log n}{\log n})(2/\alpha)}\\
        &\leq n^{o(1)} & (\alpha = \max(\sqrt{\psi(n)},1/\sqrt{\log n})) \tag{1}
    \end{align*}

    From \Cref{thm:selector} we have
    \[|\cS| \leq k^{1+\alpha}\cdot (\frac{\log n \log k}{\eps})^{3+3/\alpha} \leq \frac{n}{2}\cdot (\frac{\eps}{\log^2 n})^{1+1/\alpha}\cdot  (\frac{\log n \log k}{\eps})^{3+3/\alpha} \leq k\cdot n^{f(n)}\cdot n^{o(1)}\leq  k\cdot n^{o(1)},\]
    the second inequality follows from premise, the third inequality follows as $k=n^{1-f(n)}$ and \Cref{eq:upperbound}. The time taken is $\tO{n}\cdot (\frac{\log n \log k}{\eps})^{1+1/\alpha} \leq n^{1+o(1)}$. The size of every set $T\in \cS$ is at least $n/m \geq n/(k^{1+\alpha}\cdot (\frac{\log n \log k}{\eps})^{2+2/\alpha}) \geq 2$ as $k\leq \tilde{k} \leq (\frac{n}{2})^{\frac{1}{1+\alpha}}\cdot (\frac{\eps}{\log^2 n})^{\frac{2}{\alpha}}$.
\end{proof}

We construct the selector family from a lossless expander, which is defined as follows:

\begin{definition}[$(k,\alpha)$-lossless expander]
    A $d$-left regular bipartite graph $G=(V_L,V_R, \Gamma)$ with left vertex set $V_L$, right vertex set $V_R$ and the neighborhood function $\Gamma : V_L\times [d] \rightarrow V_R$ where $\Gamma_i(v)\coloneqq \Gamma(v,i)$ for $i\in[d],v\in V_L$ denotes the $i^{th}$ neighbor of vertex $v$; is a $(k, \alpha)$-lossless expander for some $k<|V_L|$ and $\alpha \in [0,1]$, if every subset $S$ of $V_L$ of size at most $k$ has the size of its neighborhood $|N_G(S)|\geq \alpha d |S|$.
\end{definition}

$\Gamma_i^{-1}(r)$ denotes the set of all left vertices whose $i^{th}$ neighbor is $r$ for some $i\in [d], r\in [m]$. A lossless expander is said to be \emph{explicit} if $\Gamma_i(v)$ can be computed in $\polylog(|V_L+V_R|)$ time for any $i,v$.

\cite{Cheraghchi11} shows how to construct an explicit lossless expander\footref{cheragchi} stated below.

\begin{lemma}[Corollary 2.23 from \cite{Cheraghchi11}\protect\footnote{\label{cheragchi}The original lemma in \cite{Cheraghchi11} constructs a lossless condenser, however, according to Lemma 2.2.1 of \cite{TUZ01} lossless condenser is equivalent to a lossless expander. The original version also states the lemma over any field of fixed prime power $p$, however we set $p=2$ as it is enough for our purpose.}]
\label{linearExpander}
    Let $\alpha>1/\poly(n)$ be an arbitrary number\footnote{Although the original version of the lemma states $\alpha$ to be constant, lemma still holds for any $\alpha>1/\poly(n)$}. Then, for parameters $N\in \mathbb{N}$, $K \leq N$, and $\eps > 0$, there is an explicit $2^D$-left regular $(2^K, (1-\eps))$-lossless expander $G=(\{0,1\}^N, \{0,1\}^{D+M}, \Gamma)$ with $D\leq (1+1/\alpha)(\log(NK/
    \eps)+O(1))$ and $M \leq 2D+(1+\alpha)K$\footnote{One additive factor of $D$ comes from the condenser construction of \cite{Cheraghchi11}, the other $D$ comes from the equivalence of condenser and expander in \cite{TUZ01}}. Moreover, $\Gamma_i$ is a linear function (over $\bF_2$) for every fixed choice of $i\in \{0,1\}^D$.
\end{lemma}

The neighborhood function of lossless expander $\Gamma$ in \Cref{linearExpander} is constructed based on evaluations of polynomials in $\bF_p$ for prime $p$. We state the following property of linear functions, which is required to prove a lower bound on the size of the sets in the selector family.

\begin{proposition}
    \label{equalRightDegree}
    For some prime number $p$, let $f: \bF_p^a \rightarrow \bF_p^b$ be a linear function for some integers $a,b$ with $a\geq b$ in the field $\bF_p$. For $r\in \bF_p^b$, let $f^{-1}(r)$ be the set of all vectors in $\bF_p^a$ that map to $r$, then every nonempty set $f^{-1}(r)$, has size at least $p^{a-b}$.
\end{proposition}

This proposition follows from the property that $f(x+y) = f(x)+f(y)$ for any $x,y \in \bF_p^a$ where operations are performed in the field $\bF_p$.

We conclude the following theorem using the \Cref{linearExpander,equalRightDegree}.

\begin{lemma}\label{losslessExpander}
    There exists an algorithm given $n,k, \eps, \alpha$ as inputs with $k\leq O(\frac{n\eps}{\log^2 n})^{1/(1+\alpha)}$ satisfying $n=2^N$ for some $N\in \mathbb{N}$, $\eps\in (0,1)$ and $\alpha \in (1/n, 1)$, constructs a $d$-left regular bi-partite graph $G=([n],[m],E)$ is a $(k,1-\eps)$-lossless expander in $\tO{nd}$ time. We have $d \leq O(\frac{\log n \cdot \log k}{\eps})^{1+\frac{1}{\alpha}}$ and $m\leq d^2\cdot k^{1+\alpha}$.
    Let $\Gamma$ be the neighborhood function of $G$, we have the guarantee that every non-empty set $|\Gamma_i^{-1}(r)|$ has size at least $n/m$ for $i\in[d],r\in[m]$.
\end{lemma}

\begin{proof}
    Let's provide $N,K(=\log k),\eps,\alpha$ as inputs to \Cref{linearExpander},
    the output $G$ is a $(k, (1-\eps))$-lossless expander.

    From \Cref{linearExpander} we have the following bound on the size of the right vertices $m = p^M \leq 2^{2D}\cdot 2^{K(1+\alpha)} = d^2\cdot k^{1+\alpha}$ and the left degree $d = 2^D \leq O(\frac{\log n \log k}{\eps})^{1+\frac{1}{\alpha}}$. As mentioned in \cref{linearExpander}, it takes $\tO{1}$ time to compute each edge.
    Since there are at most $nd$ edges, it takes $\tO{nd}$ time.

    For any fixed $i\in[d]$, from the construction in \Cref{linearExpander} we have $\Gamma_i$ is a linear function that maps $\bF^N_2$ to $\bF^M_2$. Thus, for every right vertex $r \in [m],i\in[d]$ such that $\Gamma_i^{-1}(r)\neq \emptyset$ has size at least $2^{N-M} = n/m$ according to \Cref{equalRightDegree}.
\end{proof}

We define another type of expander called unique neighbor expander with guarantees that follow from the lossless expander and are enough for proving \Cref{thm:selector}.

\begin{definition}[$(k,\alpha)$-unique neighbor expander]
    A $d$-left regular bipartite graph $G=(V_L, V_R, E)$ is a $(k, \alpha)$-unique neighbor expander for some $k<|V_L|$ and $\alpha \in [0,1]$, if every subset $S$ of $V_L$ of size at most $k$ has at least $\alpha d|S|$ many vertices in its neighborhood $N_G(S)$ which have unique neighbor in $S$.
\end{definition}

\begin{proposition}\label{expanderIsUniqueNeighborExp}
    A bipartite graph $G = (V_L\cup V_R, E)$ that is a $(k, (1-\eps))$-lossless expander is also a $(k, (1-2\eps))$-unique neighbor expander.
\end{proposition}

This proposition follows from the definitions. We are now ready to prove \Cref{thm:selector}.

\begin{proof}[Proof of \Cref{thm:selector}]
    Given integers $n,k$ as inputs let $n'\in [n, 2n]$ be a power of $2$. We duplicate arbitrary $n'-n$ many elements from $[n]$ and add them to $[n]$ so the size is exactly $n'$. Let $\Gamma:[n']\times [d]\rightarrow [m]$ be the neighborhood function of lossless expander $G = (V_L\cup V_R, E)$ obtained by giving $n',k'=4k,\eps' = \eps/8, \alpha$ as inputs to \Cref{losslessExpander}. We have $|V_L|=n'$ and let $|V_R|=m$ and $d$ is the left-degree of $G$. For each $i\in[d]$, $x\in \bF_2^{\log n'}$, let $\Gamma_i(x) \coloneqq \Gamma(x, i)$. We construct the selector family as follows

    \[\cS \coloneqq \{\Gamma_i^{-1}(j) : i\in [d], j\in [m]\}.\]

    From \Cref{losslessExpander} we know that $d\leq O(\frac{\log n' \log k'}{\eps'})^{1+\frac{1}{\alpha}} \leq O(\frac{\log n \log k}{\eps})^{1+\frac{1}{\alpha}}$ and $m \leq d^2(k')^{1+\alpha}\leq O(d^2k^{1+\alpha})$. So the size of the family is at most $dm\leq O(k^{1+\alpha}d^3) \leq O(k^{1+\alpha}\cdot (\frac{\log n\log k}{\eps})^{3+3/\alpha})$ and it takes $\tO{nd}\leq \tO{n}\cdot (\frac{\log n\log k}{\eps})^{1+1/\alpha}$. We also have that each non-empty set $\Gamma^{-1}_i(j)$ has size at least $n/m \geq O(n/(k^{1+\alpha}\cdot (\frac{\log n\log k}{\eps})^{2+2/\alpha}))$ based on the bound of $m$.

    From \Cref{expanderIsUniqueNeighborExp} we have the guarantee that $G$ is also a $(k',(1-2\eps'))$-unique neighbor expander. For any set $A\subset V_L$ of size at most $k'$, we have the guarantee that there are at least $(1-2\eps')d|A|$ vertices in $V_R$ that have unique neighbors in $A$. There are at least $(1-2\eps')$ fraction of vertices in $A$ that are unique neighbors to some vertex in $V_R$. Otherwise, the number of vertices in the neighborhood of $A$ that have unique neighbors is strictly less than $(1-2\eps')d|A|$, as each left vertex can have at most $d$ neighbors, leading to contradiction.

    Let $v\in A$ be the vertex that is the unique neighbor of some vertex $r$ in $V_R$. There exists an $i\in[d]$ such that $\Gamma(v,i) = r$. Then we have the guarantee that $A\cap \Gamma_i^{-1}(r) = \{v\}$. Hence, at least $(1-2\eps')$ fraction of vertices in $A$ can be \emph{selected} by some set in $\cS$.

    We now prove that the $\cS$ is a $(n,k,\eps)$-selector. Let $L, S\subset [n]$ be two disjoint sets of sizes $\eps k < |L| \leq k$ and $S\leq k$. Let $A\subset V_L$ be the set of vertices corresponding to $L, S$ along with their duplicates if they exist. We consider that an element is selected even if its duplicate is selected. We have $|A|\leq 2|L\cup S| \leq 4k = k'$. Since less than $2\eps'|A|\leq \eps k$ cannot be selected by $\cS$ and the size of $L$ is greater than $\eps k$, there exists at least one vertex in $L$ that can be selected by $\cS$, avoiding all other vertices in $A = L\cup S$. Hence, $\cS$ is $(n,k,\eps)$-selector family.
\end{proof}

\section*{Acknowledgements}
We thank Mahdi Cheraghchi for the helpful discussion and references on lossless expanders.

\bibliographystyle{alpha}
\bibliography{ref_new}

\appendix

\section{Exposition of Gabow's Algorithm}
\label{sec:gabow alg}

Through this section, we aim to provide an exposition of Gabow's algorithm for finding the minimum vertex cut in light of recent developments in maxflow literature. We prove the following theorem.

\begin{theorem}\label{thm:gabowMain}
    There exists an algorithm that takes an undirected graph $G = (V, E)$ and a cut parameter $k$ as inputs and outputs either a vertex cut of size less than $k$ or guarantees that the graph is $k$-vertex connected in $\tilde{O}(n + k\sqrt{n})$ maxflows on graphs of size at most $m$ where $|V| = n, |E| = m$.
\end{theorem}

The main idea of Gabow's algorithm is to find a pair of vertices $(s, t)$ such that $s \in L$, $t \in R$ so that applying a unit-vertex-capacity maxflow across the pair would return the separator $S$ of size less than $k$. Gabow uses the Ramanujan expander to find such a pair.

\subsection{Ramanujan Expander}

We start with defining the Ramanujan expander. Let $A$ denote the adjacency matrix of the graph $G$, and $D$ represent the diagonal matrix with the $i$-th diagonal entry as $\deg(v_i)$, where $v_i$ is the $i$-th vertex of the graph $G$. The eigenvalues of the matrix $AD^{-1}$ are denoted by $\lambda_1(G), \lambda_2(G), \dots \lambda_n(G)$, sorted from large to small. Notably, $AD^{-1}$ is often called the column-wise normalized adjacency matrix or the random walk matrix. Alon and Boppana famously showed that $\lambda_2(G) \geq 2\frac{\sqrt{d-1}}{d} - o_n(1)$ for any $d$-regular graph $G$ of $n$ vertices. The graphs that (almost) achieve this lower bound are called the Ramanujan expanders.

\begin{definition}[Ramanujan Expander]
    \label{lem:RamanujanEigen}
    A $d$-regular graph $X=(V_X,E_X)$ is a Ramanujan expander if $\lambda_2(X) \leq 2\sqrt{d-1}/d$.
\end{definition}

We use the following explicit construction of the Ramanujan graphs from \cite{Lubotzky86,Margulis88}.

\begin{lemma}[Explicit construction of Ramanujan Expander]\label{lem:explicit ramanujan}
    Given integers $n$ and $d\in \mathbb{N}$, there exists a $d'$-regular Ramanujan expander $X=(V_X,E_X)$ with $|V_X|\leq \cGabow n$ for some universal constant $\cGabow$, where $d\leq d'\leq 4d$. This expander can be constructed in $O(nd)$ time.
\end{lemma}

The critical property of the Ramanujan graph that is exploited in Gabow's algorithm is that any two large disjoint subsets of the graph are connected by an edge across them. This is implied by the \emph{expander mixing lemma} (see, e.g.,~\cite{Salil2012}) stated below.

\begin{lemma}[Expander mixing lemma]
    Let $G=(V,E)$ be any $d$-regular graph. For all $S, T\subseteq V$ we have,
    \[\Big| |E(S,T)| - \frac{|S|\cdot |T|\cdot d}{|V|}\Big| \leq \lambda_2(G)\cdot d \cdot \sqrt{|S|\cdot|T|}\]
\end{lemma}

Intuitively, the expander mixing lemma states that for any $d$-regular graph with small $\lambda_2$ value behaves close to a random graph, i.e., the number of edges between any two subsets ($S,T$) is close to the ``expected" number of edges between them in a random graph with $p = d/n$. The expander mixing lemma implies the following corollary for Ramanujan expanders whose proof is left as an exercise to the readers.

\begin{corollary}\label{cor:large sets are connected}
    Let $X = (V_X,E_X)$ be a  Ramanujan expander. For any $L,R\subseteq V_X$, if $|L|\cdot |R|\cdot d' \geq 4|V_X|^2$, then $E_X(L,R)\neq \emptyset$.
\end{corollary}

The technical issue of \Cref{lem:explicit ramanujan} is that, given a number $n$, it does not give a graph with precisely $n$ vertices. Below, we show that this problem can be easily fixed by taking an induced subgraph of the Ramanujan graph of the required size $n$. Since any two large sets in the sub-graph also belong to the original Ramanujan graph and the edge between them is retained as we take the induced sub-graph.

\begin{restatable}[Two large sets are connected]{lemma}{expanderMixing}
\label{lem:expander mixing}
    For an integer $n$ and a degree parameter $d$, we can construct a graph $H=(V_H,E_H)$ with $n$ vertices and a maximum degree of $4d$ in $O(nd)$ time such that, for any $A,B\subseteq V_H$, if $|A|\cdot |B| \cdot d \geq 4\cGabow^2n^2$, then $E_H(A,B)\neq \emptyset$ for some universal constant $\cGabow$.
\end{restatable}

\begin{proof}[Proof of \Cref{lem:expander mixing}]
    Given integers $n,d$ apply \Cref{lem:explicit ramanujan} and construct $d'$-regular Ramanujan expander $X=(V_X,E_X)$ in time $O(nd)$. We construct $H=(V_H,E_H)$ from $X$ as follows. Let $V_H$ be any arbitrary subset of $V_X$ of size $n$ and $E_H=E_X(V_H, V_H)$ which is the subset of edges $E_X$ whose both endpoints are in $V_H$. Hence, it takes at most $O(nd)$ total time to construct the graph $H$. Maximum degree of $H$ is at most $d'\leq 4d$ from \Cref{lem:explicit ramanujan}.

    For any $L,R\subseteq V_H \subseteq V_X$, if we have \[|L|\cdot |R| \cdot d' \geq |L|\cdot |R| \cdot d \geq 4\cGabow^2n^2 \geq 4|V_X|^2\] as $d'\geq d$, $|L||R|d\geq 4\cGabow^2n^2$ and $|V_X| \leq \cGabow n$ respectively. From \Cref{cor:large sets are connected} we have guarantee that $E_X(L,R) \neq \emptyset$, hence $E_H(L,R)\neq \emptyset$.
\end{proof}

\subsection{Overview of Gabow's algorithm}

For any mincut $(L, S, R)$ the key observation is that both $L$ and $R$ contain at least $\delta_G - \kappa_G$ many vertices, where $\delta_G$ is the minimum degree of a vertex in the graph $G$, and $\kappa_G$ is the vertex connectivity of the graph $G$. We define the difference between $\delta_G - \kappa_G$ as the \emph{gap} $\gamma_G$ of the graph $G$. We omit the subscript $G$ when it is clear from the given context.

In the following lemma, we demonstrate how to leverage the gap and find a cut smaller than $k$ for any given parameter $k$.

\begin{lemma}\label{lem:largeGap}
    Given an undirected graph $G = (V, E)$ with a gap $\gamma$ and a cut parameter $k$, there exists an algorithm that either outputs a cut of size less than $k$ or declares that the graph is $k$-connected in $O(n + \frac{\delta^2}{\gamma})$ maxflows on graphs of size at most $|E|$.
\end{lemma}

From the above lemma, it is clear that the graph should have a large gap $\gamma$ for a smaller runtime. We can modify the graph $G$ to construct a new graph $G'$ with a large gap. We then use \Cref{lem:largeGap} on $G'$ to find a cut of size less than $k'$, say $S'$, and then construct the separator $S$ in $G$ from the separator $S'$. Since any general graph need not have a large gap, we create a new one by modifying the original one, as proven in \Cref{lem:increaseGap}.

\begin{lemma}\label{lem:increaseGap}
    Given an undirected graph $G = (V, E)$, there exists an algorithm that constructs a graph $H = (V_H, E_H)$ of size at most $|E|$ with a gap $\gamma$, a vertex cut $(L, S, R)$ of the graph $G$ in $O(\gamma n)$ maxflows. If $\kappa_G < k$, we also have the guarantee that $|S| < k$, or $V \setminus V_H$ is contained in every mincut of $G$.
\end{lemma}

Given these two lemmas, let us first prove \Cref{thm:gabowMain}.

\begin{proof}[Proof of \Cref{thm:gabowMain}]
Using \cite{NagamochiI92}, we can reduce the number of edges in the graph to $O(nk)$ without affecting the vertex mincut if a mincut of size less than $k$ exists. Thus, $\delta = O(k)$.

If $k < \sqrt{n}$, then $\delta = O(\sqrt{n})$, and we can directly apply \Cref{lem:largeGap}. The number of maxflows is $O(n + \delta^2/\gamma) = O(n)$ and is dominated by $O(n + k\sqrt{n})$, so we are done. Hence, $k \geq \sqrt{n}$.

We can assume $\delta \geq \sqrt{n}$ as otherwise $\delta < \sqrt{n} \leq k$, and we return the minimum degree cut. Given a graph $G$, we use \Cref{lem:increaseGap} to increase the gap to $\delta/\sqrt{n}$, which outputs $H$, a cut $(L, S, R)$ and takes $O(\delta\sqrt{n})$ maxflows. If $|S| < k$, then we are done.

Otherwise, we apply \Cref{lem:largeGap} with $H$ and $k' = k - |V \setminus V_H|$ as inputs to either find a cut $(L', S', R')$ of $H$ of size less than $k'$ or declare that $\kappa_H \geq k'$. If $|S'| < k'$, then $(L', S' \cup (V \setminus V_H), R')$ is a valid vertex cut of $G$ of size less than $k$. If $\kappa_{H} \geq k'$ implies $\kappa_G \geq k$, because if $\kappa_G < k$, from \Cref{lem:increaseGap}, we have the guarantee that $|V \setminus V_H|$ is contained in every mincut of $G$, which implies $\kappa_H < k'$, and we arrive at a contradiction.

This takes $O(n + \delta^2/\gamma) \leq O(n + \delta\sqrt{n})$ maxflows as $\gamma \geq \delta/\sqrt{n}$. Note that the number of edges in $H$ is at most that of $G$. Hence, the runtime is at most $O(n + \delta\sqrt{n})$ maxflows on graphs of size at most $G$. Since $\delta = O(k)$, the runtime is at most $O(n + k\sqrt{n})$ maxflows.
\end{proof}

We now prove \Cref{lem:largeGap,lem:increaseGap}, starting with \Cref{lem:largeGap}.

\subsection{Exploiting a Large Gap}
The critical idea in proving \Cref{lem:largeGap} is to use Ramanujan expanders to find a pair of vertices $(s, t)$ that lies across a vertex mincut $(L, S, R)$, i.e., $s \in L$, $t \in R$. Below, we recall the lemma regarding Ramanujan expanders in \Cref{sec:pseudo random objects}.

Let $H$ be the Ramanujan expander constructed with $|V_H|, d$ as inputs to the above lemma for any subset $V_H$ of $V$. From the above lemma, if $L, R$ of a mincut $(L, S, R)$ of $G$ have a large intersection with $V_H$, i.e., $A = L \cap V_H, B = R \cap V_H$, then the degree $d$ needed to satisfy the condition in \Cref{lem:expander mixing} is small, thus optimizing the size of the Ramanujan graph. Thus, we need to find subsets $V_H$ with large $|A|, |B|$, which brings us to the following definition.

\begin{definition}[$\rho$-rich]
    A set $T \subset V$ is said to be $\rho$-rich in a graph $G$ for some integer parameter $\rho$ if there exists a mincut $(L, S, R)$ of $G$ such that $|T \cap L| \geq \rho$ and $|T \cap R| \geq \rho$.
\end{definition}

Using \Cref{lem:expander mixing}, we conclude the following lemma from \cite{Gabow06}, rewritten here for completeness.

\begin{lemma}[Lemma 2.4 from \cite{Gabow06}]\label{lem:kappausingramanujan}
    Take any $k \geq \kappa_G$ and let $T$ be a $\rho$-rich set of size $t$. Let $X_{t,d} = (V_X, E_X)$ be the Ramanujan expander constructed with $t, d = 1+\frac{4 (\cGabow t)^2}{\rho \rho'}$ as inputs to \Cref{lem:expander mixing}, where $\rho' = \max(\rho, \frac{t-k}{2})$. Then
    \[\kappa_G = \min\{\kappa_G(x, y): (x, y) \in E_X\}\]
\end{lemma}

\begin{proof}
    Since $T$ is $\rho$-rich, let $(L, S, R)$ be a vertex mincut such that $\alpha = |L \cap T| \geq \rho$ and $\beta = |R \cap T| \geq \rho$. Since $(L, S, R)$ is a partition, we have $\alpha + \beta + |S| = t$. Without loss of generality, let $\alpha \geq \beta$, which implies $\alpha$ has a size of at least $\frac{t-|S|}{2} \geq \frac{t-k}{2}$. Thus, $\alpha \beta d \geq \rho' \rho d \geq 4 (\cGabow t)^2$, which implies there exists an edge across $E_X(L \cap T, R \cap T) \neq \emptyset$ according to \Cref{lem:expander mixing}. Hence, the result follows.
\end{proof}

We now see how to find a very rich set (high $\rho$) but small in size. Below, we show how to construct such a set. We have the following lemma for any vertex $a \in V$.

\begin{lemma}\label{lem:neighborhoodistaurich}
    For any vertex $a \in V$, if $\kappa(a) \geq k$, then for any vertex mincut $(L,S,R)$, $|N_G(a) \cap L|, |N_G(a) \cap R| \geq \min(\delta, \kappa(a)) - \kappa + 1$.
\end{lemma}

\begin{proof}
    Let $(A, C, B)$ be the vertex cut corresponding to $\kappa(a)$. Assume, for contradiction, $|N_G(a) \cap L| < \min(\kappa(a), \delta) - \kappa + 1$. Let $l \in L$ be any vertex; we have $|L| + |S| \geq \delta + 1$, implying $|L| \geq \delta - \kappa + 1$. Hence, $L \setminus N_G(a) \neq \emptyset$. Thus, $(L \setminus N_G(a), (N_G(a) \cap L) \cup S \setminus \{a\}, \{a\} \cup R)$ separates $a$ from the rest of the graph with a separator of size at most $|N_G(a) \cap L| + |S| - 1 < \min(\kappa(a), \delta)$, which is a contradiction.
\end{proof}

According to \Cref{lem:neighborhoodistaurich}, if the graph has a large gap, i.e., $\delta - \kappa \geq \tau$, and $\kappa(a) \geq k + \tau \geq \kappa + \tau$, then the neighborhood set of $a$ is $\tau$-rich.

We define $\kappa_W(T)$ for a set $T \subset V$ as the minimum number of vertices we need to remove, including the vertices in $T$, to disconnect the set $T$ from the rest of the graph. It can be computed by computing the rooted vertex connectivity of a new vertex $t$ in a graph $G_T$ constructed by joining the vertices of $T$ to $t$ with an edge.

Let $(A, C, B)$ be the vertex cut corresponding to $\kappa(a)$. Since we have $(A, C, B)$ as a vertex cut, the idea is to either use them to construct a set whose weak separation gives us the mincut or return a "small" set that is rich. We first note the following property of a mincut and how it intersects with a rooted mincut for some vertex $a$.

\begin{lemma}\label{lem:cutscross}
    A mincut $(L, S, R)$ exists such that one of the following is true.
    \begin{itemize}
        \item $A$ contains $L$ or $R$
        \item \emph{(or)} $C$ contains $L$ or $R$
        \item \emph{(or)} $A$ or $B$ is contained in $S$, and $S$ is $\min(|A|, |B|)$-rich.
    \end{itemize}
\end{lemma}

\begin{proof}
    Assume $A$ and $C$ do not contain $L$ or $R$. Hence, $L, R$ span across $A, B, C$. $L\cap C, R\cap C \neq \emptyset$; otherwise, $L\cap A, L\cap B$ does not have an edge across them, and $(L\cap A, S, R\cup (L\cap B))$ is a vertex mincut that satisfies the first condition.

    Since $L$ spans across $A, B, C$ and $L\cap C \neq \emptyset$, we have either $L\cap A \neq \emptyset$ or $L\cap B \neq \emptyset$. Similarly, with $R$, we have $R\cap A \neq \emptyset$ or $R\cap B \neq \emptyset$.

    If $L\cap A \neq \emptyset$ and $R\cap B\neq \emptyset$, then we have
    \[|S\cap B|+|S\cap C| + |R\cap C| = |N_G(R\cap B)| \geq \kappa(a) = |C|\]
    \[|S\cap B| \geq |C\cap L|\]
    We now consider $N_G(L\cap A)$,
    \[|N_G(L\cap A)| = |S\cap A|+|S\cap C|+|L\cap C| \leq |S| = \kappa\]
    Thus, $(L\cap A, S, V\setminus((L\cap A) \cup S))$ is a mincut that satisfies the first condition. Hence, at least one of $L\cap A, R\cap B$ is empty. Similarly, at least one of $L\cap B, R\cap A = \emptyset$.

    Thus, if we have $L\cap A = \emptyset$, then $R\cap B \neq \emptyset$. Since $L\cap A=\emptyset$, we have $L\cap B\neq \emptyset$ as otherwise $C$ contains $L$. Since $L\cap B\neq \emptyset$, we have $R\cap A = \emptyset$, thus $A \subset S$.

    Any subset of $L, R$ has its neighborhood set of size at least $\kappa$ as it cannot be smaller. We have $|N_G(L\cap B)| = |L\cap C|+|S\cap C| + |S\cap B| \geq |S|$ which implies $|L\cap C| \geq |S\cap A| = |A|$. Similarly, $|R\cap C| \geq |A|$.

    If we start with assuming $L\cap A\neq \emptyset$, then we will have $|L\cap C|, |R\cap C|\geq |B|$. Thus, $C$ is $\min(|A|, |B|)$-rich.
\end{proof}

Given the above lemma, we are ready to define the algorithm. The algorithm takes a graph with a large gap, and either returns a $\tau$-rich set or a cut of size less than $k$. The critical steps of the algorithm are as follows:

\begin{algorithm}[H]
    \DontPrintSemicolon
    \KwIn{A graph $G = (V,E)$ and connectivity parameter $k$.}
    \KwOut{Either a vertex cut of size less than $k$ or a $\tau$-rich set $T$.}
    \BlankLine
    \If{$\delta > n/2$}{\Return{$V$}}
    Let $a$ be a vertex of degree $\delta$.\\
    Let $(A,C,B)$ be the minimum rooted vertex separator for $a$.\\
    \uIf{$\kappa(a)< k$}
    {\Return{$(A,C,B)$}}
    \uElseIf{$\kappa(a) \geq k+\tau$}{\Return{$N_G(a)$}}
    \Else
    {
        Let $T\subset A\cup B$ and $|T|=k$\\
        \uIf{$\kappa_W(C)< k$ or $\kappa_W(T)< k$}{\Return{Corresponding vertex cut}}
        \Else{\Return{$C$}}
    }
    \caption{\textsc{kCutOrRichSet}$(G,k)$}
    \label{alg:gabow}
\end{algorithm}

\begin{lemma}
    There exists an algorithm that takes a graph $G$ and a connectivity parameter $k$ as inputs and finds either a separator of size less than $k$ or a $\tau$-rich set of size at most $2\delta$ for $\tau = \frac{\delta-k}{2}$ in at most $O(n)$ maxflows.
\end{lemma}

\begin{proof}
    We know from the previous discussion that the vertex set $V$ is $(\delta-\kappa)$-rich, which implies it is also $(\delta-k)/2$-rich. If the minimum degree is very high, i.e., $\delta > n/2$, we obtain a $\tau$-rich set of size $n \leq 2\delta$.

    As shown in \Cref{alg:gabow}, we compute the rooted vertex connectivity of a vertex $a$ with degree $\delta$, which gives the vertex cut $(A, C, B)$ where $C$ is a separator of size $\kappa(a)$. Note that any subset $\hat{B} \subset B$ has $|N_G(\hat{B})| \geq \kappa(a)$.

    If $\kappa(a) < k$, we return $(A, C, B)$ and we are done. Otherwise, if $\kappa(a) \geq k$, and it is significantly larger than $k$ (i.e., $\kappa(a) \geq k+\tau$), then, according to \Cref{lem:neighborhoodistaurich}, $N_G(a)$ is $\tau$-rich, which is of size $\delta$.

    If $\kappa(a) < k+\tau$ and we have $\delta = k+2\tau$, then there exist at least $k$ vertices in $A \cup B$, and we construct a set $T$ of size $k < 2\delta$. Note that we can increase the size of a $\rho$-rich set to $2\delta$ by adding arbitrary vertices while still keeping it $\tau$-rich.

    Let $(L, S, R)$ be a mincut of size less than $k$, and from \Cref{lem:cutscross}, we have the following cases:
    \begin{itemize}
        \item If $A$ contains $L$ or $R$, in which case $C$ spans across only $S, R$ or $L, S$, the weak separator of $C$ gives the mincut, i.e., $\kappa = \kappa_W(C)$.
        \item If $C$ contains $L$ or $R$, in which case $A \cup B$ spans across only $S, R$ or $L, S$, the weak separator of $A \cup B$ gives the mincut. Since we are interested only in a cut smaller than $k$, we take $k$ vertices from $A \cup B$ as the set $T$ and find its weak separator $\kappa_W(T)$.
        \item Otherwise, $C$ is $\min(|A|,|B|)$-rich. Since $\delta = k+2\tau$ and $\kappa(a) < k+\tau$, we have $\min(|A|,|B|) \geq \tau$. Hence, $C$ is a $\tau$-rich set, and $|C| < k+\tau = \frac{\delta+k}{2} < 2\delta$.
    \end{itemize}

    Thus, \Cref{alg:gabow} either returns a cut of size less than $k$ or a $\tau$-rich set. We first compute the rooted vertex connectivity for a vertex, which takes $n$ maxflows. Then, based on the $\kappa(a)$ value, we run at most two more weak separator calls, each taking $O(n)$ maxflows.
\end{proof}

\begin{proof}[Proof of \Cref{lem:largeGap}]
    Given a graph $G$ with a gap $\gamma$, we apply \Cref{alg:gabow} to either find a cut of size less than $k$ or a $\tau$-rich set $T$ of size $t \leq 2\delta$ in at most $O(n)$ maxflows. We then apply \Cref{lem:kappausingramanujan} with inputs $T$ and $k$. The algorithm guarantees that applying $(s, t)$-maxflows on edges of $X_{t, d}$ returns mincuts, which are at most
    \begin{align*}
        |E_X| &\leq t \cdot 4d && \text{From \Cref{lem:expander mixing}} \\
        &\leq 2\delta \cdot \frac{O(\delta^2)}{\tau \max(\tau, \frac{t-k}{2})} && \text{From \Cref{lem:kappausingramanujan}} \\
        &\leq O\left(\frac{\delta^2}{\gamma}\right) && \tau = \frac{\gamma}{2}, \max(\tau, \frac{t-k}{2}) = \frac{t-k}{2} \geq \frac{\delta}{2} \\
    \end{align*}
    Hence, the total number of maxflows is at most $O(n + \frac{\delta^2}{\gamma})$ on graphs of size at most $|E|$. If the mincut is less than $k$, according to \Cref{lem:cutscross}, we either get a vertex cut of size less than $k$ or a $\tau$-rich set. According to \Cref{lem:kappausingramanujan}, we find the mincut if a $\tau$-rich set is returned. Thus, we find a cut of size less than $k$; otherwise, we conclude that the graph is $k$-connected.
\end{proof}

\subsection{Gap Enlargement}

We now see how to increase the gap and construct a graph with a large gap. We first do the following preliminary observations.

\begin{lemma}[Lemma 3.1 from \cite{Gabow06}]
    \label{lem:mincutProps}
    Let $G=(V, E)$ be a graph.
    \begin{enumerate}
        \item \label{lem:mincutProps1} For any vertex $x\in V$ with $\kappa(x)> \kappa_G$, $x$ belongs to the vertex mincut of $G$.
        \item \label{lem:mincutProps2} For any pair $(x,y)\in V\times V$ with $\kappa(x,y)> \kappa_G$, $E\cup \{(x,y)\}$ does not destroy any vertex mincut in $G$.
    \end{enumerate}
\end{lemma}

\begin{proof}
    \begin{itemize}
        \item If $\kappa(x) = \kappa_G$, then there exists a mincut that separates $x$ from the rest of the graph $G$. Since $\kappa(x) > \kappa_G$, every mincut of $G$ contains $x$.
        \item For any pair $(x,y)$, if $\kappa(x,y) > \kappa_G$, then there exists no mincut such that $x$ and $y$ lie on both sides of the mincut. Thus, adding an edge does not destroy any mincut of the graph $G$.
    \end{itemize}
\end{proof}

These two important observations play a key role in constructing the graph $H$ from $G$ that has a large gap. Based on \Cref{lem:mincutProps1}, we have the guarantee that the vertex $x$ belongs to the mincut of the graph $G$ when we assume the mincut is of size less than $k$. Let $(L, S, R)$ be any mincut of the graph $G$. We can remove $x$ from the graph $G \setminus \{x\}$, and the set $S\setminus \{x\}$ remains the mincut of the residual graph ``$G-x$''. This modification reduces $\delta_{G}$ by $1$. However, removing $x$ from $G$ might also reduce $\delta_G$ by $1$.

Let $F$ be the set of vertices whose degree is $\delta_G - 1$ in the residual graph. For any $u\in F$, we can increase the degree back to $\delta_G$ by adding an edge with one of its non-neighbors, say $v$, after confirming that they always lie on the same side of every mincut by computing $\kappa(u,v)$. If $\kappa(u,v)\geq k$ and we have $\kappa_G < k$, then we know according to \Cref{lem:mincutProps2} that adding edge $(u,v)$ does not destroy any mincut, thereby increasing the degree of $u$ to $\delta_G$. Thus, it takes $|F|$ maxflows to increase all the degrees in $F$.

We formally define the algorithm below that increases the gap by $1$ and prove its correctness and runtime.

\begin{algorithm}[H]
    \SetAlgoLined
    \KwIn{A graph $H = (V_H,E_H)$ with gap $\gamma(H)$, a cut $(L, S, R)$}
    \KwOut{A graph $H' = (V_{H'},E_{H'})$ with gap $\gamma(H)+1$, an updated cut $(L, S, R)$}
    Choose any $x\in V_H$. Compute $\kappa_H(x)$ and the corresponding cut $(X, Z, Y)$\\
    \If{$|Z\cup (V\setminus V_H)|< |S|$}
    {
        $(L, S, R) \gets (X, Z\cup (V\setminus V_H), Y)$
    }
    Let $H' = H\setminus x$ removing $x$ from the graph $H$\\
    Let $F=\{y\in N_H(x)\mid\deg_{H'}(y)=\delta(H)-1\}$\\
    \For{$u\in F$}{
    Iterate over $V_{H'}$ to find $w$ such that $(u,w)\notin E_{H'}$\\
    Compute $\kappa_{H'}(u,w)$ and the corresponding cut $(U, Q, W)$\\
    \If{$|Q\cup (V\setminus V_{H'})| < |S|$}
    {
        $(L, S, R) \gets (U, Q\cup (V\setminus V_{H'}), W)$
    }
    $H' = H'\cup (u,w)$ adding the edge $(u,w)$ to $H'$\\
    }
    \Return{$H', (L, S, R)$}
    \caption{$IncreaseGap$}
    \label{alg:increaseGap}
\end{algorithm}

\begin{claim}\label{lem:increaseGap1}
    The \Cref{alg:increaseGap} takes a graph $H$ with gap $\gamma(H)$ as input and outputs a graph $H'$ with gap $\gamma(H) + 1$, both of size at most $H$, in $O(|V_H|)$ maxflows on graphs of size at most $m$, where $m$ is the size of the original graph. It also maintains the smallest vertex cut in the original graph $G$.
\end{claim}

\begin{proof}
    We maintain the invariant that $(L, S, R)$ is a valid vertex cut in the original graph $G$ and update it whenever we find a smaller one. Assuming $\kappa_G < k$ and conditioned on not finding a cut of size less than $k$ so far, we have the guarantee that $V \setminus V_H$ is contained in every mincut.

    Hence, $(X, Z \cup (V \setminus V_H), Y)$ is a valid cut in $G$. Therefore, we update $(L, S, R)$ with the newly found cut. Conditioned on not finding the mincut, we have the guarantee that $x$ belongs to every mincut in $G$ according to \Cref{lem:mincutProps1}. Thus, we can safely remove it from the graph $H$, constructing a new graph $H'$. We have $\kappa(H') = \kappa(H) - 1$.

    However, $\delta(H')$ can also decrease by $1$, keeping the gap the same. Now, we increase the minimum degree of $H'$ by adding edges to the vertices with degree $\delta(H) - 1$, checking if $(s, t)$-vertex mincut across the two vertices in $H'$ occurs after extending the cut to the graph $G$ by taking the union with $V \setminus V_{H'}$. Thus, we increase the minimum degree by $1$, which increases $\gamma(H')$ by one. The number of edges in $H'$ is, at most, the number of edges in $H$, as we only add edges to a vertex if an edge is deleted from it due to the removal of $x$.

    Computing $\kappa(x)$ takes $n$ maxflows. It takes a maxflow call to increase the degree of a vertex in $F$. Hence, the total number of maxflow calls is at most $n + |F| \leq 2n = O(n)$.
\end{proof}

We are now ready to prove \Cref{lem:increaseGap} by using \Cref{lem:increaseGap1}.

\begin{proof}[Proof of \Cref{lem:increaseGap}]
    We initialize $H$ with $G = (V, E)$ and $(L, S, R)$ with the minimum degree cut in $G$. Applying the \textsc{increaseGap} subroutine $\gamma$ times constructs a graph $H' = (V_{H'}, E_{H'})$ with gap $\gamma$ of size at most $|E|$, and $(L, S, R)$ is a valid vertex cut in $G$.

    If $|S| < k$, we are done; otherwise, $|S| \geq k$. If $\kappa_G < k$, then from \Cref{lem:increaseGap1}, we have the guarantee that the smallest cut found so far is not the mincut; hence, the vertices removed from $G$, i.e., $T = V \setminus V_{H'}$, belong to every mincut of $G$.
\end{proof}

\section{Omitted Proofs from \texorpdfstring{\Cref{sec:pseudo random objects}}{Pseudorandom Objects section}}
\label{sec:omittedproofs}

\subsection{Proof of \texorpdfstring{\Cref{lem:generalDisperser}}{general disperser}}
Before we give the proof we state and prove two simple claims that help in proving \Cref{lem:generalDisperser}.
The following claim generalizes the \Cref{thm:disperser} to work for any given $n,k$ and not just the powers of $2$.

\begin{claim}
\label{clm:disperseranysize}
    Given two positive integers $n,k$ with $k<n$ and a constant $\eps\in (0,1)$, we can construct an explicit $(k,\eps)$-disperser $D=(V,W,E)$ with $|V|=n$, left-degree $d=\Theta((\log n)^{\cTUZ})$ and $|W|\in [c_{\text{low}}\frac{kd}{\log^3 n},c_{\text{high}}kd]$ where $c_{\text{low}},c_{\text{high}}$ are some universal constants.
\end{claim}

\begin{proof}
    Let $n',k'$ be the powers of $2$ such that $n\leq n' < 2n$ and $k/2 < k' \leq k$. It follows that $k'<n'$. Let $D'=(V',W',E')$ be the explicit $(k',\eps)$-disperser obtained with $n',k',\eps$ as inputs to \Cref{thm:disperser} with left-degree $d=(\log n')^{\cTUZ}$ and $|W|\in [c'_{\text{low}}\frac{k'd}{\log^3 n},c'_{\text{high}}k'd]$.

    Consider any arbitrary subset $V$ of $V'$ of size exactly $n$ and consider the induced subgraph of $D'$ restricted to $V,W'$ and let $E$ be the edge set of the induced graph of $D'$ on $V\cup W'$. Any subset $S$ of size at least $k$ of $V$ (of $V'$ as well), is also a subset of $V'$ of size at least $k'$ as $k'\leq k$, so the neighborhood of $S$ in $D'$ is at least of size $(1-\eps)|W'|$. Since we take the induced subgraph of $V\cup W'$, the neighborhood of $S$ in $D$ is same, thus $D=(V,W=W',E)$ is a $(k,\eps)$-disperser with number of left vertices exactly $n$. Since $\log n \leq \log n' \leq 1+\log n$ the left-degree $d=(\log n')^{\cTUZ}\leq c\cdot (\log n)^{\cTUZ}$ for some constant $c$, $d\geq (\log n)^{\cTUZ}$ and $W$ satisfies the size requirements for some universal constants given by \Cref{thm:disperser}.
\end{proof}

The left-degree $d$ of the disperser is $\Theta((\log n)^{\cTUZ})$ as given in \Cref{thm:disperser}. \Cref{lem:generalDisperser} (recalled below) generalizes \Cref{thm:disperser} by making the degree arbitrarily large as needed. Note that the degree cannot be made smaller than $c(\log n)^{\cTUZ}$ where $c$ is some universal constant.

\generaldisperserlem*

\begin{proof}
    If $d\leq (\log n)^{\cTUZ}$ then let $D=(V,W,E)$ be the explicit disperser obtained by giving $n,k,\eps$ as inputs to \Cref{clm:disperseranysize} and just output $D$. From the guarantees of \Cref{clm:disperseranysize} we have $D$ as an explicit $(k,\eps)$-disperser with $|V|=n$, the left-degree $d=\Theta((\log n)^{\cTUZ})=d'$ as $d\leq (\log n)^{\cTUZ}$ and the bounds of $W$ are satisfied.

    If $d> (\log n)^{\cTUZ}$, let $\gamma = \ceil{d/(\log n)^{(\cTUZ)}}$. Construct a disperser $D=(V,W,E)$ with $\gamma n, \gamma k, \eps$ as inputs to \Cref{clm:disperseranysize}. Let $V'$ be the set of super vertices obtained by contracting arbitrary partition $\mathcal{F}$ of $V$ where each set contained in the partition is of size exactly $\gamma$, we have $|V'|=n$ as $|V|=\gamma n$. For any $u\in V'$, let $S_u\subset V$ be the set in $\mathcal{F}$ that contains $u$.
    \[E'= \{(u,v)\mid u\in V', v\in W \text{ and there exists an edge from some vertex in $S_u$ to $W$ in $E$}\}\]
    For every $u\in V'$ the set of edges in $E'$ are precisely those obtained by taking the union of edges in $E$ that are incident to every vertex $s\in S_u$. We output $D'=(V',W,E')$. The disperser might have multi edges and need not be simple.

    We now prove the required properties of $D'$. For any $u\in V'$ the left-degree in $D'$ is $|S_u|\cdot d''$ where $d''$ is left-degree of $D$. From the guarantee of \Cref{clm:disperseranysize} we have $d''=\Theta((\log (\gamma n))^{\cTUZ}) = \Theta((\log n)^{\cTUZ})$ as $\log(\gamma n) = \Theta(\log n)$ and $|S_u|=\gamma$. Thus the left-degree of $D'$ is $\gamma \cdot \Theta((\log n)^{\cTUZ}) = \Theta(\max(d,(\log n)^{\cTUZ}))$ as required.

    It is an explicit disperser i.e., each edge in $E'$ can be computed in $\tO{1}$. Let there be some ordering among the vertices in $V$, the same ordering applies to vertices in $S_u$ for any $u\in V'$. We can compute $i^{th}$ edge of vertex $u$ for $i\in [\gamma \cdot d'']$ by computing $(i \mod d'')^{th}$ edge of $(i/d'')^{th}$ vertex in $S_u$ of $D$ in $\tO{1}$ time. We have $|V'|=|V|/\gamma = n$ as $|V|=\gamma n$ is guaranteed from \Cref{clm:disperseranysize}.

    We have $|W|\in [c_{\text{low}}\frac{(\gamma k)d''}{\log^3 (\gamma n)},c_{\text{high}}(\gamma k)d''] = [c'_{\text{low}}\frac{kd'}{\log^3 n},c'_{\text{high}}kd']$ as $\gamma d'' = \Theta(d)$ and $d'=d=\max(d,(\log n)^{\cTUZ})$ where $c'_{\text{low}},c'_{\text{high}}$ are some universal constants.

    It is left to prove that $D'$ is a $(k,\eps)$-disperser. Any subset $T$ of $V'$ of size at least $k$ corresponds to a subset of $V$ of size at least $\gamma k$. From the guarantees of \Cref{clm:disperseranysize} and based on the construction of $E'$ the size of the neighborhood of $T$ is at least $(1-\eps)|W|$ as required.
\end{proof}

\end{document}